%% file: main_SIAM.tex
\documentclass[final]{siamltex}

\usepackage{amssymb,amsfonts,amsmath}
\usepackage{graphicx,epsfig,color}
\usepackage{graphics, graphicx}
\usepackage{enumerate}
\usepackage{subfigure}
\usepackage{verbatim}
\usepackage{bm} 

\usepackage{framed}
\usepackage{algorithm}
\usepackage{algorithmic}
\usepackage{amsfonts,amsmath}
\usepackage{graphicx}
\usepackage{subfigure,enumerate,wrapfig,array,color}
\usepackage{graphics}
\usepackage{multirow}

\title{Improved matrix algorithms via the Subsampled Randomized Hadamard Transform}

\author{
Christos Boutsidis\thanks{Mathematical Sciences Department, IBM T. J. Watson Research Center. Email: \texttt{\{cboutsi\}@us.ibm.com}}
\and
Alex Gittens\thanks{Applied and Computational Mathematics Department,
California Institute of Technology. Email: \texttt{\{gittens\}@caltech.edu}}
}

\input{newcmds}

\begin{document}

\maketitle

\begin{keywords}
low-rank approximation, least-squares regression, hadamard transform, sampling, randomized algorithms.
\end{keywords}

\begin{AMS}
15B52, 15A18, 11K45
\end{AMS}

\pagestyle{myheadings}
\thispagestyle{plain}
\markboth{BOUTSIDIS AND GITTENS}{IMPROVED MATRIX ALGORITHMS VIA THE SRHT}

\input{text}

\input{BG12.bbl}

\end{document}

%% file: newcmds.tex
\long\def\symbolfootnote[#1]#2{\begingroup%
\def\thefootnote{\fnsymbol{footnote}}\footnote[#1]{#2}\endgroup}

\newcommand{\Probab}[1]{\mbox{}{\mathbb{P}}\left[#1\right]}

\newcommand{\Expect}[1]{\mbox{}{\mathbb{E}}\left[#1\right]}

\newcommand{\Trace }[1]{\mbox{}{\mathrm{Tr}}\left(#1\right)}

\newcommand{\FNorm }[1]{\mbox{}\|#1\|_\mathrm{F}  }
\newcommand{\FNormS}[1]{\mbox{}\|#1\|_\mathrm{F}^2}

\newcommand{\TNormB }[1]{\mbox{}\left\|#1\right\|_2  }
\newcommand{\TNormBS}[1]{\mbox{}\left\|#1\right\|_2^2}
\newcommand{\TNorm }[1]{\mbox{}\|#1\|_2  }
\newcommand{\TNormS}[1]{\mbox{}\|#1\|_2^2}

\newcommand{\XNorm }[1]{\mbox{}\|#1\|_{\xi}  }
\newcommand{\XNormS}[1]{\mbox{}\|#1\|_{\xi}^2}

\newcommand{\VTNorm }[1]{\mbox{}\left|#1\right|  }
\newcommand{\VTNormS}[1]{\mbox{}\left|#1\right|^2}

\newcommand{\INorm}[1]{\mbox{}\|#1\|_\infty}

\newcommand{\transp}{^{\textsc{T}}}

\newcommand{\mat}[1]{{\ensuremath{\bm{\mathrm{#1}}}}}

\newcommand{\abs }[1]{\left|#1\right|}

\newcommand{\condE}[2]{\ensuremath{\mathbf{E}_{#1}\left[{#2}\right] } }

\newcommand{\lambdamax}[1]{\ensuremath{\lambda_{\mathrm{max}}\left(#1\right)}}
\renewcommand{\vec}[1]{\ensuremath{\bm{#1}}}

\renewcommand{\const}[1]{\ensuremath{\mathrm{#1}}}
\newcommand{\pinv}[1]{ {#1}^\dagger}

\def\rank{\hbox{\rm rank}}
\newcommand{\stablerank}[1]{\operatorname{sr}\left(#1\right)}

\def\b{{\mathbf b}}
\def\e{{\mathbf e}}
\def\expe{{\mathrm e}}

\def\v{{\mathbf v}}

\def\E{{\cl E}}

\def\matA{\mat{A}}
\def\matB{\mat{B}}
\def\matC{\mat{C}}
\def\matD{\mat{D}}

\def\matG{\mat{G}}
\def\matH{\mat{H}}
\def\matI{\mat{I}}
\def\matM{\mat{M}}
\def\matP{\mat{P}}
\def\matQ{\mat{Q}}
\def\matR{\mat{R}}
\def\matS{\mat{S}}

\def\matU{\mat{U}}
\def\matV{\mat{V}}
\def\matW{\mat{W}}
\def\matX{\mat{X}}
\def\matY{\mat{Y}}
\def\matZ{\mat{Z}}
\def\matSig{\mat{\Sigma}}
\def\matTh{\mat{\Theta}}
\def\matOmega{\mat{\Omega}}

\DeclareMathSymbol{\Prob}{\mathbin}{AMSb}{"50}
\newcommand\remove[1]{}

\def\nnz{{ \rm nnz }}

\def\math#1{$#1$}

\def\frac#1#2{{#1\over #2}}

\def\eqan#1{\begin{eqnarray*}
#1
\end{eqnarray*}}

\DeclareMathSymbol{\R}{\mathbin}{AMSb}{"52}

\newcommand{\argmin}{\operatorname*{argmin}}

\def\h{{\mathbf h}}

\def\x{{\mathbf x}}
\def\y{{\mathbf y}}
\def\z{{\mathbf z}}

\def\a{{\mathbf a}}
\def\b{{\mathbf b}}

\def\E#1{{{\mathbf E}\left[#1\right]}}
\def\norm#1{{\|#1\|}}

\def\dotfil{\leaders\hbox to 1.5mm{.}\hfill}
\def\RN#1{\setcounter{rmnum}{#1}\uppercase\expandafter{\romannumeral\value{rmnum}}}
\def\rn#1{\setcounter{rmnum}{#1}\expandafter{\romannumeral\value{rmnum}}}

%% file: text.tex
\begin{abstract}
Several recent randomized linear algebra algorithms rely upon fast dimension reduction methods.
A popular choice is the Subsampled Randomized Hadamard Transform (SRHT). In this article, we
address the efficacy, in the Frobenius and spectral norms, of an SRHT-based low-rank matrix
approximation technique introduced by Woolfe, Liberty, Rohklin, and Tygert. We establish a
slightly better Frobenius norm error bound than currently available, and a much sharper spectral
norm error bound (in the presence of reasonable decay of the singular values). Along the way, we
produce several results on matrix operations with SRHTs (such as approximate matrix multiplication)
that may be of independent interest. Our approach builds upon Tropp's in ``Improved analysis of the Subsampled Randomized Hadamard Transform".
\end{abstract}

\section{Introduction}
\label{sec:introduction}
Numerical linear algebra algorithms are traditionally deterministic.  For example, given a full-rank matrix $\matA \in \R^{m \times m}$
and a vector $\b \in \R^{m}$, Gaussian elimination requires at most $ 2 m^3 / 3$ arithmetic operations to compute a vector $\x \in \R^n$
that satisfies $\matA \x = \b$, while the matrix-matrix multiplication $\matA \matA\transp$ requires at most $(2m-1)m^2$ operations,
assuming that the matrix multiplication exponent equals $3$.
Another important problem is eigenvalue computation: current state-of-the-art solvers compute
all $m$ eigenvalues of $\matA \matA\transp$ in $\const{O}(m^3)$ arithmetic operations.
All these computations are deterministic, i.e ensure that the solution of the underlying problem is returned after the corresponding operation count.

Although these algorithms are numerically stable and run in polynomial time,
$\const{O}(m^3)$  arithmetic operations can be prohibitive for many applications when the size of the matrix is large, e.g. on the order of millions or billions~\cite{MMDS08,Mah10}.
One way to speed up these algorithms is to reduce the size of $\matA$, and then apply standard deterministic procedures to the
resulting matrix.  In more detail, for a matrix $\matOmega \in \R^{m \times r}$ ($ m > r = \const{o}(m) $), let $\matY = \matA \matOmega \in \R^{m \times r}$.
$\matOmega$ is a so-called ``dimension reduction'' matrix and $\matY$ contains as much information of $\matA$ as possible.
Consider for example the matrix-matrix multiplication operation mentioned above. In this setting, one can compute
$\matY \matY\transp$ instead of $\matA \matA\transp$. If $\matOmega$ is chosen carefully, then
$$\matY \matY\transp \approx \matA \matA\transp,$$
and the number of operations needed to compute $\matY \matY\transp$ is at most $\const{o}(m^3)$~\cite{DK01,DKM06a}.

Recent years have produced a large body of research on designing random matrices $\matOmega$ with which many popular problems
in numerical linear algebra (e.g. low-rank matrix approximation~\cite{DKM06b,DKM06c}, least-squares regression~\cite{Sar06,BD09,Cla13}, k-means clustering~\cite{BZD10})
can be solved approximately
in $\const{o}(m^3)$ arithmetic operations. We refer the reader to a recent comprehensive
survey of the topic~\cite{HMT}, which has now emerged as \emph{Randomized Numerical Linear Algebra}.

Some proposed choices for $\matOmega$ include: (i) every entry of $\matOmega$ takes the values $+1,-1$ with equal probability~\cite{CW09,Zou10}; (ii) the entries of $\matOmega$ are
i.i.d. Gaussian random variables with zero mean and unit variance~\cite{HMT};
(iii) the columns of $\matOmega$ are chosen independently from the columns of the $m \times m$ identity matrix
with probabilities that are proportional to the Euclidean length of the columns of $\matA$~\cite{FKV98,DKM06b};
(vi) the columns of $\matOmega$ are chosen independently from the columns of the $m \times m$ identity matrix uniformly at random~\cite{Git12};
(v) $\matOmega$ is designed carefully such that $\matA \matOmega$ can be computed in at most $\const{O}(\nnz(\matA)) $ arithmetic operations,
where $\nnz(\matA)$ denotes the number of non-zero entries in $\matA$~\cite{CW12}.

In this article we focus on the so-called Subsampled Randomized Hadamard Transform (SRHT), i.e. the matrix
$\matOmega$ contains a subset of the columns of a randomized Hadamard matrix (see Definitions~\ref{def:walsh} and~\ref{def:srht} below).
This form of dimension reduction was introduced in~\cite{AC06}.
It is of particular interest because the highly structured nature of $\matOmega$ can be
exploited to reduce the time of computing $\mat{Y} = \mat{A} \mat{\Omega}$ from $\const{O}(m^2r)$ to $\const{O}(m^2\log_2 r)$
(see Lemma~\ref{prop:SRHT-compute-time} below).

\begin{definition}[Normalized Walsh--Hadamard Matrix]
\label{def:walsh}
Fix an integer $n = 2^p$, for $p = 1,2,3, ...$. The (non-normalized) $n \times n$ matrix of the Walsh--Hadamard transform is defined recursively as,
\vspace{-.0751in}
$$ \matH_n = \left[
\begin{array}{cc}
  \matH_{n/2} &  \matH_{n/2} \\
  \matH_{n/2} & -\matH_{n/2}
\end{array}\right],
\qquad \mbox{with} \qquad
\matH_2 = \left[
\begin{array}{cc}
  +1 & +1 \\
  +1 & -1
\end{array}\right].
$$
The $n \times n$ normalized matrix of the Walsh--Hadamard transform is equal to $\matH = n^{-\frac{1}{2}} \matH_n \in \R^{n \times n}.$
\end{definition}

\begin{definition}[Subsampled Randomized Hadamard Transform (SRHT) matrix]
\label{def:srht}

Fix integers $r$ and $n = 2^p$ with $r < n$ and $p = 1,2,3, ...$. An SRHT matrix is an $r \times n$ matrix of the form $$ \matTh = \sqrt{\frac{n}{r}} \cdot \matR \matH  \matD;$$
\begin{itemize}
\item $\matD \in \R^{n \times n}$ is a random diagonal matrix whose entries are independent random signs, i.e. random variables uniformly distributed on $\{\pm 1\}$.
\item $\matH \in \R^{n \times n}$ is a normalized Walsh--Hadamard matrix.
\item $\matR \in \R^{r \times n}$ is a subset or $r$ rows from the $n \times n$ identity matrix, where the rows are  chosen uniformly at random and without replacement.
\end{itemize}
\end{definition}
\begin{lemma} [Fast Matrix-Vector Multiplication, Theorem 2.1 in~\cite{AL08}]
\label{prop:SRHT-compute-time}
Given $\x \in \R^n$ and $r < n$, one can construct $\matTh \in \R^{r \times n}$ and compute $\matTh \x$ in at most $2 n \log_2(r + 1) )$ operations.
\end{lemma}

The purpose of this article is to analyze the theoretical performance of an SRHT-based randomized low-rank approximation algorithm
introduced in~\cite{WLRT07} and analyzed in~\cite{WLRT07,HMT,NDT09}. Our analysis (see Theorem~\ref{thm:quality-of-approximation-guarantee})
provides sharper approximation bounds
than those in~\cite{WLRT07,HMT,NDT09}.


Our study should also be viewed as followup to the work of Drineas et al.~\cite{DMMS11}  and~\cite{RT08,AMT10}
on designing fast approximation algorithms for solving least-squares regression problems.
One of the two algorithms presented in~\cite{DMMS11} employs the SRHT
to quickly reduce the dimension of the least squares problem and then solves the smaller problem with a direct least-squares solver,
while~\cite{RT08,AMT10} use the SRHT to design
a good preconditioner for an iterative method, which is then used to solve the regression problem. The results in this article along with
the work in~\cite{Tro11} have implications in all these studies~\cite{RT08,DMMS11,AMT10}. We discuss these implications in Section~\ref{sec:regression}.

\subsection{Beyond the SRHT}
Finally, notice that the SRHT is defined only when the matrix dimension is a power of two.
An alternative option is to use other structured orthonormal randomized transforms such as the
discrete cosine transform (DCT) or the discrete Hartley transform (DHT)~\cite{WLRT07,NDT09,RT08,AMT10}, whose entries are on the order of $n^{-1/2}.$
All these transforms do not place any restrictions on the size of the matrix. The results of this paper - with minimal effort -
can be extended \emph{unchanged} to encompass these transforms. To see this, notice that Lemma 3.3 in~\cite{Tro11} remains unchanged for all these orthogonal transforms.
Thus Lemma~\ref{lemma:SRHT-preserves-geometry} in our work as well as all other results presented in this article are true for these orthogonal transforms as well.

\subsection{Roadmap} This article is structured as follows. Section~\ref{sec:preliminaries} introduces the notation.
In Section~\ref{sec:lowrank}, we present our main results on the quality of SRHT low-rank approximations and compare them to prior results in the literature. In Section~\ref{sec:regression}, we discuss two approaches to least-squares regression involving SRHT dimensionality-reduction.
Section~\ref{sec:SRHT} first recalls known facts on the application of SRHTs to orthogonal matrices and then presents new results on the application of SRHTs to general matrices and the approximation of matrix multiplication using SRHTs under the Frobenius norm. Section~\ref{sec:proofs} contains the proofs of our two main theorems presented in Sections~\ref{sec:lowrank} and~\ref{sec:regression}.
We conclude the paper with an experimental evaluation of the SRHT low-rank approximation algorithm in Section~\ref{sec:experiments}.

\subsection{Preliminaries}
\label{sec:preliminaries}

We use \math{\matA,\matB,\ldots} to denote real matrices and \math{\a,\b,\ldots} to denote real column vectors.
$\matI_{n}$ is the $n \times n$ identity matrix;  $\bm{0}_{m \times n}$ is the $m \times n$ matrix of zeros; $\bm{e}_i$ is the standard basis (whose dimensionality will be clear from the context). $\matA_{(i)}$ denotes the $i$th row of $\matA$; $\matA^{(j)}$ denotes the $j$th column of $\matA$; $\matA_{ij}$ denotes the $(i,j)$th element of $\matA$.
We use the Frobenius and the spectral norm of a matrix: $ \FNorm{\matA} = \sqrt{\sum_{i,j} \matA_{ij}^2}$ and $\TNorm{\matA} = \max_{\x:\TNorm{\x}=1}\TNorm{\matA \x}$, respectively.
The notation $\XNorm{\matA}$ indicates that an expression holds for both $\xi = 2$ and $\xi = \mathrm{F}$.

A (compact) Singular Value Decomposition (SVD) of the matrix $\matA \in \R^{m \times n}$ with $\rank(\matA) = \rho$ is a decomposition of the form
\begin{eqnarray*}
\label{svdA} \matA
         = \underbrace{\left(\begin{array}{cc}
             \matU_{k} & \matU_{\rho-k}
          \end{array}
    \right)}_{\matU_{\matA} \in \R^{m \times \rho}}
    \underbrace{\left(\begin{array}{cc}
             \matSig_{k} & \\
              & \matSig_{\rho - k}
          \end{array}
    \right)}_{\matSig_\matA \in \R^{\rho \times \rho}}
    \underbrace{\left(\begin{array}{c}
             \matV_{k}\transp\\
             \matV_{\rho-k}\transp
          \end{array}
    \right)}_{\matV_\matA\transp \in \R^{\rho \times n}},
\end{eqnarray*}
where the singular values of $\matA$ are ordered \math{\sigma_1\ge\ldots\sigma_k\geq\sigma_{k+1}\ge\ldots\ge\sigma_\rho > 0}.
Here $k$ is a parameter in the interval $1 \le k \le \rho$ and the above formula corresponds to a partition of the SVD in block form using $k$.
We denote the $i$th singular value of $\matA$ by $\sigma_i\left(\matA\right)$ and sometimes refer to $\sigma_1$ as $\sigma_{\max}$ and $\sigma_\rho$ as $\sigma_{\min}$. The matrices $\matU_k \in \R^{m \times k}$ and $\matU_{\rho-k} \in \R^{m \times (\rho-k)}$ contain the left singular vectors of~$\matA$; similarly, the matrices $\matV_k \in \R^{n \times k}$ and $\matV_{\rho-k} \in \R^{n \times (\rho-k)}$ contain the right singular vectors of~$\matA$. We denote $\matA_k = \matU_k \matSig_{k} \matV_k\transp \in \R^{m \times n}$. $\matA_k$ minimizes $\XNorm{\matA - \matX}$ over all $m \times n$ matrices $\matX$ of rank at most $k$.
$\pinv{\matA} = \matV_\matA \matSig_\matA^{-1} \matU_\matA\transp \in \R^{n \times m}$ denotes the Moore-Penrose pseudo-inverse of $\matA \in \R^{m \times n}.$
Let $\matX \in \R^{m \times n}$ ($n \ge m$) and $\matB=\matX \matX\transp \in \R^{m \times m}$; any matrix that can be written in this form is called a symmetric positive semidefinite (SPSD) matrix. For all $i=1, ...,m$, $\lambda_{i}\left(\matB\right) = \sigma_{i}^2\left(\matX\right)$ denotes the $i$th eigenvalue of $\matB$. We sometimes use $\lambda_{\min}\left(\matB\right)$ and $\lambda_{\max}\left(\matB\right)$ to denote the smallest (nonzero) and largest eigenvalues of $\matB$, respectively.

\section{Low-rank matrix approximation using SRHTs}\label{sec:lowrank}
Using an SRHT matrix (see Definition~\ref{def:srht}),
one can quickly construct a low-rank approximation to a given matrix $\matA.$
Our main result, Theorem~\ref{thm:quality-of-approximation-guarantee} below, provides theoretical guarantees on the spectral and Frobenius norm accuracy of these approximations.
\begin{theorem}
\label{thm:quality-of-approximation-guarantee}
Let $\matA \in \R^{m \times n}$ 
with rank $\rho$ and $n$ is a power of 2. Fix an integer $k$ satisfying $ 2 \leq k < \rho$. Let $0 < \varepsilon < 1/3$
be an accuracy parameter, $0 < \delta < 1$ be a failure probability, and $\const{C}\ge 1$ be any specified constant.
Let
$\matY = \matA \matTh \transp,$
where $\matTh \in \R^{r \times n}$ is an SRHT with $r$ satisfying 
\begin{equation}\label{eqn:r}
 6 \const{C}^2 \varepsilon^{-1} \left[\sqrt{k} + \sqrt{8\ln(n/\delta)} \right]^2 \ln(k/\delta) \leq r \leq n.
\end{equation}
Let $\ell = \min\{m,r\}$.
Furthermore,
let $\matQ \in \R^{m \times \ell}$ satisfy $\matQ\transp \matQ = \matI_{\ell}$ and be such that the column space of $\matY$ is contained in the range of $\matQ$
(e.g. such a $\matQ$ can be computed with the $QR$ factorization of $\matY$ in $\const{O}(m \ell^2 )$ arithmetic operations),
and let $\tilde{\matA}_k = \matQ \matX_{opt} \in \R^{m \times n},$
where $\matX_{opt} $ is computed via the SVD of $ \matQ\transp \matA$ as follows,
$$\matX_{opt} = \argmin_{\matX \in \R^{\ell \times n},\,\, \rank(\matX) \le k}\FNorm{ \matQ\transp \matA - \matX }.$$

Given this setup, with probability at least $1 -  \delta^{\const{C}^2 \ln(k/\delta)/4} - 7\delta$ the following Frobenius norm bounds hold simultaneously:
\begin{align}
\FNorm{ \matA - \matY \pinv{\matY}  \matA }   & \le \left(1 + 22 \varepsilon \right)  \cdot \FNorm{\matA-\matA_k}, \tag{i} \\
\FNorm{\matA - \tilde{\matA}_k}               &  \le \left(1 + 22 \varepsilon \right)  \cdot \FNorm{\matA-\matA_k}, \tag{ii}\\
\FNorm{ \matA_k - \matY \pinv{\matY}  \matA } & \le  \left(  1+ 22 \varepsilon  \right) \cdot \FNorm{\matA-\matA_k}, \tag{iii}\\
\FNorm{ \matA_k -\tilde{\matA}_k }            & \le  \left(  2+ 22 \varepsilon  \right) \cdot \FNorm{\matA-\matA_k}.  \tag{iv}
\end{align}

Similarly, the same setup ensures that with probability at least $1 - 5\delta,$ the following spectral norm bounds hold simultaneously:
\begin{align}
 \TNorm{\matA - \matY \pinv{\matY} \matA} & \leq \left(4 +
 \sqrt{\frac{3 \ln(n/\delta)\ln(\rho/\delta)}{r}} \right) \cdot \TNorm{\matA - \matA_k} +
 \sqrt{\frac{3 \ln(\rho/\delta)}{r}} \cdot \FNorm{\matA - \matA_k}, \tag{v} \\
 \TNorm{\matA -  \tilde{\matA}_k  } & \leq \left(6 +
 \sqrt{\frac{6 \ln(n/\delta)\ln(\rho/\delta)}{r}} \right) \cdot \TNorm{\matA - \matA_k} +
 \sqrt{\frac{6 \ln(\rho/\delta)}{r}} \cdot \FNorm{\matA - \matA_k}, \tag{vi} \\
 \TNorm{\matA_k - \matY \pinv{\matY} \matA} & \leq
\left(4 + \sqrt{\frac{3 \ln(n/\delta)\ln(\rho/\delta)}{r}} \right) \cdot \TNorm{\matA - \matA_k}  +
\sqrt{\frac{3 \ln(\rho/\delta)}{r}} \cdot \FNorm{\matA - \matA_k}, \tag{vii} \\
 \TNorm{\matA_k - \tilde{\matA}_k} & \leq \left(7 +
 \sqrt{\frac{12 \ln(n/\delta)\ln(\rho/\delta)}{r}} \right) \cdot \TNorm{\matA - \matA_k} +
 \sqrt{\frac{6 \ln(\rho/\delta)}{r}} \cdot \FNorm{\matA - \matA_k}. \tag{viii}
\end{align}
Recall that $\ell = \min\{m,r\}.$ The matrix $\matY $ can be constructed using $2 m n \log_2 (r+1)$ arithmetic operations and,
given $\matY,$ the matrices $\matY\pinv{\matY} \matA$ and $\tilde{\matA}_k$ can be formed using
$\const{O}(m n \ell + m r \ell)$ and $\const{O}( mn \ell + \ell^2 n)$ additional arithmetic operations, respectively.
\end{theorem}

We prove this theorem in Section~\ref{sec:guarantees}. Notice that the Theorem provides residual and forward error bounds for two low-rank matrices in the
spectral and Frobenius norms. The matrix $ \matY \pinv{\matY} \matA$ has rank at most $r > k,$ while the matrix $\tilde{\matA}_k$ has rank at most $k.$
Prior works have provided only residual error bounds~\cite{WLRT07,HMT,NDT09}.

The first two Frobenius norm bounds in this theorem (residual error analysis) are slightly stronger than the best bounds appearing in prior efforts~\cite{NDT09}.
The spectral norm bounds on the residual error are significantly better than the bounds presented in prior work and shed light on an open question mentioned in~\cite{NDT09}.
We do not, however, claim that the error bounds provided are the tightest possible. Certainly the specific constants ($22, 6,$ etc.) in the error estimates are not optimized.

We now present a detailed comparison of the guarantees given in Theorem~\ref{thm:quality-of-approximation-guarantee} with those available in the existing literature. 

\subsection{Detailed Comparison to Prior Work}
\label{sec:priorwork}

\subsubsection{Halko et al.~\cite{HMT}}
To put our result into perspective, we compare it to prior efforts at analyzing the SRHT algorithm introduced above.
Halko et al.~\cite{HMT} argue that if $r$ satisfies
\begin{equation}\label{rhmt}
4 \left[\sqrt{k} + \sqrt{8\ln(kn)} \right]^2 \ln(k) \leq r \leq n,
\end{equation}
then, for both $\xi=2,\mathrm{F}$,
$$ \XNorm{\matA - \matY \pinv{\matY} \matA} \le \left( 1 + \sqrt{7n/r} \right) \cdot \XNorm{\matA - \matA_k},$$
with probability at least $1 - \const{O}(1/k)$. Our first Frobenius norm bound is always tighter than the Frobenius norm bound given here.
To compare the spectral norm bounds, note that our first spectral norm bound is on the order of
\begin{equation}
\label{eqn:residspecbnd}
 \max\left\{\sqrt{\frac{ \ln(\rho/\delta) \ln(n/\delta)}{r} } \cdot \TNorm{\matA - \matA_k},\,  \sqrt{ \frac{ \ln(\rho/\delta) }{ r } } \cdot \FNorm{\matA - \matA_k}\right\}.
\end{equation}
If the singular values of $\matA$ are flat and $\matA$ has close to full rank, then the spectral norm result in~\cite{HMT} is perhaps optimal. But in the cases where it makes most sense to ask for low-rank approximations---viz., $\matA$ is rank-deficient or the singular values of $\matA$
decay fast---the spectral error norm bound in Theorem~\ref{thm:quality-of-approximation-guarantee} is more useful. Specifically, if
$$ \FNorm{\matA - \matA_k} \ll \sqrt{ \frac{n}{ \ln(\rho/\delta) } } \cdot \TNorm{\matA - \matA_k},$$
then when $r$ is chosen according to Theorem~\ref{thm:quality-of-approximation-guarantee} the quantity in Eqn.~(\ref{eqn:residspecbnd}) is much smaller than $$\sqrt{7 n/r} \cdot \TNorm{\matA - \matA_k}.$$

We were able to obtain this improved bound
by using the results in Section~\ref{sec:SRHT-orthonormal}, which allow one to take into account decays in the spectrum of $\matA$.
Finally, notice that our theorem makes explicit the intuition that the probability of failure can be driven to zero independently of the target rank $k$ by increasing the number of samples $r.$

\subsubsection{Nguyen et al.~\cite{NDT09}}
A tighter analysis of the Frobenius norm error term of the SRHT low-rank matrix approximation algorithm appeared in Nguyen et al.~\cite{NDT09}.
Let $\delta$ be a probability parameter with $0 < \delta < 1$ and $\varepsilon$ be an accuracy parameter with $0 < \varepsilon <1 $.
Then, Nguyen et al. show that in order to
get a rank-$k$ matrix $\tilde{\matA}_k$ satisfying
$$ \FNorm{ \matA - \tilde{\matA}_k} \le \left(1 + \varepsilon \right)  \cdot \FNorm{\matA-\matA_k}$$
and
\[
 \TNorm{\matA -\tilde{\matA}_k} \leq
 \left( 2 + \sqrt{2 n / r} \right) \cdot \TNorm{\matA-\matA_k}
\]
with probability of success at least $1 - 5 \delta$, one requires
\[
r = \Omega\left( \varepsilon^{-1} \max\{  k, \sqrt{k} \ln(2n/\delta) \}  \cdot \max\{ \ln k, \ln(3/\delta) \}  \right).
\]
Theorem~\ref{thm:quality-of-approximation-guarantee} gives a tighter spectral norm error bound in the cases most of interest, where $\FNorm{\matA - \matA_k} \ll \sqrt{ \frac{n}{ \ln(\rho/\delta) } } \cdot \TNorm{\matA - \matA_k}.$ It also provides an equivalent Frobenius norm error bound with a comparable failure probability for a smaller number of samples. Specifically, if
\[
 r \geq 528 \varepsilon^{-1}[\sqrt{k} + \sqrt{8 \ln(8 n/\delta)}]^2 \ln(8k/\delta) = \Omega\left( \varepsilon^{-1} \max\{k, \ln(n/\delta)\} \cdot \max\{\ln k, \ln(1/\delta)\} \right),
\]
then the second Frobenius norm bound in Theorem~\ref{thm:quality-of-approximation-guarantee} ensures $\FNorm{ \matA - \tilde{\matA}_k} \le \left(1 + \varepsilon \right)  \cdot \FNorm{\matA-\matA_k},$
with probability at least $1 - 8\delta.$

In the conclusion of~\cite{NDT09}, the authors left as a subject for future research the explanation of a curious experimental phenomenon: when the singular values decay according to power laws, the SRHT low-rank approximation algorithm empirically achieves relative-error spectral norm approximations. Our spectral norm result provides an explanation of this phenomenon: when the singular values of $\matA$ decay fast enough, as in power law decay, one has $ \FNorm{\matA-\matA_k} = \Theta\left( 1 \right) \cdot \TNorm{\matA-\matA_k} $. In this case, by choosing  $r$
$$ 24 \varepsilon^{-1} \left[\sqrt{k} + \sqrt{8\ln(n/\delta)} \right]^2 \ln(k/\delta) \ln(n/\delta) \leq r \leq n$$
our second spectral norm bound ensures
$
 \TNorm{\matA - \tilde{\matA}_k} \leq \const{O}( 1 ) \cdot \TNorm{\matA-\matA_k}
$
with probability of at least $1 - 8\delta,$ thus predicting the observed empirical behavior of the algorithm.

\subsubsection{The subsampled randomized Fourier transform (SRFT)}
The algorithm in Section 5.2 of~\cite{WLRT07},
which was the first to use the idea of employing subsampled randomized orthogonal transforms to compute low-rank approximations to matrices,
provides a spectral norm error bound but replaces the SRHT with an SRFT, i.e. the matrix $\matH$ of Definition~\ref{def:srht} is replaced by a matrix
where the $(j,h)$th entry is $\matH_{jh} = e^{- 2 \pi i (j-1) (h-1)/n }$, where $i = \sqrt{-1}$, i.e. $\matH$ is the unnormalized discrete Fourier transform.
Woolfe et al.~\cite{WLRT07} (see eqn. 190) argue that, for any $\alpha > 1$, $\beta > 1$,
if
$$r \ge \alpha^2 \beta \left( \alpha-1 \right)^{-1} (2k)^2,$$
then with probability at least $1 - 3/\beta$ ($\omega = \max\{m,n\}$),
$$ \TNorm{ \matA - \tilde{\matU}_k \tilde{\matSig}_k \tilde{\matV}_k\transp } \le
2 \left( \sqrt{2\alpha-1} + 1 \right) \cdot \left( \sqrt{\alpha \omega +1} + \sqrt{ \alpha \omega } \right) \cdot \TNorm{\matA - \matA_k}.$$
Here, $\tilde{\matU}_k \in \R^{m \times k}$ contains orthonormal columns, as does $\tilde{\matV}_k \in \R^{n \times k}$, while
$\tilde{\matSig}_k \in \R^{k \times k} $ is diagonal with nonegative entries. These matrices can be computed deterministically from $\matA \matTh\transp$
in $\const{O}( k^2(m+n) + k r^2 \ln r)$ time. Also, computing $\matY = \matA \matTh\transp$ takes $O( m n \ln r)$ time.

\subsubsection{Two alternative dimensionality-reduction algorithms}
Instead of using an SRHT matrix, one can take $\matTh\transp$  in Theorem~\ref{thm:quality-of-approximation-guarantee} to be a matrix of i.i.d standard Gaussian random variables. One gains theoretically and often empirically better worse-case trade-offs between the number of samples taken, the failure probability, and the error guarantees. The SRHT algorithm is still faster, though, since matrix multiplications with Gaussian matrices require $\const{O}(mnr)$ time.
One can also take $\matTh\transp$ to be a matrix of i.i.d.~random signs ($\pm 1$ with equal probability). In many ways, this is analogous to the Gaussian algorithm---in both cases $\matTh$ is a matrix of i.i.d. subgaussian random variables---so we expect this algorithm to have the same advantages and disadvantages relative to the SRHT algorithm. We now compare the best available performance bounds for these schemes to our SRHT performance bounds.

We use the notion of the stable rank of a matrix,
$$
 \stablerank{\matA} = \FNormS{\matA}/\TNormS{\matA},
$$
to capture the decay of the spectrum of $\matA$ (spectrum here refers to the singular values of $\matA$).
As can be seen by considering a matrix with a flat spectrum, in general the stable rank is no smaller than the rank; the smaller the stable rank, the more pronounced the decay in the spectrum of $\matA.$

When $r > k+4,$ Theorem 10.7 and Corollary 10.9 in~\cite{HMT} imply that, when using Gaussian sampling,
with probability at least $1 - 2 \cdot 32^{-(r-k)} - e^{\frac{-(r-k+1)}{2}}$,
\[
 \FNorm{\matA - \matY \pinv{\matY} \matA} \leq \left( 1 + 32 \frac{ \sqrt{3 k} + \expe \sqrt{r} }{\sqrt{r -k+1}}  \right) \cdot \FNorm{\matA - \matA_k}
\]
and with probability at least $1 - 3e^{-(r-k)}$,
\[
 \TNorm{\matA - \matY\pinv{\matY} \mat A} \leq \left(1 + 16\sqrt{1 + \frac{k}{r-k}} \right) \cdot \TNorm{\matA - \matA_k} + \frac{8\sqrt{r}}{r-k+1} \cdot \FNorm{\matA - \matA_k}.
\]
Comparing to the guarantees of Theorem~\ref{thm:quality-of-approximation-guarantee} we see that these bounds suggest that with the same number of samples, Gaussian low-rank approximations outperform SRHT low-rank approximations. In particular, the spectral norm bound guarantees that if $\stablerank{\matA - \matA_k} \leq k$, i.e.
$ \FNorm{\matA - \matA_k} \le \sqrt{k} \TNorm{\matA - \matA_k},$
then the Gaussian low-rank approximation algorithm requires $\const{O}(k/\varepsilon^2)$ samples to return a
$(17+\varepsilon)$ constant factor spectral norm error approximation with high probability. Similarly, the Frobenius norm bound guarantees that the same number of samples returns a $(1 + 32 \varepsilon)$ constant factor Frobenius norm error approximation with high probability. Neither the spectral nor Frobenius bounds given in Theorem~\ref{thm:quality-of-approximation-guarantee} for SRHT low-rank approximations apply for this few samples.

\cite{Zou10} does not consider the Frobenius norm error of the random sign low-rank approximation algorithm, but Remark 4 in~\cite{Zou10} shows that when
$r = \const{O}(k/\varepsilon^4 \ln(1/\delta) )$, for $1 < \delta < 0,$ and $\stablerank{\matA - \matA_k} \leq k,$ this algorithm ensures that with high probability of at least $1-\delta$,
\[
 \TNorm{\matA - \matY\pinv{\matY} \mat A} \leq (1 + \varepsilon) \TNorm{\matA - \matA_k}.
\]

To compare our results to those stated in~\cite{HMT,Zou10} we assume that $k \gg \ln(n/\delta)$ so that $r > k\ln k$ suffices for Theorem~\ref{thm:quality-of-approximation-guarantee} to apply.  Then, in order to acquire a $(4 + \varepsilon)$ relative error bound from Theorem~\ref{thm:quality-of-approximation-guarantee}, it suffices that
(here $\const{C}^\prime$ is an explicit constant no larger than 6)
\[
 r \geq \const{C}^\prime \varepsilon^{-2} k \ln(\rho/\delta) \quad \text{and}\quad \stablerank{\matA - \matA_k} \leq \const{C}^\prime k.
\]

We see that the Gaussian and random sign approximation algorithms return $(17+\varepsilon)$ and $(1+\varepsilon)$ relative spectral error approximations, respectively, when $r$ is on the order of $k$ and the relatively weak spectral decay condition $\stablerank{\matA - \matA_k} \leq k$ is satisfied, while our bounds for the SRHT algorithm require $r > k \ln (\rho/\delta)$ and the spectral decay condition
$$\stablerank{\matA - \matA_k} \leq \const{C}^\prime k$$
to ensure a $(6 + \varepsilon)$ relative spectral error approximation. We note that the SRHT algorithm can be used to obtain relative spectral error approximations of matrices with arbitrary stable rank at the cost of increasing $r$ (the same is of course true for the Gaussian and random sign algorithms).

The disparity in the bounds for these three schemes---the presence of the logarithmic factors in the SRHT bounds and the fact that these bounds apply only when $r > k \ln(\rho/\delta)$---may reflect a fundamental trade-off between the structure and randomness of $\matTh\transp$. The highly structured nature of SRHT matrices makes it possible to calculate $\matY$ much faster than when Gaussian or random sign sampling matrices are used, but this moves us away from the very nice isotropic randomness present in the Gaussian $\matTh\transp$ and the similarly nice properties of a matrix of i.i.d subgaussian random variables, thus resulting in slacker bounds which require more samples.

\section{Least squares regression}\label{sec:regression}

We now show how one can use the SRHT to  solve least squares problems of the form
$$ \min_{\x} \TNorm{ \matA \x - \b }. $$
Here $\matA$ is an $m \times n$ matrix with $m \gg n$ and $\rank(\matA)=n$, $\b \in \R^m$, and $\x \in \R^n$.
One approach to solve this optimization problem is via the  SVD of $\matA$: $ \x_{opt} = \pinv{\matA} \b;$
while an example of an iterative algorithm is the LSQR algorithm in~\cite{PS82}.

During the last decade, researchers have developed several randomized algorithms that (approximately) solve the
regression problem in less running time than the approaches mentioned above~\cite{Sar06,RT08,Zou10,AMT10,DMMS11}. We refer the reader to Section 3.3 in~\cite{Bou11a} for a survey of these methods. The fastest non-iterative method is in~\cite{DMMS11} while the fastest iterative algorithm is in~\cite{RT08,AMT10}.
Both approaches employ the Subsampled Randomized Hadamard Transform. 

\subsection{Least-squares via the SRHT and the SVD}
The idea in the SRHT algorithm of Drineas et al.~\cite{DMMS11} is to reduce the dimensions of $\matA$ and $\b$ by pre-multiplication
with an SRHT matrix $\matTh \in \R^{r \times m}$ (the matrix $\matR$ in this case is constructed by uniform sampling without replacement)
and then solve quickly the smaller problem,
$$ \min_{\x} \TNorm{ \matTh \matA \x - \matTh\b }. $$
Let $\tilde{\x}_{opt} = \pinv{ \left( \matTh \matA \right) } \matTh \b$; then, assuming $r$ satisfies ($\varepsilon > 0$ is an accuracy parameter)
$$ r = \max\{  48^2 n \ln(40 m n) \ln (10^4 n \ln(40 m n)) ,   40 n \ln(40 m n) / \varepsilon  \}, $$
\cite{DMMS11} shows that with probability at least $0.8$,
$$  \TNorm{  \matA \tilde{\x}_{opt} - \b } \le \left( 1 + \varepsilon \right) \cdot   \TNorm{  \matA \x_{opt} - \b }. $$
Furthermore, assume that there exists a $\gamma \in (0,1]$ such that $ \TNorm{\matU_{\matA} \matU_{\matA}\transp\b} = \gamma \TNorm{\b}$.
Then, with the same probability,
\[
 \TNorm{ \x_{opt} - \tilde{\x}_{opt} } \le \sqrt{\varepsilon} \left( \kappa{ \left( \matA \right)} \sqrt{ \gamma^{-2}   - 1 } \right) \TNorm{\x_{opt}}.
\]
Here, $\kappa(\matA)$ is the two-norm condition number of $\matA$:
\[
\kappa(\matA) = \TNorm{\matA}\TNorm{\pinv{\matA}}.
\]
The running time of this approximation algorithm is $\const{O}(  m n \log_2 r + rn^2)$,
since the SRHT multiplication takes $\const{O}(  m n \log_2 r)$ time and the solution of the small regression problem
another $\const{O}(rn^2)$.

Below, we provide a novel analysis of this SRHT least squares algorithm which shows that one needs
asymptotically fewer samples $r$. This immediately implies an improvement on the running time of the algorithm.
Additionally, we show logarithmic dependence on the failure probability.
\begin{theorem}
\label{regression1}
Let $\matA \in \R^{m \times n}$ ($m \gg n$) have rank $\rho=n$ and $n$ be a power of 2; let $\b \in \R^m$.  Let $0 < \varepsilon < 1/3$
denote an accuracy parameter, $0 < \delta < 1$ be a failure probability, and $\const{C}\ge 1$ be a constant.
Let $\matTh$ be an $r \times m$ SRHT matrix with $r$ satisfying
\[
 6 \const{C}^2 \varepsilon^{-1} \left[\sqrt{n} + \sqrt{8\ln(m/\delta)} \right]^2 \ln(n/\delta) \leq r \leq m.
\]
Then, with probability at least $1 -  \delta^{\const{C}^2\ln(n/\delta)/4} - 7\delta$,
\[
\TNorm{\matA \tilde{\x}_{opt}-\b } \le  \left( 1 +  22 \varepsilon \right) \cdot \TNorm{\matA \x_{opt}-\b}.
\]
Furthermore, assume that there exists a $\gamma \in (0,1]$ such that $ \TNorm{\matU_{\matA} \matU_{\matA}\transp\b} = \gamma \TNorm{\b}$.
Then, with the same probability,
\[
 \TNorm{ \x_{opt} - \tilde{\x}_{opt} } \le \left( \frac{ 1-\sqrt{\varepsilon} }{4 \varepsilon} \right)^{\frac{1}{2}} \left( \kappa{ \left( \matA \right)} \sqrt{ \gamma^{-2}   - 1 } \right) \TNorm{\x_{opt}}.
\]
\end{theorem}
\noindent
We prove this theorem in  Section~\ref{sec:guarantees3}. Another possibility to obtain a better analysis of the method of Drineas et al.
is to use Lemma~\ref{lemma:SRHT-preserves-geometry2} in this article, which was proved in~\cite{IW12} and presents bounds for sampling without
replacement. This analysis is not straightforward and is beyond the scope of this paper.

\subsection{Iterative methods}
The key idea of an iterative algorithm such as the LSQR method of~\cite{PS82} is \emph{preconditioning}.
Blendenpik in~\cite{AMT10} constructs such a preconditioner by using the SRHT
(the matrix $\matR$ in this case is constructed by uniform sampling without replacement) as follows.
First, an SRHT matrix $\matTh \in \R^{r \times m}$ is constructed. Then, one forms a QR factorization
$\matTh \matA = \matQ \matR_{\matA}$, with $\matQ \in \R^{r \times n}$ and $\matR_{\matA} \in \R^{n \times n}$. Finally,
$\matA$ and  $\matR_{\matA}$ are given as inputs to LSQR to find a solution to the least squares
problem. We refer the reader to~\cite{AMT10} (see also~\cite{IW12}) for a detailed discussion  of this approach. The purpose of our discussion
here is to comment on the first step of the above procedure and
show that a preconditioner of the same quality can be constructed with a smaller $r$.
Avron et al.~\cite{AMT10} argue that if the number of samples is sufficiently large
then the two-norm condition number of $\matA \matR_{\matA}^{-1}$ is small. A small condition number is desirable because
the number of iterations required for convergence of the LSQR method  is proportional to the condition number.
More specifically, Theorem 3.2 in~\cite{AMT10} argues that with  $r = \Omega\left( n \ln(m) \ln(n \ln(m))\right)$, and with constant probability (e.g. $0.9$),
$$ \kappa{ \left( \matA \matR_{\matA}^{-1} \right) } = \const{O}( 1).$$
The analysis of Blendenpik was recently improved in~\cite{IW12}.
More specifically,  Corollary 3.11 in~\cite{IW12}, along with Lemma~\ref{prop:SRHT-equalizes-columns-of-orthonormal-matrices} in our manuscript,
which gives a bound on the coherence,  show that if
\[
 \frac{8}{3} \varepsilon^{-2} \left[\sqrt{n} + \sqrt{8\ln(m/\delta)} \right]^2 \ln(2n/\delta) \leq r \leq m,
\]
then, with probability at least $1 -  2 \delta$,
$$\kappa{ \left( \matA \matR_{\matA}^{-1} \right) } \le \sqrt{ \frac{1 +\varepsilon}{1 - \varepsilon}}.$$
 We now provide a similar bound in the case where the SRHT is constructed via sampling without replacement.
This bound is a simple combination of results in prior work.
More specifically, Theorem 1 in~\cite{RT08} argues that the two-norm condition number of $\matA \matR_{\matA}^{-1}$ equals the two-norm condition number
of $\matU\transp \matTh\transp$, where $\matU \in \R^{m \times n}$ contains the top $n$ left singular vectors of $\matA$.
Combine this fact with the bounds on the singular values of $\matU\transp \matTh\transp$ from Lemma~\ref{lemma:SRHT-preserves-geometry},
to obtain the following observation.

{ \bf Remark.}
Let $\matA \in \R^{m \times n}$ ($m \gg n$) have rank $\rho=n$ and $n$ be a power of $2$. Fix $0 < \delta < 1$ and $0 < \varepsilon < 1/3$.
Construct the upper triangular matrix $\matR_{\matA} \in \R^{n \times n}$ via the QR factorization $\matTh \matA = \matQ \matR_{\matA}$,
where $\matTh$ is an $r \times m$ SRHT matrix with
$r$ satisfying
\[
 6 \varepsilon^{-2} \left[\sqrt{n} + \sqrt{8\ln(m/\delta)} \right]^2 \ln(2n/\delta) \leq r \leq m.
\]
Then, with probability at least $1 -  2 \delta$,
$$\kappa{ \left( \matA \matR_{\matA}^{-1} \right) } \le \sqrt{ \frac{1 +\varepsilon}{1 - \varepsilon}}.$$
Finally, notice that we form the SRHT by uniform sampling without replacement while Blendenpik samples the columns of the randomized Hadamard matrix \emph{with} replacement.
A different sampling scheme - Bernoulli sampling - was analyzed in Theorem 6.1 in~\cite{TailBounds} and Section 4 in~\cite{IW12}.

\subsubsection{The subsampled randomized Fourier transform (SRFT)} Finally, we mention the work of Rokhlin and Tygert~\cite{RT08}, which was the first to use
the idea of employing subsampled randomized orthogonal transforms to precondition iterative solvers for least squares regression problems.
~\cite{RT08} replaces the SRHT with the SRFT; notice though that one still needs $O\left( m n \ln r \right)$ time to compute the product $\matTh \matA$.
In this case, for any $\alpha > 1$, $0 < \delta < 1$, if
$$r \ge \left( \frac{\alpha^2+1}{\alpha^2-1} \right)^2  \frac{n^2}{\delta},$$
then with probability at least $1 - \delta$,
$$ \kappa\left( \matA \matR_{\matA}^{-1}\right) \le \alpha.$$

\section{Matrix Computations with SRHT matrices}
\label{sec:SRHT}

\subsection{SRHTs applied to orthonormal matrices}
\label{sec:SRHT-orthonormal}


An important ingredient in analyzing the low-rank approximation algorithm of Theorem~\ref{thm:quality-of-approximation-guarantee}
is understanding how an SRHT changes the spectrum of a matrix after postmultiplication: given a matrix $\matX$ and
an SRHT matrix $\matTh$, how are the singular values of $\matX$ and $\matX\matTh\transp$ related? To be more precise,
Lemma~\ref{prop:structural-result} in Section~\ref{sec:presth} suggests that one path towards establishing the efficacy of SRHT-based low-rank approximations lies in understanding how the SRHT perturbs the singular values of orthonormal matrices. To see this, we informally repeat the statement of the lemma here. Let $\matA \in \R^{m \times n}$ have rank $\rho.$ Fix $k$ satisfying $0 \leq k \leq \rho$.
Given a matrix $\matOmega \in \R^{n \times r}$, with $r \ge k$, construct $\matY = \matA \matOmega.$ If $\matV_k \transp \matOmega$ has full row-rank, then,
for $\xi=2, \mathrm{F}$,
\begin{equation}
\XNormS{\matA - \matY \pinv{\matY} \matA}
\leq
\XNormS{ \matA - \matA_k } + \XNormS{\matSig_{\rho - k} \matV_{\rho-k} \transp \matOmega \pinv{\left( \matV_k \transp \matOmega \right)} }.
\end{equation}
Now take $\matOmega = \matTh\transp$
and observe that if the product $\matSig_{\rho - k}\matV_{\rho-k}\transp \matTh\transp \pinv{\left( \matV_k\transp \matTh\transp\right)}$ has small norm, then the residual error of the approximant $\matY\pinv{\matY}\matA$ is small. The norm of this product is small when the norms of the perturbed orthonormal matrices $\matV_{\rho-k}\transp\matTh\transp$ and $\pinv{\left(\matV_k\transp \matTh\transp\right)}$ are in turn small, because
\begin{equation}
\label{eqn:basicestimate}
\XNormS{\matSig_{\rho - k} \matV_{\rho-k} \transp \matTh\transp \pinv{\left( \matV_k \transp \matTh\transp \right)}} \leq \XNormS{\matSig_{\rho - k}} \XNormS{\matV_{\rho-k} \transp \matTh\transp} \XNormS{\pinv{\left( \matV_k \transp \matTh\transp \right)}}.
\end{equation}
These perturbed orthogonal matrices have small norm precisely when their singular values are close to those of the original orthogonal matrices.

\subsubsection{SRHTs by uniform sampling without replacement}

In this section, we collect known results on how the singular values of a matrix with orthonormal rows are affected
by postmultiplication by an SRHT matrix.

It has recently been shown by Tropp~\cite{Tro11} that, if the SRHT matrix is of sufficiently large dimensions,
post-multiplying a short-fat matrix with orthonormal rows with an SRHT matrix preserves
the singular values of the orthonormal matrix, with high probability, up to a small multiplicative factor.
The following lemma is essentially a restatement of Theorem 3.1 in~\cite{Tro11}, but we include a full proof (later in this subsection) for completeness.
\begin{lemma}[The SRHT preserves geometry]
\label{lemma:SRHT-preserves-geometry}
Let $\matV \in \R^{n \times k}$ have orthonormal columns and $n$ be a power of 2. Let $0 < \varepsilon < 1/3$ and $0 < \delta < 1.$
Construct an SRHT matrix $\matTh \in \R^{r \times n}$ with $r$ satisfying
\begin{equation}\label{eqn:r2}
6\varepsilon^{-1}\left[\sqrt{k} + \sqrt{8\ln(n/\delta)} \right]^2 \ln (k/\delta) \leq r \leq n.
\end{equation}
Then, with probability at least $1 - 3\delta$, for all $i=1,...,k$,
$$
\sqrt{1 - \sqrt{\varepsilon}} \le \sigma_i( \matV \transp \matTh\transp ) \le \sqrt{1 + \sqrt{\varepsilon}}
$$
and
$$
\TNorm{ \pinv{(\matV\transp \matTh\transp)} - (\matV\transp \matTh\transp)\transp } \le 1.54 \sqrt{\varepsilon}.
$$
\end{lemma}
Tropp~\cite{Tro11} (see also~\cite{AC06}) argues that the above lemma follows from a more fundamental fact: if $\matV$ has orthonormal columns,
then the rows of the product $\matH\matD \matV$ all have roughly the same norm. That is, premultiplication by $\matH\matD$
equalizes the row norms of an orthonormal matrix.

\begin{lemma}[Row norms, Lemma 3.3 in~\cite{Tro11}]
\label{prop:SRHT-equalizes-columns-of-orthonormal-matrices}
Let $\matV \in \R^{n \times k}$ have orthonormal columns ($n$ is a power of 2), $\matH \in \R^{n \times n}$ be a normalized Hadamard matrix, $\matD \in \R^{n \times n}$ be a diagonal matrix of independent random signs, and $0 < \delta < 1$ be a failure probability. Recall that $\left(\matH \matD  \matV \right)_{(i)}$ denotes the $i$th row of the matrix $\matH \matD  \matV \in \R^{n \times k}$. Then, with probability at least $1-\delta$,
\[
\max\nolimits_{i=1,...,n} \TNorm{\left(\matH \matD \matV \right)_{(i)}} \le \sqrt{\frac{k}{n}}  + \sqrt{ \frac{8 \ln(n/\delta)}{n}}.
\]
\end{lemma}
To prove Lemma~\ref{lemma:SRHT-preserves-geometry} we need one more result on uniform random sampling (without replacement) of rows
from tall-thin matrices with orthonormal columns.

\begin{lemma}[Uniform Sampling without replacement from an Orthonormal Matrix, Corollary to Lemma~3.4 of~ \cite{Tro11} ]
\label{lemma:sampling-ortho}
Let $\matW \in \R^{n \times k}$ have orthonormal columns.
Let $0 < \varepsilon < 1$ and $0 < \delta < 1$. Let $ M:= n \cdot  \max\nolimits_{i=1,...,n} \TNormS{\matW_{(i)}}$.
Let $r$ be an integer such that
\begin{equation}\label{eqn:r3}
6 \varepsilon^{-2} M \ln (k/\delta) \leq r \leq n \,.
\end{equation}
Let $\matR \in \R^{r \times n}$ be a matrix which consists of a subset of $r$ rows from $\matI_n$
where the rows are  chosen uniformly at random and without replacement.
Then, with probability at least $1-2\delta$, for $i\in[k]$:
$$ \sqrt{\frac{r}{n}} \cdot \sqrt{1-\varepsilon} \le \sigma_i( \matR \matW) \le  \sqrt{1+\varepsilon} \cdot \sqrt{\frac{r}{n}} .$$
\end{lemma}
\begin{proof}
Apply Lemma 3.4 from~\cite{Tro11} with the following choice of parameters:
$\ell = \alpha M \ln(k/\delta),$
$\alpha = 6/\varepsilon^2,$ and
$\delta_{tropp} = \eta = \varepsilon$.
Here, $\ell$, $\alpha$, $M$, $k$, $\eta$ are the variables of  Lemma 3.4 from~\cite{Tro11} (we also use $M$ and $k$), and
$\delta_{tropp}$ plays the role of $\delta$, an error parameter, of  Lemma 3.4 from~\cite{Tro11}.
The variables $\varepsilon$ and $\delta$ are from our Lemma.
The choice of $\ell$ proportional to $\ln(k/\delta)$ rather than proportional to $\ln(k)$, as in the original statement of Lemma~3.4, is what results in a probability proportional to $\delta$ instead of $k$; this can easily be seen by tracing the modified choice of $\ell$ through the proof of Lemma~3.4.
\end{proof}

\begin{proof} (of Lemma~\ref{lemma:SRHT-preserves-geometry})
To obtain the bounds on the singular values, we combine Lemmas~\ref{prop:SRHT-equalizes-columns-of-orthonormal-matrices} and~\ref{lemma:sampling-ortho}.
More specifically, apply Lemma~\ref{lemma:sampling-ortho} with $\matW = \matH \matD \matV$ and use the bound for $M$ from Lemma~\ref{prop:SRHT-equalizes-columns-of-orthonormal-matrices}. Then, the bound on $r$ in Eqn.~(\ref{eqn:r3}), the bound on the singular values in Lemma~\ref{lemma:sampling-ortho},
and the union bound, establish that with probability at least $1 - 3\delta$,
$$ \sqrt{\frac{r}{n}} \cdot \sqrt{1-\varepsilon} \le \sigma_i( \matR \matH \matD \matV) \le  \sqrt{1+\varepsilon} \cdot \sqrt{\frac{r}{n}} .$$
Now, multiply this inequality with $\sqrt{n/r}$ and recall the definition $\matTh = \sqrt{\frac{n}{r}} \cdot \matR \matH  \matD$ to obtain
$$  \sqrt{1-\varepsilon} \le \sigma_i( \matTh \matV) \le  \sqrt{1+\varepsilon}.$$
Replacing $\varepsilon$ with $\sqrt{\varepsilon}$ and using the bound on $r$ in Eqn.~(\ref{eqn:r2}) concludes the proof.

The second bound in the lemma follows from the first bound after a simple algebraic manipulation.
Let $\matX = \matV\transp \matTh\transp \in \R^{k \times r}$ with SVD $\matX = \matU_{\matX} \matSig_{\matX} \matV_\matX\transp $. Here,  $\matU_{\matX} \in \R^{k \times k}$, $\matSig_{\matX} \in \R^{k \times k}$, and $\matV_\matX \in \R^{r \times k}$, since $r > k$. Consider taking the SVDs of $\pinv{(\matV\transp \matTh\transp)}$ and $(\matV\transp \matTh\transp)\transp $,
 \begin{align*}
\TNorm{\pinv{(\matV\transp \matTh\transp)} - (\matV\transp \matTh\transp)\transp} =
    \TNorm{\matV_{\matX} \matSig_{\matX}^{-1} \matU_{\matX}\transp - \matV_{\matX} \matSig_{\matX} \matU_{\matX}\transp  } &=&
    \TNorm{\matV_{\matX}(\matSig_{\matX}^{-1} - \matSig_{\matX}) \matU_{\matX}\transp}\\ &=&  \TNorm{\matSig_{\matX}^{-1} - \matSig_{\matX}},
\end{align*}
since $\matV_{\matX}$ and $\matU_{\matX}\transp $ can be dropped without changing the spectral norm. Let $\matY = \matSig_{\matX}^{-1} - \matSig_{\matX} \in \R^{k \times k}$. Then, for all $i=1,\ldots ,k$, $\matY_{ii}  = \frac{ 1 - \sigma_i^2(\matX)  }{ \sigma_{i}(\matX) }.$ We conclude the proof as follows,
 \begin{align*}
\TNorm{ \matY } =
\max\nolimits_{1 \leq i \leq k} \abs{\matY_{ii} } =
\max\nolimits_{1 \leq i \leq k} \abs{\frac{ 1 - \sigma_i^2(\matX)}{ \sigma_{i}(\matX) } } =
\max\nolimits_{1 \leq i \leq k} \frac{ \abs{1 - \sigma_i^2(\matX)} }{ \sigma_{i}(\matX) } &\le& \frac{ \sqrt{\varepsilon} }{\sqrt{1-\sqrt{\varepsilon}}} \\ &\le& 1.54 \sqrt{\varepsilon}.
 \end{align*}
\end{proof}

\subsubsection{SRHTs by uniform sampling with replacement}
Lemma~\ref{lemma:SRHT-preserves-geometry} and Lemma~\ref{lemma:sampling-ortho} analyze uniform random
sampling without replacement. Below, we present the analogs of these two lemmas for uniform random sampling with replacement.
Lemma~\ref{lemma:sampling-ortho2} is essentially a restatement of Algorithm 2 (with the probabilities set to $1/m$)
along with the third point in Remark 3.9 and Lemma 2.1 (with $\alpha=\sqrt{n/r}$) in~ \cite{IW12}.
\begin{lemma}[Uniform Sampling with replacement from an Orthonormal Matrix~ \cite{IW12} ]
\label{lemma:sampling-ortho2}
Let $\matW \in \R^{n \times k}$ have orthonormal columns.
Let $0 < \varepsilon < 1$ and $0 < \delta < 1$. Let $ M:= n \cdot \max\nolimits_{i=1,...,n} \TNormS{\matW_{(i)}}$.
Let $r$ be an integer such that
\begin{equation}\label{eqn:r3v2}
\frac{8}{3} \varepsilon^{-2} M \ln (k/\delta) \leq r \leq n \,.
\end{equation}
Let $\hat\matR \in \R^{r \times n}$ be a matrix which consists of a subset of $r$ rows from $\matI_n$
where the rows are  chosen uniformly at random and with replacement.
Then, with probability of at least $1-2\delta$, for $i\in[k]$:
$$ \sqrt{1-\varepsilon} \le \sigma_i\left( \sqrt{\frac{n}{r}} \hat\matR \matW \right) \le  \sqrt{1+\varepsilon} .$$
\end{lemma}

\begin{lemma}[The SRHT preserves geometry]
\label{lemma:SRHT-preserves-geometry2}
Let $\matV \in \R^{n \times k}$ have orthonormal columns and $n$ be a power of 2. Let $0 < \varepsilon < 1$ and $0 < \delta < 1.$
Construct an SRHT matrix $\matTh \in \R^{r \times n}$ ($\matR$ is constructed as in Lemma~\ref{lemma:sampling-ortho2}, i.e. via
uniform random sampling with replacement) with $r$ satisfying
\begin{equation}\label{eqn:r4}
\frac{8}{3} \varepsilon^{-1} \left[\sqrt{k} + \sqrt{8\ln(n/\delta)} \right]^2 \ln (k/\delta) \leq r \leq n.
\end{equation}
Then, with probability at least $1-3\delta,$ for all $i=1,...,k$,
$$
\sqrt{1 - \sqrt{\varepsilon}} \le \sigma_i\left( \matV \transp \matTh\transp \right) \le \sqrt{1 + \sqrt{\varepsilon}}
$$
\end{lemma}
\begin{proof}
To obtain the bounds on the singular values, we combine Lemmas~\ref{prop:SRHT-equalizes-columns-of-orthonormal-matrices} and~\ref{lemma:sampling-ortho2}.
More specifically, apply Lemma~\ref{lemma:sampling-ortho2} with $\matW = \matH \matD \matV$ and use the bound for $M$ from Lemma~\ref{prop:SRHT-equalizes-columns-of-orthonormal-matrices}. Then, the bound on $r$ in Eqn.~(\ref{eqn:r3}), the bound on the singular values in Lemma~\ref{lemma:sampling-ortho2},
and the union bound, establish that with probability of at least $1 - 3\delta$,
$$  \sqrt{1-\varepsilon} \le \sigma_i\left( \sqrt{\frac{n}{r}} \cdot \matR \matH \matD \matV \right) \le  \sqrt{1+\varepsilon}  .$$
Replacing $\varepsilon$ with $\sqrt{\varepsilon}$ and using the bound on $r$ in Eqn.~(\ref{eqn:r3}) concludes the proof.

\end{proof}

\subsection{SRHTs applied to general matrices}
\label{sec:SRHT-orthogonal}

The structural result in Lemma~\ref{prop:structural-result},
Lemma~\ref{lemma:SRHT-preserves-geometry} on the perturbative effects of SRHTs on the singular values of orthonormal matrices,
and the basic estimate in~\eqref{eqn:basicestimate} are enough to reproduce the results on the approximation error of SRHT-based low-rank approximation in~\cite{HMT}. The main contribution of this paper is the realization that one can take advantage of the decay in the singular values of $\mat{A}$ encoded in $\matSig_{\rho -k}$ to obtain sharper results. In view of the fact that
\begin{equation}
 \XNormS{\matSig_{\rho - k} \matV_{\rho-k} \transp \matTh\transp \pinv{\left( \matV_k \transp \matTh\transp \right)}} \leq \XNormS{\matSig_{\rho - k} \matV_{\rho-k} \transp \matTh\transp} \XNormS{\pinv{\left( \matV_k \transp \matTh\transp \right)}},
\end{equation}
we should consider the behavior of the singular values of $\matSig_{\rho - k}\matV_{\rho-k} \transp \matTh\transp$ instead of those of $\matV_{\rho-k} \transp \matTh\transp.$ Accordingly, in this section we extend the analysis of~\cite{Tro11} to apply to the application of SRHTs to general matrices.

Our main tool is a generalization of Lemma~\ref{prop:SRHT-equalizes-columns-of-orthonormal-matrices} that states that the maximum column norm of a matrix to which an SRHT has been applied is, with high probability, not much larger than the root mean-squared average of the column norms of the original matrix.
\subsubsection{SRHT equalizes column-norms}
\
\begin{lemma}[SRHT equalization of column-norms]
\label{lemma:colnorm-tail-bound}
Suppose that $\matA$ is a matrix with $n$ columns and $n$ is a power of 2. Let $\matH \in \R^{n \times n}$ be a normalized Walsh--Hadamard matrix, and $\matD \in \R^{n \times n}$ a diagonal matrix of independent random signs. Then for every $t \geq 0,$
\[
 \Probab{    \max\nolimits_{i=1,...,n} \TNorm{ \left(\matA \matD \matH\transp \right)^{(i)}}   \leq \frac{1}{\sqrt{n} } \FNorm{\matA} + \frac{t}{\sqrt{n}} \TNorm{\matA} } \geq 1 - n \cdot \expe^{-t^2/8}.
\]
\end{lemma}

\begin{proof}
Our proof of Lemma~\ref{lemma:colnorm-tail-bound} is essentially that of Lemma~\ref{prop:SRHT-equalizes-columns-of-orthonormal-matrices} in \cite{Tro11}, with attention paid to the fact that $\matA$ is no longer assumed to have orthonormal columns. In particular, the following concentration result for Lipschitz functions of Rademacher vectors is central to establishing the result.
Recall that a Rademacher vector is a random vector whose entries are independent and take the values $\pm 1$ with equal probability.

\begin{lemma}[Concentration of convex Lipschitz functions of Rademacher random variables {[Corollary 1.3 ff. in \cite{Ledoux96}]} ]
\label{prop:rademacher-concentration}
 Suppose $f$ is a convex function on vectors that satisfies the Lipschitz bound
\[
 |f(\x) - f(\y)| \leq L \TNorm{\x - \y} \quad \text{for all $\x, \y$.}
\]
Let $\vec{\varepsilon}$ be a Rademacher vector. For all $t \geq 0,$
\[
 \Probab{f(\vec{\varepsilon}) \geq \Expect{f(\vec{\varepsilon})} + Lt} \leq \expe^{-t^2/8}.
\]
\end{lemma}
Lemma~\ref{lemma:colnorm-tail-bound} follows immediately from the observation that the norm of any one column of $\matA \matD \matH \transp$ is a convex Lipschitz function of a Rademacher vector.
Consider the norm of the $j$th column of $\matA \matD \matH \transp$ as a function of $\vec{\varepsilon},$ where $\matD = \diag{\vec{\varepsilon}}:$
\[
 f_j(\mat{\varepsilon}) = \|\matA \matD \matH \transp \e_j\| = \TNorm{\matA \diag{\vec{\varepsilon}} \h_j} = \TNorm{\matA \diag{\h_j} \vec{\varepsilon}},
\]
where $\h_j$ denotes the $j$th column of $\matH \transp.$
Evidently $f_j$ is convex. Furthermore,
\[
 |f_j(\x) - f_j(\y)| \leq \TNorm{\matA \diag{\h_j} (\x - \y)} \leq \TNorm{\matA} \TNorm{\diag{\h_j}} \TNorm{\x - \y} = \frac{1}{\sqrt{n}} \TNorm{\matA} \TNorm{\x - \y},
\]
where we used the triangle inequality and the fact that $\TNorm{\diag{\h_j}} = \INorm{\h_j} = \frac{1}{\sqrt{n}}.$ Thus $f_j$ is convex and Lipschitz with Lipschitz constant at most $\TNorm{\matA}/\sqrt{n}.$

We calculate
\begin{eqnarray*}
 \Expect{f_j(\varepsilon)} \leq \Expect{ f_j(\varepsilon)^2}^{1/2}
 &=& \left[ \Trace{\matA \diag{\h_j} \Expect{\vec{\varepsilon} \vec{\varepsilon}^\star} \diag{\h_j} \transp \matA \transp} \right]^{1/2}\\
 &=& \left[\Trace{\frac{1}{n} \matA\matA \transp} \right]^{1/2} \\
 &=& \frac{1}{\sqrt{n}} \FNorm{\matA}.
\end{eqnarray*}

It now follows from Lemma~\ref{prop:rademacher-concentration} that, for all $j=1,2,\ldots,n,$ the norm of the $j$th column of $\matA \matD \matH \transp$ satisfies the tail bound
\[
 \Probab{ \TNorm{\matA \matD \matH \transp \e_j} \geq \frac{1}{\sqrt{n}} \FNorm{\matA} + \frac{t}{\sqrt{n}} \TNorm{\matA} } \leq \expe^{-t^2/8}.
\]
Taking a union bound over all columns of $\matA\matD\matH \transp,$ we conclude that
\[
 \Probab{ \max\nolimits_{j=1,\ldots,n} \TNorm{(\matA \matD \matH \transp)^{(j)}} \geq \frac{1}{\sqrt{n}} \FNorm{\matA} + \frac{t}{\sqrt{n}} \TNorm{\matA} } \leq n \cdot \expe^{-t^2/8}.
\]
\end{proof}

As an interesting aside, we note that just as Lemma~\ref{lemma:SRHT-preserves-geometry}, which states that the SRHT essentially preserves the singular value of matrices with orthonormal rows and an aspect ratio of $k/n$, follows from Lemma~\ref{prop:SRHT-equalizes-columns-of-orthonormal-matrices}, Lemma~\ref{lemma:colnorm-tail-bound} implies that the SRHT essentially preserves the singular values of general rectangular matrices with the same aspect ratio. This can be shown using, e.g., the results on the effects of column sampling on the singular values of matrices from~\cite[Section 6]{TailBounds}.

\subsubsection{SRHT preserves the spectral norm}
The following lemma shows that even if the aspect ratio is larger than $k/n,$ the SRHT does not substantially increase the spectral norm of a matrix.

\begin{lemma}[SRHT-based subsampling in the spectral norm]
\label{lemma:spectral-SRHT-subsampling}
Let $\matA \in \R^{m \times n}$ have rank $\rho$ and $n$ be a power of 2. For some $r < n$, let $\matTh \in \R^{r \times n}$ be an SRHT matrix. Fix a failure probability $0 < \delta < 1.$ Then,
\[
\Probab{\TNormS{\matA \matTh \transp} \le 5 \TNormS{\matA}
+ \frac{\ln (\rho/\delta)}{r} \left( \FNorm{\matA} + \sqrt{8 \ln(n/\delta)} \TNorm{\matA} \right)^2 }
\geq 1-2\delta.
\]
\end{lemma}

To establish Lemma \ref{lemma:spectral-SRHT-subsampling}, we use the following Chernoff bound for sampling matrices without replacement.

\begin{lemma}[Matrix Chernoff bound, Theorem 2.2 in~\cite{Tro11}; see also Corollary in~\cite{Tropp-user-friendly}]
\label{prop:matrix-chernoff-bound}
 Let $\mathcal{X}$ be a finite set of positive-semidefinite matrices with dimension $k,$ and suppose that
\[
 \max_{\matX \in \mathcal{X}} \lambdamax{\matX} \leq B.
\]
Sample $\{\matX_1, \ldots, \matX_r\}$ uniformly at random from $\mathcal{X}$ without replacement. Compute
\[
 \mu_{\mathrm{max}} = r \cdot \lambdamax{\Expect{\matX_1}}.
\]
Then
\[
 \Probab{\lambdamax{\sum\nolimits_j \matX_j} \geq (1+\nu)\mu_{\mathrm{max}} } \leq k \cdot \bigg[ \frac{\expe^\nu}{(1+\nu)^{1+ \nu}} \bigg]^{\mu_{\mathrm{max}}/B} \quad \text{ for $\nu \geq 0.$}
\]
\end{lemma}

\begin{proof}[Proof of Lemma~\ref{lemma:spectral-SRHT-subsampling}]

Write the SVD of $\matA$ as $\matU \matSig \matV\transp$ where $\matSig \in \R^{\rho \times \rho}$ and observe that the spectral norm of $\mat{A}\matTh\transp$ is the same as that of $\matSig \matV\transp\matTh\transp.$

We control the norm of $\matSig \matV\transp \matTh\transp$ by considering the maximum singular value of its Gram matrix. Define $\matM = \matSig \matV \transp \matD \matH \transp$ and let $\matG$ be the Gram matrix of $\matM \matR \transp:$
\[
 \matG = \matM \matR \transp (\matM \matR \transp) \transp.
\]
 Evidently
\begin{equation}
\label{eqn:gram-identity}
 \lambdamax{\matG} = \frac{r}{n} \TNormS{\matSig \matV \transp \matTh \transp}.
\end{equation}

Recall that $\matM^{(j)}$ denotes the $j$th column of $\matM.$ If we denote the random set of $r$ coordinates to which $\matR$ restricts by $T$, then
\[
 \matG = \sum\nolimits_{j \in T} \matM^{(j)} \big(\matM^{(j)}\big) \transp.
\]
Thus $\matG$ is a sum of $r$ random matrices $\matX_1, \ldots, \matX_r$ sampled without replacement from the set $\mathcal{X} = \{ \matM^{(j)} \big(\matM^{(j)}\big) \transp \,:\, j=1,2,\ldots,n\}.$
There are two sources of randomness in $\matG$: $\matR$ and the Rademacher random variables on the diagonal of $\matD.$

Set
\[
 B = \frac{1}{n} \left( \FNorm{\matSig} + \sqrt{8 \ln(n/\delta)} \TNorm{\matSig} \right)^2
 \]
and let $E$ be the event
\[
 \max\nolimits_{j=1,\ldots,n}\TNormS{\matM^{(j)}} \leq B.
\]
When $E$ holds, for all $j =1,2,\ldots,n,$
\[
 \lambdamax{\matM^{(j)} \big(\matM^{(j)}\big) \transp} = \TNormS{\matM^{(j)}} \leq B,
\]
so $\matG$ is a sum of random positive-semidefinite matrices each of whose norms is bounded by $B.$
Note that whether or not $E$ holds is determined by $\matD$, and independent of $\matR.$

Conditioning on $E$, the randomness in $\matR$ allows us to use the matrix Chernoff bound of Lemma~\ref{prop:matrix-chernoff-bound} to control the maximum eigenvalue of $\matG.$
We observe that
\[
 \mu_{\text{max}} = r \cdot \lambdamax{\Expect{\matX_1}} = \frac{r}{n} \lambdamax{\sum\nolimits_{j=1}^n \matM^{(j)} \big(\matM^{(j)}\big)\transp} = \frac{r}{n} \TNormS{\matSig}.
\]
Take the parameter $\nu$ in Lemma~\ref{prop:matrix-chernoff-bound} to be
\[
 \nu = 4 + \frac{B}{\mu_{\text{max}}} \ln (\rho/\delta)
\]
to obtain the relation
\begin{align*}
 \Probab{ \lambdamax{\matG} \geq 5 \mu_{\text{max}} + B \ln (\rho/\delta)\, \mid\, E } & \leq (\rho-k) \cdot \expe^{[\delta - (1+\nu) \ln (1 + \nu)] \frac{\mu_{\text{max}}}{B}} \\
 & \leq \rho \cdot \expe^{\left(1-\tfrac{5}{4}\ln 5\right) \delta \frac{\mu_{\text{max}}}{B}} \\
 & \leq \rho \cdot \expe^{-\left(\tfrac{5}{4} \ln 5 -1\right) \ln (\rho/\delta)} < \delta.
\end{align*}
The second inequality holds because $\nu \geq 4$ implies that $(1 + \nu) \ln(1+\nu) \geq \nu\cdot \tfrac{5}{4}\ln 5 .$

We have conditioned on $E,$ the event that the squared norms of the columns of $\matM$ are all smaller than $B.$ By Lemma~\ref{lemma:colnorm-tail-bound}, $E$ occurs with probability at least $1-\delta.$ Thus, substituting the values of $B$ and $\mu_{\text{max}},$ we find that
\[
 \Probab{ \lambdamax{\matG} \geq \frac{r}{n} \left(5 \TNormS{\matSig} + \frac{\ln (\rho/\delta)}{r} \left( \FNorm{\matSig} + \sqrt{8 \ln(n/\delta)} \TNorm{\matSig} \right)^2 \right) }
\leq 2 \delta.
\]
Use equation~\eqref{eqn:gram-identity} to wrap up.
\end{proof}

\subsubsection{SRHT preserves the Frobenius norm}
Similarly, the SRHT is unlikely to substantially increase the Frobenius norm of a matrix.
\begin{lemma}[SRHT-based subsampling in the Frobenius norm]
\label{lemma:frobenius-SRHT-subsampling}
 Let $\matA \in \R^{m \times n}$ ($n$ is a power of 2) and let $\matTh \in \R^{r \times n}$ be an SRHT matrix for some $r < n.$ Fix a failure probability $0 < \delta < 1.$ Then, for any $\eta \geq 0,$
\[
 \Probab{\FNormS{\matA \matTh\transp} \leq (1 + \eta) \FNormS{\matA} } \geq 1- \left[ \frac{\expe^\eta}{(1+\eta)^{1+\eta}} \right]^{r/\big(1 + \sqrt{8 \ln(n/\delta)}\big)^2} - \delta.
\]
\end{lemma}

\begin{proof}
Let $c_j = \frac{n}{r} \TNormS{(\matA \matD \matH\transp)_j}$ denote the squared norm of the $j$th column of $\sqrt{n/r} \cdot \matA \matD \matH\transp$. Then, since right multiplication by $\matR\transp$ samples columns uniformly at random without replacement,
\begin{equation}
\label{eqn:frobenius-norm-upperbound}
 \FNormS{\matA \matTh\transp} = \frac{n}{r} \FNormS{\matA \matD \matH\transp \matR\transp } = \sum\nolimits_{i=1}^r X_i
\end{equation}
where the random variables $X_i$ are chosen randomly without replacement from the set $\{c_j\}_{j=1}^n.$ There are two independent sources of randomness in this sum: the choice of summands, which is determined by $\matR$, and the magnitudes of the $\{c_j\}$, which is determined by $\matD.$

To bound this sum, we first condition on $\matD$ being such that each $c_j$ is bounded by a quantity $B.$ Call this event $E$, then
\[
\Probab{\sum\nolimits_{i=1}^r X_i \geq (1 + \eta) \sum\nolimits_{i=1}^r \E{X_i} } \leq \Probab{\sum\nolimits_{i=1}^r X_i \leq (1 + \eta) \sum\nolimits_{i=1}^r \E{X_i} \,\mid\, E} + \Probab{E^c}.
\]
To select $B,$ we observe that Lemma~\ref{lemma:colnorm-tail-bound} implies that with probability $1 - \delta,$ the entries of $\matD$ are such that
\[
 \max\nolimits_j c_j \leq \frac{n}{r} \cdot \frac{1}{n} (\FNorm{\matA} + \sqrt{8 \ln(n/\delta)} \TNorm{\matA})^2 \leq  \frac{1}{r} (1 + \sqrt{8 \ln(n/\delta)})^2 \FNormS{\matA}.
\]
Accordingly, we take
\[
 B = \frac{1}{r}(1 + \sqrt{8 \ln(n/\delta)})^2 \FNormS{\matA},
\]
thereby arriving at the bound
\begin{equation}
\label{eqn:condfrobbound}
\Probab{\sum\nolimits_{i=1}^r X_i \geq (1 + \eta) \sum\nolimits_{i=1}^r \E{X_i} } \leq \Probab{\sum\nolimits_{i=1}^r X_i \leq (1 + \eta) \sum\nolimits_{i=1}^r \E{X_i} \,\mid\, E} + \delta.
\end{equation}

After conditioning on $\matD$, we observe that the randomness remaining on the righthandside of Eqn.~(\ref{eqn:condfrobbound}) is the choice of the summands $X_i,$ which is determined by $\matR.$ We address this randomness by applying a scalar Chernoff bound (Lemma~\ref{prop:matrix-chernoff-bound} with $k=1$). To do so, we need $\mu_{\text{max}},$ the expected value of the sum; this is an elementary calculation:
\[
 \E{X_1} = n^{-1} \sum\nolimits_{j=1}^n c_j = \frac{1}{r} \FNormS{\matA},
\]
so $\mu_{\text{max}} = r \E{X_1} = \FNormS{\matA}.$

Applying Lemma~\ref{prop:matrix-chernoff-bound} conditioned on $E,$  we conclude that
\[
 \Probab{\FNormS{\matA \matTh\transp} \geq (1 + \eta) \FNormS{\matA}\,\mid\, E } \leq \left[ \frac{\expe^\eta}{(1+\eta)^{1+\eta}} \right]^{r/(1 + \sqrt{8 \ln(n/\delta)})^2} + \delta
\]
for $\eta \geq 0.$
\end{proof}

\subsubsection{SRHT preserves matrix multiplication}
Finally, we prove a novel result on approximate matrix multiplication involving SRHT matrices.
\begin{lemma}[SRHT for approximate matrix multiplication]\label{lem:mm}
Let $\matA \in \R^{m \times n}$, $\matB \in \R^{n \times p}$, and $n$ be a power of 2.
For some $r < n$, let $\matTh \in \R^{r \times n}$ be an SRHT matrix. Fix a failure probability $0 < \delta < 1.$
Assume $\const{R}$ satisfies
$
 0 \leq \const{R} \leq \frac{\sqrt{r}}{1+\sqrt{8\ln(n/\delta)}}.
$
Then,
\[
 \Probab{ \FNorm{\matA \matTh\transp \matTh \matB - \matA \matB}  \leq 2 (\const{R} + 1)
 \frac{\FNorm{\matA}\FNorm{\matB} + \sqrt{8 \ln(n/\delta)}\FNorm{\matA}\TNorm{\matB}}{\sqrt{r}} } \geq 1 - \expe^{-\const{R}^2/4} - 2\delta.
\]
\end{lemma}

{\bf Remark.} Recall that the stable rank $\stablerank{\matA} = \FNormS{\matA}/\TNormS{\matA}$ reflects the decay of the spectrum of the matrix $\matA.$ Lemma~\ref{lem:mm} can be rewritten as a bound on the relative error of the approximation $\matA \matTh\transp \matTh \matB$ to the product $\matA \matB:$
\[
\frac{\FNorm{\matA \matTh\transp \matTh \matB - \matA \matB}}{\FNorm{\matA \matB}} \leq \frac{\FNorm{\matA}\FNorm{\matB}}{\FNorm{\matA\matB}} \cdot \frac{R+2}{\sqrt{r}} \cdot \left(1 + \frac{\sqrt{8 \ln(n/\delta)}}{\stablerank{\matB}} \right).
\]
In this form, we see that the relative error is controlled by the deterministic condition number for the matrix multiplication problem as well as the stable rank of $\mat{B}$ and the number of column samples $r.$ Since the roles of $\matA$ and $\matB$ in this bound can be interchanged, in fact we have the bound
\[
\frac{\FNorm{\matA \matTh\transp \matTh \matB - \matA \matB}}{\FNorm{\matA \matB}} \leq \frac{\FNorm{\matA}\FNorm{\matB}}{\FNorm{\matA\matB}} \cdot \frac{R+2}{\sqrt{r}} \cdot \left(1 + \frac{\sqrt{8 \ln(n/\delta)}}{\max(\stablerank{\matB}, \stablerank{\matA})} \right),
\]

\subsection*{Proof of Lemma~\ref{lem:mm}}
To prove the Lemma, we first develop a generic result for approximate matrix multiplication via uniform sampling (without replacement)
of the columns and the rows of the two matrices involved in the product (see Lemma~\ref{lemma:matrix-multiplication} below).
Lemma~\ref{lem:mm} is a simple instance of this generic result.
We mention that Lemma 3.2.8 in~\cite{Dri02} gives a
similar result for approximate matrix multiplication, which, however gives a bound for the expected value of the error term, while our Lemma~\ref{lem:mm}
gives a comparable bound which holds with high probability. To prove Lemma~\ref{lemma:matrix-multiplication},
we use the following vector Bernstein inequality for sampling without replacement in Banach spaces; this result follows directly from a similar inequality for sampling with replacement established by Gross in~\cite{Gross11}.

\begin{lemma}
 \label{lemma:vector-bernstein}
  Let $\mathcal{V}$ be a collection of $n$ vectors in a normed space with norm $\VTNorm{\cdot}.$ Choose $\vec{V}_1, \ldots, \vec{V}_r$ from $\mathcal{V}$ uniformly at random \emph{without} replacement. Also choose $\vec{V}_1^\prime, \ldots, \vec{V}_r^\prime$ from $\mathcal{V}$ uniformly at random \emph{with} replacement. Let
\[
 \mu = \E{\VTNorm{\sum\nolimits_{i=1}^r (\vec{V}_i^\prime - \E{\vec{V}_i^\prime})}}
\]
and set
\[
\sigma^2 \geq 4r\E{\VTNormS{\vec{V}_1^\prime}} \quad \text{ and } \quad B \geq 2 \max_{\vec{V} \in \mathcal{V}} \VTNorm{\vec{V}}.
\]
If $ 0 \leq t \leq \sigma^2/B,$ then
\[
 \Probab{\VTNorm{\sum\nolimits_{i=1}^r \vec{V}_i - r\E{\vec{V}_1}} \geq \mu+t} \leq \mathrm{exp}\left( -\frac{t^2}{4 \sigma^2} \right).
\]
\end{lemma}

\begin{proof}
 We proceed by developing a bound on the moment generating function (mgf) of
\[
\VTNorm{\sum\nolimits_{i=1}^r \vec{V}_i - r\E{\vec{V}_1}} - \mu.
\]
This mgf is controlled by the mgf of a similar sum where the vectors are sampled with replacement. That is, for $\lambda \geq 0,$
\begin{equation}
\label{eqn:mgfineq}
 \E{\mathrm{exp}\left(\lambda \cdot \VTNorm{\sum\nolimits_{i=1}^r \vec{V}_i - r\E{\vec{V}_1}} - \lambda \mu\right)} \leq
\E{\mathrm{exp}\left(\lambda \cdot \VTNorm{\sum\nolimits_{i=1}^r \vec{V}_i^\prime - r\E{\vec{V}_1}} - \lambda \mu\right)}.
\end{equation}
This follows from a classical observation due to Hoeffding \cite{Hoe63} (see also \cite{GN10} for a more modern exposition) that for any convex $\R$-valued function $g,$
\[
 \E{g\left( \sum\nolimits_{i=1}^r \vec{V}_i \right)}\leq \E{g\left(\sum\nolimits_{i=1}^r \vec{V}_i^\prime \right)}.
\]
Specifically, take $g(\vec{V}) = \mathrm{exp}\left(\lambda\VTNorm{\vec{V} - r \E{\vec{V}_1}} - \lambda \mu\right)$ to obtain the asserted inequality of mgfs.

In the proof of Theorem 12 in \cite{Gross11}, Gross establishes that any random variable $Z$ whose mgf is less than the righthand side of Eqn.~(\ref{eqn:mgfineq}) satisfies a tail inequality of the form
\begin{equation}
\label{eqn:grosstail}
 \Probab{ Z \geq \mu + t } \leq \mathrm{exp}\left( -\frac{t^2}{4s^2} \right)
\end{equation}
when $t \leq s^2/M,$ where
\[
s^2 \geq \sum_{i=1}^r \E{\VTNorm{\vec{V}_i^\prime - \E{\vec{V}_1^\prime} }^2}
\]
and $M$ almost surely bounds $\VTNorm{\vec{V}_i^\prime - \E{\vec{V}_1^\prime}}$ for all $i=1,\ldots,r.$
 To apply this result, note that for all $i=1,\ldots,r,$
 \[
  \VTNorm{\vec{V}_i^\prime - \E{\vec{V}_1^\prime}} \leq 2 \max_{\vec{V} \in \mathcal{V}} \VTNorm{\vec{V}} = B.
 \]
Also take $\vec{V}_1^{\prime\prime}$ to be an i.i.d. copy of $\vec{V}_1^\prime$ and observe that, by Jensen's inequality,
 \begin{align*}
  \sum_{i=1}^r \E{\VTNorm{\vec{V}_i^\prime - \E{\vec{V}_1^\prime} }^2} & = r \E{\VTNorm{\vec{V}_1^\prime - \E{\vec{V}_1^\prime} }^2} \\
  & \leq r \E{\VTNormS{\vec{V}_1^\prime - \vec{V}_1^{\prime\prime}} } \leq r \E{ (\VTNorm{\vec{V}_1^\prime} + \VTNorm{\vec{V}_1^{\prime \prime}})^2} \\
  & \leq 2 r \E{ \VTNormS{\vec{V}_1^\prime} + \VTNormS{\vec{V}_1^{\prime\prime}} } \\
  & = 4 r \E{ \VTNormS{\vec{V}_1^\prime} } \leq \sigma^2.
 \end{align*}

The bound given in the statement of Lemma~\ref{lemma:vector-bernstein} follows from taking $s^2 = \sigma^2$ and $M = B$ in Eqn.~(\ref{eqn:grosstail}).
\end{proof}

This vector Bernstein inequality gives us a tail bound on the Frobenius error of a simple approximate matrix multiplication scheme based upon column and row sampling.

\begin{lemma}[Matrix Multiplication]
\label{lemma:matrix-multiplication}
Let $\matX \in \R^{m \times n}$ and $\matY \in \R^{n \times \ell}$. Fix $r \leq n$. Select uniformly at random and without replacement $r$ columns from $\matX$ and the corresponding rows from $\matY$ and multiply the selected columns and rows with $\sqrt{n/r}$. Let $\hat{\matX} \in \R^{m \times r}$ and $\hat{\matY} \in \R^{r \times \ell}$ contain the selected columns and rows, respectively.
Choose
\[
\sigma^2 \geq \frac{4 n}{r} \sum_{i=1}^n \TNormS{\vec{X}^{(i)}} \TNormS{\vec{Y}_{(i)}} \quad \text{and} \quad B \geq \frac{2 n}{r} \max_i \TNorm{\vec{X}^{(i)}} \TNorm{\vec{Y}_{(i)}}.
\]
Then if $ 0 \leq t \leq \sigma^2/B,$
\[
\Probab{\FNorm{\hat{\matX} \hat{\matY} - \matX \matY} \geq t + \sigma } \leq \mathrm{exp}\left( -\frac{t^2}{4\sigma^2}\right).
\]
\end{lemma}
\begin{proof}
 Let $\mathcal{V}$ be the collection of vectorized rank-one products of columns of $\sqrt{n/r}\cdot\matX$ and rows of $\sqrt{n/r}\cdot\matY.$ That is, take
\[
 \mathcal{V} = \bigg\{ \frac{n}{r} \text{vec}(\vec{X}^{(i)} \vec{Y}_{(i)}) \bigg\}_{i=1}^n.
\]
Sample $\vec{V}_1, \ldots, \vec{V}_r$ uniformly at random from $\mathcal{V}$ without replacement, and observe that $\E{\vec{V}_i} = \frac{1}{r} \text{vec}(\matX \matY).$
With this notation, the quantities $\FNorm{\hat{\matX} \hat{\matY} - \matX \matY}$ and
\[
 \TNormB{\sum\nolimits_{i=1}^r (\vec{V}_i - \E{\vec{V}_i})},
\]
have the same distribution, therefore any probabilistic bound developed for the latter holds for the former. The conclusion of the lemma follows from applying
Lemma~\ref{lemma:vector-bernstein} to bound the second quantity.

We calculate the variance-like term in Lemma~\ref{lemma:vector-bernstein}, $4 r \E{\TNormS{\vec{V}_1}}:$
\[
 4 r \E{\TNormS{\vec{V}_1}} = 4r \frac{1}{n} \sum_{i=1}^n \frac{n^2}{r^2} \TNormS{\vec{X}^{(i)}} \TNormS{\vec{Y}_{(i)}} = 4\frac{n}{r} \sum_{i=1}^n \TNormS{\vec{X}^{(i)}} \TNormS{\vec{Y}_{(i)}} \leq \sigma^2.
\]

Now we consider the expectation
\[
 \mu  =  \E{\TNormB{\sum\nolimits_{i=1}^r (\vec{V}_i^\prime - \E{\vec{V}_i^\prime})}}.
\]
In doing so, we will use the notation $\condE{A,B,\ldots}{C}$ to denote the conditional expectation of a random variable $C$ with respect to the random variables $A,B,\ldots.$
Recall that a Rademacher vector is a random vector whose entries are independent and take the values $\pm 1$ with equal probability. Let $\vec{\varepsilon}$ be a Rademacher vector of length $r$ and sample $\vec{V}_1^\prime, \ldots, \vec{V}_r^\prime$ and $\vec{V}_1^{\prime\prime}, \ldots, \vec{V}_r^{\prime\prime}$ uniformly at random from $\mathcal{V}$ with replacement. Now $\mu$ can be bounded as follows:
\begin{eqnarray*}
 \mu  & =  & \E{\TNormB{\sum\nolimits_{i=1}^r (\vec{V}_i^\prime - \E{\vec{V}_i^\prime})}} \\
      &\leq& \condE{\{\vec{V}_i^\prime\}, \{\vec{V}_i^{\prime\prime}\}}{\TNormB{\sum\nolimits_{i=1}^r (\vec{V}_i^\prime - \vec{V}_i^{\prime\prime})}} \\
      & =  & \condE{\{\vec{V}_i^\prime\}, \{\vec{V}_i^{\prime\prime}\}, \vec{\varepsilon}}{\TNormB{\sum\nolimits_{i=1}^r \varepsilon_i (\vec{V}_i^\prime - \vec{V}_i^{\prime\prime})}}\\
      &\leq& 2 \condE{\{\vec{V}_i^\prime\}, \vec{\varepsilon}}{\TNormB{\sum\nolimits_{i=1}^r \varepsilon_i \vec{V}_i^\prime}} \\
      &\leq& 2 \sqrt{ \condE{\{\vec{V}_i^\prime\}, \vec{\varepsilon}}{\TNormBS{\sum\nolimits_{i=1}^r \varepsilon_i \vec{V}_i^\prime}} }\\
      & = & 2 \sqrt{ \condE{\{\vec{V}_i^\prime\}}{\condE{\vec{\varepsilon}}{\sum\nolimits_{i,j=1}^r \varepsilon_i \varepsilon_j {\vec{V}_i^\prime}\transp \vec{V}_j^\prime} }}\\
      &=& 2 \sqrt{\E{ \sum\nolimits_{i=1}^r \TNormS{\vec{V}_i^\prime}} }.
\end{eqnarray*}
The first inequality is Jensen's, and the following equality holds because the components of the sequence $\{\vec{V}_i^\prime - \vec{V}_i^{\prime\prime}\}$ are symmetric and independent. The next two manipulations are the triangle inequality and Jensen's inequality. This stage of the estimate is concluded by conditioning and using the orthogonality of the Rademacher variables. Next, the triangle inequality and the fact that $\E{\TNormS{\vec{V}_1^\prime}} = \E{\TNormS{\vec{V}_1}}$ allow us to further simplify the estimate of $\mu:$
\[
 \mu \leq 2 \sqrt{\E{ \sum\nolimits_{i=1}^r \TNormS{\vec{V}_i}} } = 2 \sqrt{r \E{\TNormS{\vec{V}_1}}} \leq \sigma.
\]
We also calculate the quantity
\[
2 \max_{\vec{V} \in \mathcal{V}} \TNorm{\vec{V}} = \frac{2n}{r} \max_i \TNorm{\vec{X}^{(i)}}\TNorm{\vec{Y}_{(i)}} \leq B.
\]

The stipulated tail bound follows from applying Lemma~\ref{lemma:vector-bernstein} with our estimates for $B$, $\sigma^2,$ and $\mu.$
\end{proof}

Lemma~\ref{lem:mm} now follows from this result on matrix multiplication.

\begin{proof} (of Lemma~\ref{lem:mm})
 Let $\matX = \matA \matD \matH\transp$ and $\matY = \matH \matD \matB$ and form $\hat{\matX}$ and $\hat{\matY}$ according to Lemma~\ref{lemma:matrix-multiplication}. Then, $\matX \matY = \matA \matB$ and
\[
 \FNorm{\matA \matTh\transp \matTh \matB - \matA \matB} = \FNorm{\hat{\matX} \hat{\matY} - \matX \matY}.
\]
To apply Lemma~\ref{lemma:matrix-multiplication}, we first condition on the event that the SRHT equalizes the column norms of our matrices. Namely, we observe that, from Lemma~\ref{lemma:colnorm-tail-bound}, with probability at least $1 - 2\delta,$
\begin{align}
 \label{eqn:mm-norm-bounds}
 \max\nolimits_i \TNorm{\matX^{(i)}} & \leq \frac{1}{\sqrt{n}} (\FNorm{\matA} + \sqrt{8 \ln(n/\delta)} \TNorm{\matA}), \text{ and } \\
 \max\nolimits_i \TNorm{\matY_{(i)}} & \leq \frac{1}{\sqrt{n}} (\FNorm{\matB} + \sqrt{8 \ln(n/\delta)}\TNorm{\matB}). \notag
\end{align}

Conditioning on these nice interactions, we choose the parameters $\sigma$ and $B$ in Lemma~\ref{lemma:matrix-multiplication}. We first take
\begin{equation}
\label{eqn:sigmachoice}
\sigma^2 = \frac{4}{r} (\FNorm{\matB} + \sqrt{8 \ln(n/\delta)}\TNorm{\matB})^2 \FNormS{\matA}.
\end{equation}
Observe that because of~\eqref{eqn:mm-norm-bounds},
\[
 \sigma^2 = 4 \frac{n}{r} \cdot \frac{(\FNorm{\matY} + \sqrt{8 \ln(n/\delta)}\TNorm{\matY})^2}{n} \FNormS{\matX} \geq 4 \frac{n}{r} \sum\nolimits_{i=1}^n \TNormS{\matX^{(i)}} \TNormS{\matY_{(i)}}
\]
so this choice of $\sigma$ satisfies the inequality stipulated in Lemma~\ref{lemma:matrix-multiplication}. Next we choose
\[
 B = \frac{2}{r}(\FNorm{\matA} + \sqrt{8 \ln(n/\delta)} \TNorm{\matA})(\FNorm{\matB} + \sqrt{8 \ln(n/\delta)} \TNorm{\matB}).
\]
Again, because of~\eqref{eqn:mm-norm-bounds}, $B$ satisfies the stipulation $B \geq \frac{2 n}{r}\max_i \TNorm{\mat{X}^{(i)}} \TNorm{\mat{Y}_{(i)}}.$

For simplicity, let $\gamma = 8\ln(n/\delta).$ With these choices for $\sigma^2$ and $B,$
\begin{align*}
 \frac{\sigma^2}{B} & = \frac{ 2 \FNormS{\matA} (\FNorm{\matB} + \sqrt{\gamma} \TNorm{\matB})^2 }{ (\FNorm{\matA} + \sqrt{\gamma} \TNorm{\matA}) (\FNorm{\matB} + \sqrt{\gamma} \TNorm{\matB})} \\
 & \geq \frac{ 2 \FNormS{\matA} (\FNorm{\matB} + \sqrt{\gamma} \TNorm{\matB})^2 }{ (\FNorm{\matA} + \sqrt{\gamma} \FNorm{\matA}) (\FNorm{\matB} + \sqrt{\gamma} \TNorm{\matB})} \\
 & = \frac{2 \FNorm{\matA}(\FNorm{\matB} + \sqrt{\gamma}\TNorm{\matB})}{1 + \sqrt{\gamma}}.
\end{align*}
Now, referring to Eqn.~(\ref{eqn:sigmachoice}), identify the numerator as $\sqrt{r}\sigma$ to see that
\[
 \frac{\sigma^2}{B} \geq \frac{\sqrt{r} \sigma}{1+ \sqrt{8 \ln(n/\delta)}}.
\]

Apply Lemma~\ref{lemma:matrix-multiplication} to see that, when Eqns.~(\ref{eqn:mm-norm-bounds}) hold and $0 \leq R \sigma \leq \sigma^2/B,$
\[
 \Probab{\FNorm{\matA \matTh\transp\matTh \matB - \matA \matB} \geq (R + 1) \sigma }  \leq \mathrm{exp} \left( -\frac{R^2}{4} \right).
\]
From our lower bound on $\sigma^2/B,$ we know that the condition $R \sigma \leq \sigma^2/B$ is satisfied when
$$R \leq \sqrt{r}/(1 + \sqrt{8 \ln(n/\delta)}).$$ Also, we
established above that Eqns.~(\ref{eqn:mm-norm-bounds}) hold with probability at least $1 - 2\delta.$
From these two facts, it follows that when $0 \leq R \leq \sqrt{r}/(1 + \sqrt{8 \ln(n/\delta)}),$
\[
 \Probab{\FNorm{\matA \matTh\transp\matTh \matB - \matA \matB} \geq (R + 1) \sigma }  \leq \mathrm{exp} \left( -\frac{R^2}{4} \right) + 2 \delta.
\]

The tail bound given in the statement of Lemma~\ref{lem:mm} follows from substituting our estimate of $\sigma.$
\end{proof}

\section{Proofs of our main Theorems}\label{sec:proofs}

\subsection{Preliminaries}

To prove Theorem~\ref{thm:quality-of-approximation-guarantee} we first need some background on restricted (within a subspace) low-rank
matrix approximations. Let $\matA \in \mathbb{R}^{m \times n}$, let $k < n$ be an integer, and let
$\matY \in \mathbb{R}^{m \times r}$ with $ r > k$
(the case $m = r$ corresponds to the standard unrestricted
low rank approximation problem which can be addressed via the SVD).
We call $\Pi_{\matY,k}^\xi(\matA) \in \mathbb{R}^{m \times n}$  the best rank
\math{k} approximation to \math{\matA} in the column space of \math{\matY}, with respect to the $\xi$ norm ($\xi=2$ or $\xi = \mathrm{F}$).
Formally, for fixed $\xi$, we can write
$\Pi_{\matY,k}^\xi(\matA) = \matY\matX^\xi$, where
$$
\matX^\xi = \argmin_{\matX \in {\R}^{r \times n}:\rank(\matX)\leq k}\XNormS{\matA-
\matY\matX}.
$$
In order to compute (or approximate) $\Pi_{\matY,k}^{\xi}(\matA)$ we will use the following 3-step procedure:
\begin{center}
\begin{algorithmic}[1]
\STATE Let $\ell = \min\{m,r\}.$ Use an SVD to construct a matrix $\matQ \in \R^{m \times \ell}$ that satisfies $\matQ\transp \matQ = \matI_\ell$ and spans the range of $\matY.$
This construction takes $\const{O}(m r \ell)$ time.
\STATE Compute $\matX_{opt} = \argmin_{\matX \in \R^{\ell \times n},\,\, \rank(\matX) \le k}\FNorm{ \matQ\transp \matA - \matX }$ in
\math{\const{O}(mn\ell+ n\ell^2)}  time. In fact, since $\ell \leq m,$ we see that $\matX_{opt}$ can be computed in $\const{O}(mn\ell)$ time.
%
\STATE Return 
$\matQ\matX_{opt} \in \mathbb{R}^{m \times n}$ in $\const{O}(mn\ell)$ time.
\end{algorithmic}
\end{center}
$\matQ\matX_{opt}$ is a matrix of rank at most $k$ that lies within the column span
of $\matY$. Note that though  $\Pi_{\matY,k}^{\xi}(\matA)$ can depend on
\math{\xi}, the algorithm above computes the same matrix, independent of \math{\xi}.
The following result, which appeared as Lemma 18 in~\cite{BDM11a},
proves that this algorithm computes $\Pi_{\matY,k}^{\mathrm{F}}(\matA)$ and a constant factor approximation to $\Pi_{\matY,k}^{2}(\matA)$.
\begin{lemma}\label{lem:bestF}[Lemma 18 in~\cite{BDM11a}]
Given $\matA \in {\R}^{m \times n}$, $\matY\in\R^{m\times r}$,
and an integer $k \le r,$ the matrix
$\matQ\matX_{opt} \in \mathbb{R}^{m \times n}$ described above satisfies:
\begin{eqnarray*}
\norm{\matA-\matQ\matX_{opt}}_{\mathrm{F}}^2 &=& \FNormS{\matA-\Pi_{\matY,k}^{\mathrm{F}}(\matA)},\\
\norm{\matA-\matQ\matX_{opt}}_2^2 &\leq& 2\TNormS{\matA-\Pi_{\matY,k}^{2}(\matA)}.
\end{eqnarray*}
\end{lemma}
The discussion above the Lemma shows that $\matQ \matX_{opt}$ can be computed in $\const{O}(m n \ell + m r \ell)$ time.

\subsubsection{Matrix Pythagoras and generalized least-squares regression}
Lemma~\ref{lem:pyth} is the analog of the Pythagoras theorem in the matrix setting. A proof of this lemma can be found in~\cite{BDM11a}.
Lemma~\ref{lem:genreg} is an immediate corollary of Matrix-Pythogoras.
\begin{lemma}\label{lem:pyth}
If \math{\matX,\matY\in\R^{m\times n}} and
\math{\matX\matY\transp=\bm{0}_{m \times m}} or \math{\matX\transp\matY=\bm{0}_{n \times n}}, then for both $\xi = 2, \mathrm{F}:$
\eqan{
\XNorm{\matX+\matY}^2 \le \XNorm{\matX}^2+\XNorm{\matY}^2.}
\end{lemma}

\begin{lemma}\label{lem:genreg}
Given $\matA \in \R^{m \times n}$, $\matC \in \R^{m \times r}$, and for all  \math{\matX \in\R^{r \times n} } and  for both $\xi = 2, \mathrm{F}:$
$$\XNormS{ \matA- \matC\matC^+ \matA} \le \XNormS{ \matA - \matC \matX }.$$
\end{lemma}
\begin{proof}
Write \math{\matA-\matC\matX=(\matI-\matC\matC^+)\matA+\matC(\matC^+ \matA
-\matX)}. Observe that \math{((\matI-\matC\matC^+)\matA)\transp
\matC(\matC^+ \matA)=\bm0_{n \times n}}. By Lemma~\ref{lem:pyth},
\math{\XNormS{\matA-\matC\matX}       \ge
\XNormS{(\matI-\matC\matC^+)\matA} + \XNormS{\matC_1(\matC^+ \matA
-\matX)} \ge
\XNormS{(\matI-\matC\matC^+)\matA}
}.
\end{proof}

\subsubsection{Low-rank matrix approximation based on projections}\label{sec:presth}

The low-rank matrix approximation algorithm investigated in this paper is an instance of a wider class of low-rank approximation schemes wherein a matrix is projected onto a subspace spanned by some linear combination of its columns. The problem of providing a general framework for studying the error of such projection schemes is well studied~\cite{BMD09a,HMT,BDM11a}. The following result appeared as Lemma 7 in~\cite{BMD09a} (see also Theorem 9.1 in~\cite{HMT}).
\begin{lemma}
\label{prop:structural-result}[Lemma 7 in~\cite{BMD09a}]
Let $\matA \in \R^{m \times n}$ have rank $\rho.$ Fix $k$ satisfying $0 \leq k \leq \rho$.
Given a matrix $\matOmega \in \R^{n \times r}$, with $r \ge k$, construct $\matY = \matA \matOmega.$ If $\matV_k \transp \matOmega$ has full row-rank, then,
for $\xi=2, \mathrm{F}$,
\begin{equation}
\label{eqn:tropp-structural-result}
\XNormS{\matA - \matY \pinv{\matY} \matA}
\leq
\XNormS{\matA -  \Pi_{\matY,k}^{\xi}(\matA) }
\leq
\XNormS{ \matA - \matA_k } + \XNormS{\matSig_{\rho - k} \matV_{\rho-k} \transp \matOmega \pinv{\left( \matV_k \transp \matOmega \right)} }.
\end{equation}
\end{lemma}
This lemma provides an upper bound for the residual error of the low-rank matrix approximation obtained via projections.
We now prove a new result for the forward error.

\begin{lemma}
\label{prop:structural-result2}
Let $\matA \in \R^{m \times n}$ have rank $\rho.$ Fix $k$ satisfying $0 \leq k \leq \rho$.
Given a matrix $\matOmega \in \R^{n \times r}$, with $r \ge k$, construct $\matY = \matA \matOmega.$
If $\matV_k \transp \matOmega$ has full row-rank, then,
for $\xi=2, \mathrm{F}$,
\begin{equation}
\label{eqn:tropp-structural-result2a}
\XNormS{\matA_k - \matY \pinv{\matY} \matA} \leq  \XNormS{\matA - \matA_k} +
\XNormS{\matSig_{\rho - k} \matV_{\rho-k} \transp \matOmega \pinv{\left( \matV_k \transp \matOmega \right)} }.
\end{equation}
\end{lemma}
\begin{proof}
For both $\xi = 2, \mathrm{F}$,
\begin{eqnarray*}
\XNormS{\matA_k - \matY \pinv{\matY} \matA}
&=&  \XNormS{\matA_k  - \matY \pinv{\matY} \matA_k -\matY \pinv{\matY}\matA_{\rho-k} }  \\
&\le& \XNormS{\matA_k - \matY \pinv{\matY} \matA_k} +   \XNormS{ \matA_{\rho-k} } \\
&\le& \XNormS{\matA_k - \matY \pinv{(\matV_k\transp \matOmega)} \matV_k\transp} +   \XNormS{ \matA_{\rho-k} }\\
&=& \XNormS{\matA_k - \matU_k \matSig_k \matV_k\transp \matOmega \pinv{(\matV_k\transp \matOmega)} \matV_k\transp   + \matA_{\rho-k}\matOmega \pinv{(\matV_k\transp \matOmega)} \matV_k\transp} +   \XNormS{ \matA_{\rho-k} } \\
&=& \XNormS{ \matA_{\rho-k}\matOmega \pinv{(\matV_k\transp \matOmega)} \matV_k\transp} +   \XNormS{ \matA_{\rho-k} } \\
&\le& \TNormS{\matU_{\rho-k}}\XNormS{ \matSig_{\rho-k} \matV_{\rho-k}\matOmega \pinv{ (\matV_k\transp \matS)} } \TNormS{\matV_k\transp} +   \XNormS{ \matA_{\rho-k} }\\
&=& \XNormS{ \matSig_{\rho-k} \matV_{\rho-k}\matOmega \pinv{ (\matV_k\transp \matOmega)} } +   \XNormS{ \matA_{\rho-k} }
\end{eqnarray*}
In the above, in the first inequality we used Lemma~\ref{lem:pyth}
($ (\matA_k  - \matY \pinv{\matY} \matA_k) (-\matY \pinv{\matY}\matA_{\rho-k})\transp = {\bf 0}_{m \times m}$
because $\matA_k \matA_{\rho-k}\transp = {\bf 0}_{m \times m}$).
In the second inequality we used Lemma~\ref{lem:genreg} (with $\matX = \pinv{(\matV_k\transp \matOmega)} \matV_k\transp$).
In the third equality, we used the fact that $(\matV_k\transp \matOmega) (\matV_k\transp \matOmega)^+=\matI_{k}$, since, by assumption,
\math{\rank(\matV_k\transp \matOmega)=k}. In the last inequality we used the sub-multiplicativity property of the spectral and Frobenius norms,
i.e for any three matrices $\matX, \matY, \matZ$:
$$\XNormS{\matX \matY \matZ} \le \TNormS{\matX} \XNormS{\matY \matZ} \le \TNormS{\matX} \XNormS{\matY} \TNormS{\matZ}.$$
\end{proof}

\subsubsection{Least squares regression based on projections}

Similarly, one of the two SRHT least squares regression algorithms analyzed in this article is an instance of a wider class of approximation algorithms
where the dimensions of the input matrix and the vector of the regression problem are reduced via pre-multiplication with a random matrix.
Lemma 9 in~\cite{BDM12} provides a general framework for the analysis of such projection algorithms.
\begin{lemma} \label{prop3}
Let $\matA \in \R^{m \times n}$ ($m \ge n$) of rank $\rho$ and $\b \in \R^m$ be inputs to the least squares problem $\min_{\x \in \R^n}\TNorm{\matA \x -\b}$.
Let $\matU \in \R^{m \times \rho}$ contain the top $\rho$ left singular vectors of $\matA$ and let $\matOmega \in \R^{m \times r}$ ($\rho \le r \le m $) be a matrix
such that $\rank(\matOmega\transp \matU) = \rank(\matU)$. Then,
$$ \TNormS{ \matA \tilde{\x}_{opt} - \b  } \le \TNormS{ \matA \x_{opt} - \b  } + \TNormS{ \pinv{ \left( \matOmega\transp \matU \right) } \matOmega\transp  \left(\matA \x_{opt} - \b\right)  }.  $$
In the above, $\x_{opt} = \pinv{\matA} \b$ and $\tilde{\x}_{opt} = \pinv{ \left( \matOmega\transp \matA \right) } \matOmega\transp \b$.
\end{lemma}

The following lemma is a restatement of Lemma 2, along with Eqn.~(9) and Eqn.~(11) in~\cite{DMMS11}. It gives a bound on the forward error of the
approximation of a least-squares problem that is obtained via projections.
In~\cite{DMMS11} the parameters $\alpha$ and $\beta$ are fixed to $\alpha=1/\sqrt{2}$ and $\beta = \varepsilon/2$, for some parameter $0 < \varepsilon < 1$.
Showing the result for general $\alpha > 0$ and $\beta >$ is straightforward, hence a detailed proof is omitted.
\begin{lemma}[Lemma 2 in~\cite{DMMS11}]\label{lem:last}
Let $\matA \in \R^{m \times n}$ ($m \ge n$) of rank $\rho=n$ and $\b \in \R^m$ be inputs to the least squares problem $\min_{\x \in \R^n}\TNorm{\matA \x -\b}$.
Let $\matU \in \R^{m \times \rho}$ contain the top $\rho$ left singular vectors of $\matA$ and let $\matOmega \in \R^{m \times r}$ ($\rho \le r \le m $).
For some $\alpha>0$, and $\beta > 0$, assume that
\begin{equation}\label{eqn:first}
\sigma_{\min}\left( \matU\transp \matOmega \right) \ge \alpha^{\frac{1}{2}}
\end{equation}
and
\begin{equation}\label{eqn:second}
\TNormS{\matU\transp \matOmega \matOmega\transp \left( \matA \x_{opt} - \b \right)} \le \beta \TNormS{ \matA \x_{opt} - \b }.
\end{equation}
Furthermore, assume that there exists a $\gamma \in (0,1]$ such that $ \TNorm{\matU_{\matA} \matU_{\matA}\transp\b} \ge \gamma \TNorm{\b}$.
Then,
\begin{equation}\label{eqn:third}
 \TNorm{ \x_{opt} - \tilde{\x}_{opt} } \le \left( \frac{\alpha}{\beta} \right)^{\frac{1}{2}} \cdot \left( \kappa{ \left( \matA \right)} \sqrt{ \gamma^{-2}   - 1 } \right) \TNorm{\x_{opt}}.
\end{equation}
In the above, $\x_{opt} = \pinv{\matA} \b$ and $\tilde{\x}_{opt} = \pinv{ \left( \matOmega\transp \matA \right) } \matOmega\transp \b$.
\end{lemma}

\subsection{Proof of Theorem~\ref{thm:quality-of-approximation-guarantee}} \label{sec:guarantees}

\subsubsection{Frobenius norm bounds}

We first prove the Frobenius norm bounds in the theorem (i.e Eqns.~(i), (ii), (iii), and~(v)). We would like to apply Lemma~\ref{prop:structural-result} with $\matOmega = \matTh\transp \in \R^{n \times r}$ and $\xi=\mathrm{F}$. Notice that because of our assumption that
$$r \geq 6 \const{C}^2 \varepsilon^{-1} \big[\sqrt{k} + \sqrt{8 \ln(n/\delta)}\big]^2 \ln(k/\delta),$$
where $\const{C} > 1,$ Lemma~\ref{lemma:SRHT-preserves-geometry} implies that with probability at least $1-3\delta$,
$$\rank(\matV_k\transp \matTh\transp) = k;$$
so, for $\xi=\mathrm{F}$, Lemma~\ref{prop:structural-result} applies with the same probability, yielding
\begin{equation}
\label{eqn:frobenius-application-of-structural-result}
\FNormS{\matA - \matY \pinv{\matY} \matA } \le  \FNormS{\matA - \Pi_{\matY,k}^{\mathrm{F}}(\matA) } \le
\FNormS{\matA - \matA_k } + \FNormS{\matSig_{\rho-k} \matV_{\rho-k}\transp  \matTh\transp \pinv{(\matV_k\transp \matTh\transp)}}.
\end{equation}
We continue by bounding the second term in the right hand side of the above inequality,
\begin{eqnarray*}
 S &:=& \FNormS{\matSig_{\rho-k} \matV_{\rho-k}\transp  \matTh\transp \pinv{(\matV_k\transp \matTh\transp)}} \\
&\leq&  2\FNormS{\matSig_{\rho-k} \matV_{\rho-k}\transp   \matTh\transp \matTh \matV_k} +
2\FNormS{\matSig_{\rho-k} \matV_{\rho-k}\transp   \matTh\transp  ( \pinv{(\matV_k\transp \matTh\transp)} - (\matV_k\transp \matTh\transp)\transp )} \\
&\leq& 2\FNormS{\matSig_{\rho-k} \matV_{\rho-k}\transp   \matTh\transp \matTh \matV_k} +
2\FNormS{\matSig_{\rho-k} \matV_{\rho-k}\transp   \matTh\transp} \TNormS{ \pinv{(\matV_k\transp \matTh\transp)} - (\matV_k\transp \matTh\transp)\transp} \\
&\leq&
 8 \varepsilon \cdot \FNormS{\matSig_{\rho-k} \matV_{\rho-k}\transp  }  +
2\cdot  \left( \frac{11}{4} \FNormS{\matSig_{\rho-k} \matV_{\rho-k}\transp  } \right) \cdot \left( 2.38 \varepsilon \right) \\
&\leq&  22 \varepsilon \cdot \FNormS{\matSig_{\rho-k}}.
\end{eqnarray*}
In the above, in the first inequality we used the fact that for any two matrices $\matX, \matY:$ $\FNormS{\matX+\matY} \le 2 \FNormS{\matX} + 2 \FNormS{\matY}$.
To justify the first estimate in the third inequality, first notice that $\matV\transp_{\rho - k} \matV_k = {\bf 0}_{n \times k}.$ Next use Lemma~\ref{lem:mm} with $\const{R} = \const{C} \sqrt{\ln(k/\delta)}.$ From the lower bound on $r,$ we have that
\[
\frac{\sqrt{r}}{1 + \sqrt{8\ln(n/\delta)}} \geq \sqrt{6\varepsilon^{-1}} \cdot \frac{\sqrt{k} + \sqrt{8 \ln(n/\delta)}}{1 + \sqrt{8 \ln(n/\delta)}} \cdot \const{C} \sqrt{\ln(k/\delta)} > \const{R} > 0,
\]
so this choice of $\const{R}$ satisfies the requirements of Lemma~\ref{lem:mm}. Apply Lemma~\ref{lem:mm} to obtain
\[
\Probab{ \FNormS{\matSig_{\rho-k} \matV_{\rho-k}\transp \matTh\transp \matTh \matV_k} \le 4 (\const{R} + 1)^2 \frac{(\sqrt{k} + \sqrt{8 \ln(n/\delta)})^2}{r} \FNormS{\matSig_{\rho-k} \matV_{\rho-k}\transp   } } \geq 1-e^{-\const{R}^2/4}-2\delta.
\]
Use the lower bound on $r$ to justify the estimate
\begin{align*}
4 (\const{R} + 1)^2 \frac{\big[\sqrt{k} + \sqrt{8 \ln(n/\delta)} \big]^2}{r} & \leq 4 (\const{R} + 1)^2 \frac{\big[\sqrt{k} + \sqrt{8 \ln(n/\delta)}\big]^2}{6\const{C}^2 \varepsilon^{-1} \big[\sqrt{k} + \sqrt{8 \ln(n/\delta)}\big]^2 \ln(k/\delta)} \\
 & = \frac{2 \varepsilon}{3} \cdot \frac{( \const{C} \sqrt{\ln(k/\delta)} + 1)^2}{\const{C}^2 \ln(k/\delta)} \\
 & \leq \frac{2 \varepsilon}{3} \left( 1 + \frac{1}{\const{C} \sqrt{\ln(k/\delta)}} \right)^2.
\end{align*}
This estimate implies that
\[
 \Probab{ \FNormS{\matSig_{\rho-k} \matV_{\rho-k}\transp \matTh\transp \matTh \matV_k} \le \frac{2 \varepsilon}{3} \left(1 + \frac{1}{\const{C}\sqrt{\ln(k/\delta)}}\right)^2 \FNormS{\matSig_{\rho-k} \matV_{\rho-k}\transp   } } \geq 1 - \delta^{C^2 \ln(k/\delta) /4} - 2\delta.
\]
Since $\const{C} > 1$ and $k \geq 2$, a simple numerical estimation allows us to state that, more simply,
\[
 \Probab{ \FNormS{\matSig_{\rho-k} \matV_{\rho-k}\transp \matTh\transp \matTh \matV_k} \le 4 \varepsilon \FNormS{\matSig_{\rho-k} \matV_{\rho-k}\transp } } \geq 1 - \delta^{C^2 \ln(k/\delta)/4} - 2\delta.
\]
The remaining estimates in the third inequality follow from applying Lemma~\ref{lemma:SRHT-preserves-geometry} (keeping in mind our lower bound on $r$) to obtain
\[
\Probab{\TNormS{\pinv{(\matV_k\transp \matTh\transp)} - (\matV_k\transp \matTh\transp)\transp} \le 2.38 \varepsilon} \geq 1 - 3\delta.
\]
and Lemma~\ref{lemma:frobenius-SRHT-subsampling} with $\eta=7/4$ to obtain
\[
 \Probab{\FNormS{\matSig_{\rho-k} \matV_{\rho-k}\transp   \matTh\transp} \le  \frac{11}{4} \FNormS{\matSig_{\rho-k} \matV_{\rho-k}\transp  }} \geq 1 - \left[ \frac{\expe^{7/4}}{ (1 + 7/4)^{1 + 7/4} } \right]^{r/(1 + \sqrt{8 \ln(n/\delta)})^2} - \delta.
 \]
We have the estimate
\[
\frac{\expe^{7/4}}{ (1 + 7/4)^{1 + 7/4} } < \frac{1}{\expe},
\]
so in fact
\begin{align*}
\Probab{\FNormS{\matSig_{\rho-k} \matV_{\rho-k}\transp   \matTh\transp} \le  \frac{11}{4} \FNormS{\matSig_{\rho-k} \matV_{\rho-k}\transp  }} & \geq 1 - \expe^{-r/(1 + \sqrt{8 \ln(n/\delta)})^2} - \delta \\
 & \ge 1 - \expe^{-6 \const{C}^2 \varepsilon^{-1} \ln(k/\delta)} - \delta \\\
 & \ge 1 - \expe^{ -\ln(k/\delta)} - \delta \\
  & \ge 1 - 2\delta.
\end{align*}
Combining~\eqref{eqn:frobenius-application-of-structural-result} with the bound on $S$, we obtain
\[
\FNormS{\matA - \matY \pinv{\matY} \matA }  \le \FNormS{\matA - \Pi_{\matY,k}^{\mathrm{F}}(\matA) } \le \left( 1 +  22 \varepsilon \right) \cdot \FNormS{\matA-\matA_k}.
\]
Taking the square-roots of both sides
and using the fact that $\sqrt{1 +  22 \varepsilon} \le 1 +  22 \varepsilon$ gives the bound
\[
\FNorm{\matA - \matY \pinv{\matY} \matA } \le \FNorm{\matA - \Pi_{\matY,k}^{\mathrm{F}}(\matA) }
 \le \sqrt{1 + 22 \varepsilon} \cdot \FNorm{\matA-\matA_k} \le \left( 1 +  22 \varepsilon \right)  \FNorm{\matA-\matA_k}.
\]
Eqn.~(i) in the theorem follows directly from this:
\[
 \FNorm{\matA - \matY \pinv{\matY} \matA } \le \left( 1 +  22 \varepsilon \right)  \FNorm{\matA-\matA_k}.
\]
To derive Eqn.~(ii), recall the equality $\FNorm{\matA - \tilde{\matA}_k } =\FNorm{\matA - \Pi_{\matY,k}^{\mathrm{F}}(\matA) }, $ established in Lemma~\ref{lem:bestF}. From this it follows that
\[
 \FNorm{\matA - \tilde{\matA}_k } \leq \left( 1 +  22 \varepsilon \right)  \FNorm{\matA-\matA_k}
\]
also.

We now prove Eqn.~(iii) in the theorem. Eqn.~(\ref{eqn:tropp-structural-result2a}) with $\xi = \mathrm{F}$ and $\matOmega = \matTh\transp \in \R^{n \times r}$  gives
$$
\FNormS{\matA_k - \matY \pinv{\matY} \matA} \leq \FNorm{\matA - \matA_k} +
\FNormS{\matSig_{\rho - k} \matV_{\rho-k} \transp \matTh\transp \pinv{\left( \matV_k \transp \matTh\transp \right)} }.
$$
Now recall the bound for $S$:
$$ \FNormS{\matSig_{\rho - k} \matV_{\rho-k} \transp \matTh\transp \pinv{\left( \matV_k \transp\matTh\transp \right)} } \le 22 \varepsilon \FNormS{\matA-\matA_k}.$$
So,
$$
\FNormS{\matA_k - \matY \pinv{\matY} \matA} \leq \FNormS{\matA - \matA_k} +
22 \varepsilon \cdot \FNormS{\matA - \matA_k}  = \left( 1 + 22 \varepsilon \right) \cdot \FNormS{\matA - \matA_k}.
$$
Taking the square-roots of both sides and using the fact that $\sqrt{1 +  22 \varepsilon} \le 1 +  22 \varepsilon$ gives Eqn.~(iii).

Finally, we prove Eqn.~(iv):
$$ \FNorm{\matA_k -  \tilde{\matA}_k } = \FNorm{\matA-\matA_k - (\matA-\tilde{\matA}_k) } \leq \FNorm{\matA - \matA_k} + \FNorm{\matA -  \tilde{\matA}_k } \le (2+22\varepsilon) \FNorm{\matA-\matA_k},$$
where the first inequality follows by the triangle inequality and the second by using the bound obtained in
Eqn.~(ii) in the theorem.


The failure probability in the theorem follows from a union bound on all the probabilistic events involved in bounding $S$.

\subsubsection{Spectral norm bounds}
We now prove the spectral norm bounds in Theorem~\ref{thm:quality-of-approximation-guarantee}  (i.e Eqns.~(v), (vi), (vii), and~(viii)).
Lemma~\ref{lemma:SRHT-preserves-geometry} implies that, with this choice of $r,$
$$
  \TNormS{ \pinv{( \matV_k \transp \matTh \transp)} } \leq (1 - \sqrt{\varepsilon})^{-1},
$$
with probability at least $1 - 3\delta.$ Consequently, $\matV_k \transp \matTh \transp$ has full row-rank and Lemma~\ref{prop:structural-result}
 with $\matOmega = \matTh\transp \in \R^{n \times r}$ and  $\xi=2$ applies with the same probability, yielding
\begin{equation}\label{eq1}
 \TNormS{\matA - \matY \pinv{\matY} \matA } \leq \TNormS{\matA - \matA_k} + (1-\sqrt{\varepsilon})^{-1} \TNormS{\matSig_{\rho-k} \matV_{\rho-k} \transp \matTh \transp}.
\end{equation}
Also, the spectral norm bound in Lemma ~\ref{lem:bestF} implies
\begin{equation}\label{eq2}
 \TNormS{\matA - \tilde{\matA}_k } \leq 2 \TNormS{\matA - \Pi_{\matY,k}^{\mathrm{2}}(\matA) }  \leq 2 \left( \TNormS{\matA - \matA_k} + (1-\sqrt{\varepsilon})^{-1} \TNormS{\matSig_{\rho-k} \matV_{\rho-k} \transp \matTh \transp} \right).
\end{equation}
We now provide an upper bound for $\sqrt{Z}$ where $Z$ is the scalar
$$Z := \TNormS{\matA - \matA_k} + (1-\sqrt{\varepsilon})^{-1} \TNormS{\matSig_{\rho-k} \matV_{\rho-k} \transp \matTh \transp}.$$
From Lemma~\ref{lemma:spectral-SRHT-subsampling} we obtain
\[
Z \leq \left(1 + \frac{5}{1-\sqrt{\varepsilon}}\right) \cdot \TNormS{\matA - \matA_k} + \frac{ \ln (\rho/\delta)}{(1-\sqrt{\varepsilon}) r} \left( \FNorm{\matA - \matA_k} + \sqrt{8 \ln(n/\delta)} \TNorm{\matA - \matA_k} \right)^2
\]
with probability at least $1 - 5\delta.$ Using that $\varepsilon < 1/3$, we see that $(1-\sqrt{\varepsilon})^{-1} < 3$, so
\[
Z \leq 16 \cdot \TNormS{\matA - \matA_k} + \frac{ 3 \ln (\rho/\delta)} {r} \left( \FNorm{\matA - \matA_k} + \sqrt{8 \ln(n/\delta)} \TNorm{\matA - \matA_k} \right)^2.
\]
Use the subadditivity of the square-root function and rearrange the spectral and Frobenius norm terms to obtain that
\[
 \sqrt{Z} \leq \left(4 +
 \sqrt{\frac{3 \ln(n/\delta)\ln(\rho/\delta)}{r}} \right) \cdot \TNorm{\matA - \matA_k} +
 \sqrt{\frac{3 \ln(\rho/\delta)}{r}} \cdot \FNorm{\matA - \matA_k}.
 \]
Apply Eqn.~(\ref{eq1}) to arrive at Eqn.~(v) in the theorem,
\[
 \TNorm{\matA - \matY \pinv{\matY} \matA} \leq \left(4 +
 \sqrt{\frac{3 \ln(n/\delta)\ln(\rho/\delta)}{r}} \right) \cdot \TNorm{\matA - \matA_k} +
 \sqrt{\frac{3 \ln(\rho/\delta)}{r}} \cdot \FNorm{\matA - \matA_k}.
\]
Take the square root of both sides of Eqn.~(\ref{eq2}), use the subadditivity of the square root function, and use the bound for $\sqrt{Z}$ to find Eqn.~(vi):
\[
 \TNorm{\matA - \tilde{\matA}_k} \leq \left(6 +
 \sqrt{\frac{6 \ln(n/\delta)\ln(\rho/\delta)}{r}} \right) \cdot \TNorm{\matA - \matA_k} +
 \sqrt{\frac{6 \ln(\rho/\delta)}{r}} \cdot \FNorm{\matA - \matA_k}.
\]

We now derive the spectral norm bounds on the forward errors. Eqn.~(\ref{eqn:tropp-structural-result2a}) with  $\matOmega = \matTh\transp \in \R^{n \times r}$ gives
\[
 \TNormS{\matA_k - \matY \pinv{\matY} \matA} \leq \TNorm{\matA - \matA_k} +
 \TNormS{\matSig_{\rho - k} \matV_{\rho-k} \transp \matTh\transp \pinv{\left( \matV_k \transp \matTh\transp \right)} }.
\]
Use the inequality  $\TNormS{ \pinv{( \matV_k \transp \matTh \transp)} } \leq (1 - \sqrt{\varepsilon})^{-1}$  to obtain
$$
\TNormS{\matA_k - \matY \pinv{\matY} \matA} \leq
\TNormS{\matA - \matA_k} +
(1 - \sqrt{\varepsilon})^{-1}  \TNormS{\matSig_{\rho - k} \matV_{\rho-k} \transp \matTh\transp  }.
$$
Take the square-root of both sides of this inequality to obtain
$$
\TNorm{\matA_k - \matY \pinv{\matY} \matA} \leq
\sqrt{ \TNormS{\matA - \matA_k} + (1 - \sqrt{\varepsilon})^{-1}  \TNormS{\matSig_{\rho - k} \matV_{\rho-k} \transp \matTh\transp } }
$$
and identify the righthand side as $\sqrt{Z}.$ Use the bound on $\sqrt{Z}$ to arrive at Eqn.~(vii):
\[
\TNorm{\matA_k - \matY \pinv{\matY} \matA} \leq
\left(4 + \sqrt{\frac{3 \ln(n/\delta)\ln(\rho/\delta)}{r}} \right) \cdot \TNorm{\matA - \matA_k}  +
\sqrt{\frac{3 \ln(\rho/\delta)}{r}} \cdot \FNorm{\matA - \matA_k}.
\]

We now prove Eqn.~(viii). First, recall that $\tilde{\matA}_k = \matQ\matX_{opt}$ and observe that
\[
\TNorm{\matA_k - \tilde{\matA}_k  } =  \TNorm{\matA_k + \matA_{\rho-k} - \tilde{\matA}_k - \matA_{\rho-k} }  \le
 \TNorm{\matA -\tilde{\matA}_k } + \TNorm{ \matA_{\rho-k} }.
\]
Now, recall  Eqn.~(vi),
\[
\TNorm{\matA - \tilde{\matA}_k} \leq \left(6 +
 \sqrt{\frac{6 \ln(n/\delta)\ln(\rho/\delta)}{r}} \right) \cdot \TNorm{\matA - \matA_k} +
 \sqrt{\frac{6 \ln(\rho/\delta)}{r}} \cdot \FNorm{\matA - \matA_k}.
\]
In conjunction with the previous inequality, this gives us the desired bound:
\[
\TNorm{\matA - \tilde{\matA}_k} \leq \left(7 +
 \sqrt{\frac{12 \ln(n/\delta)\ln(\rho/\delta)}{r}} \right) \cdot \TNorm{\matA - \matA_k} +
 \sqrt{\frac{6 \ln(\rho/\delta)}{r}} \cdot \FNorm{\matA - \matA_k}.
\]
Finally, we recall that the two probabilistic events we used in our derivations --- that $\pinv{(\matV_k\transp\matTh\transp)}$ is bounded and that our application of
Lemma~\ref{lemma:spectral-SRHT-subsampling} succeeds--- hold with probabilities at least $1 - 3\delta$ and $1 - 2\delta$ respectively, so the failure probability for
each of these four spectral error bounds is no more than $5\delta.$

\subsubsection{Running time Analysis}
The matrix $\matY $ can be constructed in at most $2 m n \log_2 (r+1)$ arithmetic operations (see Lemma~\ref{prop:SRHT-compute-time}).

Given $\matY \in \R^{m \times r},$ the matrix $\matY (\pinv{\matY} \matA) \in \R^{m \times n}$ can be constructed in $\const{O}(m r \ell + m n \ell)$ arithmetic operations as follows.
Observe that $\matY (\pinv{\matY} \matA) = \matQ (\matQ\transp \matA),$ and $\matQ \in \R^{m \times \ell}$ can be computed in $\const{O}(m r \ell)$ time. Recall that $\ell = \min\{m,r\}.$
Computing $\matQ\transp \matA$ requires $\const{O}( m n \ell)$ operations, as does the subsequent computation of $\matQ (\matQ\transp \matA).$
Thus, in total, $\const{O}(m r \ell + m n \ell)$ operations are required. If $\ell = r$, this is $\const{O}(m r^2 + m n r)$,
but if $r > m$, the total operation count becomes $\const{O}(m^2 (r  +n))$.

Finally, given $\matY,$ the matrix $\tilde{\matA}_k$ can be constructed in $\const{O}(mn \ell + \ell^2 n)$ arithmetic operations as follows.
As argued above, $\matQ$ can be constructed in $\const{O}(m r \ell)$ operations, then the product $\matQ\transp \matA \in \R^{\ell \times n}$
can be computed in $\const{O}(m n \ell)$ operations.
The SVD of $\matQ\transp \matA$ requires $\const{O}( \ell n \min\{\ell,n \})$ operations,
which is proportional to  $\const{O}(\ell^2 n)$ because $\min\{ \ell, n \} = \min\{m,r, n \}= \min\{m,r \} = \ell,$ since $r < n$.
The final matrix multiplication $\matQ \matX_{opt}$ again requires  $\const{O}( m n \ell )$  arithmetic operations.
The total operation count is therefore $\const{O}(mn \ell + \ell^2 n)$. If $m > r$, the operation count is
$\const{O}(mnr + r^2 n)$, while if $m < r$, the operation count becomes $\const{O}( m^2 n )$.


\subsection{Proof of Theorem~\ref{regression1}}\label{sec:guarantees3}

To prove the first bound in the theorem (residual error analysis), we will use Lemma~\ref{prop3}, which is the analog of Lemma~\ref{prop:structural-result} but for linear regression.
Using this lemma, the proof of the first bound in Theorem~\ref{regression1} is similar to the proof of the first Frobenius norm bound
of Theorem~\ref{thm:quality-of-approximation-guarantee}.

We would like to apply Lemma~\ref{prop3} with $\matOmega = \matTh\transp \in \R^{m \times r}$. For convenience, we take $\matU = \matU_{\matA}.$
Notice that Lemma~\ref{lemma:SRHT-preserves-geometry} implies that with probability at least $1-3\delta$,
$\rank( \matTh \matU ) = \rho = n$; so, Lemma~\ref{prop3} applies with the same probability, yielding
\begin{equation}\label{eqn:regression}
\TNormS{ \matA \tilde{\x}_{opt} - \b  } \le \TNormS{ \matA \x_{opt} - \b  } + \TNormS{ \pinv{ \left( \matTh \matU \right) } \matTh  \left(\matA \x_{opt} - \b\right)  }.
\end{equation}
We continue by bounding the second term in the right hand side of the above inequality (for notational convenience, let $\z_{opt} = \matA \x_{opt}-\b$),
\begin{eqnarray*}
 S &:=& \TNormS{ \pinv{ \left( \matTh \matU \right) } \matTh  \left(\matA \x_{opt} - \b\right)  } \\
&\leq&  2\TNormS{\matU\transp \matTh\transp \matTh \z_{opt}} +
2\TNormS{( \pinv{( \matTh \matU)} - (\matTh \matU)\transp ) \matTh \z_{opt}} \\
&\leq& 2\TNormS{\matU\transp \matTh\transp \matTh \z_{opt}} +
2\TNormS{( \pinv{( \matTh \matU)} - (\matTh \matU)\transp )} \TNormS{ \matTh \z_{opt}} \\
& = & 2\TNormS{ \z_{opt}\transp \matTh\transp \matTh \matU } +
2\TNormS{( \pinv{( \matTh \matU)} - (\matTh \matU)\transp )} \TNormS{ \z_{opt}\transp \matTh\transp} \\
&\leq&
 8 \varepsilon \cdot \TNormS{\z_{opt}  }  +
2\cdot \left( 2.38 \varepsilon \right) \cdot \left( \frac{11}{4} \TNormS{\z_{opt}  } \right)  \\
&\leq&  22 \varepsilon \cdot \TNormS{\z_{opt}}.
\end{eqnarray*}
In the above, in the first inequality we used the fact that for any two matrices $\matX, \matY:$ $\FNormS{\matX+\matY} \le 2 \FNormS{\matX} + 2 \FNormS{\matY}$.
To justify the first estimate in the third inequality, first notice that $\z_{opt}\transp \matU =  {\bf 0}_{1 \times n}$ since
\[
\z_{opt}\transp \matU
= \left( \matA \x_{opt} - \b \right)\transp \matU
= \left( \matA \matA^+ \b - \b \right)\transp \matU
= \left( \matU \matU\transp \b - \b \right)\transp \matU
=  {\bf 0}_{1 \times n}.
\]
Next, use Lemma~\ref{lem:mm} with $\const{R} = \const{C} \sqrt{\ln(n/\delta)}.$ Recall that
$$r \geq 6 \const{C}^2 \varepsilon^{-1} \big[ \sqrt{n} + \sqrt{8 \ln(m/\delta)} \big]^2 \ln(n/\delta),$$ where $C \ge 1,$ so
\[
\frac{\sqrt{r}}{1 + \sqrt{8 \ln(m/\delta)}} \geq \sqrt{6\varepsilon^{-1}} \cdot \frac{\sqrt{n} + \sqrt{8 \ln(m/\delta)}}{1 + \sqrt{8 \ln(m/\delta)}} \cdot \const{C} \sqrt{\ln(n/\delta)} > \const{R} > 0,
\]
and this choice of $\const{R}$ satisfies the requirements of Lemma~\ref{lem:mm}. Apply Lemma~\ref{lem:mm} to obtain
\[
\Probab{ \TNormS{ \z_{opt}\transp \matTh\transp \matTh \matU} \le 4 (\const{R} + 1)^2 \frac{(\sqrt{n} + \sqrt{8 \ln(m/\delta)})^2}{r} \TNormS{\z_{opt}\transp }}
\geq 1-e^{-\const{R}^2/4}-2\delta.
\]
Manipulations similar to those used in proving the first Frobenius norm bound in Theorem~\ref{thm:quality-of-approximation-guarantee} show that the lower bound on $r$ implies
\[
 \Probab{ \TNormS{\z_{opt}\transp \matTh\transp \matTh \matU} \le 4 \varepsilon \TNormS{\z_{opt}\transp} } \geq 1 - \delta^{C^2 \ln(n/\delta)/4} - 2\delta.
\]
The remaining estimates in the third inequality follow from applying Lemma~\ref{lemma:SRHT-preserves-geometry} to obtain
\[
\Probab{ \TNormS{( \pinv{( \matTh \matU)} - (\matTh \matU)\transp )} \le 2.38 \varepsilon} \geq 1 - 3\delta.
\]
and Lemma~\ref{lemma:frobenius-SRHT-subsampling} with $\eta=7/4$ to obtain
\[
 \Probab{\TNormS{\z_{opt}\transp \matTh\transp} \le  \frac{11}{4} \TNormS{\z_{opt}\transp  }} \geq 1 - \left[ \frac{e^{7/4}}{ (1 + 7/4)^{1 + 7/4} } \right]^{r/(1 + \sqrt{8 \ln(m/\delta)})^2} - \delta.
\]
Manipulations similar to those used in proving the first Frobenius norm bound in Theorem~\ref{thm:quality-of-approximation-guarantee} show that the latter bound implies
\[
 \Probab{\TNormS{\z_{opt}\transp \matTh\transp} \le  \frac{11}{4} \TNormS{\z_{opt}\transp  }} \geq 1 - 2\delta.
\]

Combining~\eqref{eqn:regression} with the bound on $S$, we obtain
\[
\TNormS{\matA \tilde{\x}_{opt}-\b } \le  \left( 1 +  22\varepsilon \right) \cdot \TNormS{\matA \x_{opt}-\b}.
\]
Taking the square-root to both sides and using the fact that $\sqrt{1 +  22 \varepsilon} \le 1 +  22 \varepsilon$ gives the bound in the theorem,
\[
\TNorm{\matA \tilde{\x}_{opt}-\b } \le  \left( 1 +  22 \varepsilon \right) \cdot \TNorm{\matA \x_{opt}-\b}.
\]
The failure probability in the theorem follows by a union bound on all the probabilistic events involved in the proof.

We now prove the forward error bound in the theorem. Towards this end, we will use Lemma~\ref{lem:last} with $\matOmega = \matTh\transp$.
Recall that Eqn.~(\ref{eqn:third}) in the lemma,
$$  \TNorm{ \x_{opt} - \tilde{\x}_{opt} } \le \left( \frac{\alpha}{\beta} \right)^{\frac{1}{2}} \cdot \left( \kappa{ \left( \matA \right)} \sqrt{ \gamma^{-2}   - 1 } \right) \TNorm{\x_{opt}},
$$
is satisfied if $\alpha$ and $\beta$ satisfy Eqns.~(\ref{eqn:first}) and~(\ref{eqn:second}) respectively, and $\gamma \in (0,1]$ satisfies $\TNorm{\matU \matU\transp \b} \geq \gamma \TNorm{\b}.$ By hypothesis, such a $\gamma$ exists. We now show that appropriate $\alpha$ and $\beta$ exist.

Lemma~\ref{lemma:SRHT-preserves-geometry} implies that
$$\sigma_{\min}\left( \matU\transp \matTh\transp \right) \ge \left( 1-\sqrt{\varepsilon} \right)^{\frac{1}{2}}$$
so $\alpha = 1-\sqrt{\varepsilon}$ satisfies Eqn.~(\ref{eqn:first}).
In the proof of the residual error bound in this theorem, we showed that
$$\TNormS{\matU\transp \matTh\transp \matTh \left( \matA \x_{opt} - \b \right)} \le 4 \varepsilon \TNormS{ \matA \x_{opt} - \b },$$
so $\beta = 4 \varepsilon$ satisfies Eqn.~(\ref{eqn:second}). With these choices of $\alpha$ and $\beta,$ Eqn.~(\ref{eqn:third}) in Lemma~\ref{lem:last} gives the claimed forward error bound.

\section{Experiments}\label{sec:experiments}

In this section, we experimentally investigate the tightness of the residual and forward error bounds provided in Theorem~\ref{thm:quality-of-approximation-guarantee} for the spectral and Frobenius norm approximation errors of SRHT low-rank approximations of the forms $\matY\pinv{\matY}\matA$ and $\tilde{\matA}_k = \matQ \matX_{opt}.$ Additionally, we experimentally verify that the SRHT algorithm is not significantly less accurate than the Gaussian low-rank approximation algorithm.

\subsection{Test Matrices}
Let $n = 1024$  and consider the following three test matrices:
\begin{enumerate}[1.]
 \item Matrix $\matA \in \R^{(n+1) \times n}$ is given by
\[ \matA = [100 \vec{e}_1 + \vec{e}_2, 100 \vec{e}_1
 + \vec{e}_3, \ldots, 100 \vec{e}_1 + \vec{e}_{n+1}], \]
where $\vec{e}_i \in \R^{n+1}$ are the standard basis vectors.
\item Matrix $\matB \in \R^{n \times n}$ is diagonal with
 entries $(\matB)_{ii} = 100*(1 - (i-1)/n).$
 \item Matrix $\matC \in \R^{n \times n}$ has the same singular values
   as $\matB,$ but its singular spaces are sampled from the uniform
   measure on the set of orthogonal matrices. More precisely, $\matC =
   \matU \matB \matV\transp,$ where $\mat{G} = \matU \matSig
   \matV\transp $ is the SVD of an $n \times n$ matrix whose entries
   are standard Gaussian random variables.
\end{enumerate}

These three matrices exhibit properties that, judging from the bounds in Theorem~\ref{thm:quality-of-approximation-guarantee}, could challenge the
SRHT approximation algorithm. Matrix $\matA$ is approximately rank
one---there is a large spectral gap after the first singular
value---but the residual spectrum is flat, so for $k \geq 1,$ the
$\FNorm{\matA- \matA_k}$ terms in the spectral norm bound of
Theorem~\ref{thm:quality-of-approximation-guarantee} are quite
large compared to the $\TNorm{\matA- \matA_k}$ terms. Matrices $\matB$ and $\matC$ both have slowly decaying
spectrums, so one again has a large Frobenius term present in the
spectral norm error bound.

$\matB$ and $\matC$ were chosen to have the same singular values but
different singular spaces to reveal any effect that the structure of the
singular spaces of the matrix has on the quality of SRHT approximations. The ``coherence'' of their right singular spaces provides a
summary of the relevant difference in the singular spaces of $\mat{B}$ and $\mat{C}.$ Let $\mathcal{S}$ be a
$k$-dimensional subspace, then its coherence is defined as
\[
 \mu(\mathcal{S}) = \max_{i} \matP_{ii},
\]
where $\matP$ is the projection onto $\mathcal{S};$ the coherence of $\mathcal{S}$ is always between $k/n$ and $1$~\cite{CR09}. It is clear that all the right singular spaces of $\mat{B}$ are maximally coherent and it is known that with high probability the dominant right $k$-dimensional singular space of $\mat{C}$ is quite incoherent, with coherence on the order of $ \max\{k, \log n\}/n$~\cite{CR09}.

To gain an intuition for the potential significance of this difference in coherence, consider a randomized column sampling approach to forming low-rank approximants; that is, consider approximating $\matM_k$ with a matrix $\matY \matY^\dagger \matM$ where $\matY$ comprises randomly sampled columns of $\matM.$
Here and elsewhere we use the matrix $\matM$ to refer interchangeably to $\matA, \matB,$ and $\matC.$
It is known that such approximations are quite inaccurate unless the dominant $k$-dimensional right singular space of $\matM$ is incoherent~\cite{RT10,Git12}. One could interpret SRHT approximation algorithms as consisting of a rotation of the right singular spaces of $\mat{M}$ by multiplying from the right with $\matD \matH\transp$ followed by forming a column sample based approximation. The rotation lowers the coherence of the right singular spaces and thereby increases the probability of obtaining an accurate low-rank approximation. One expects that if $\matM$ has highly coherent right singular spaces then the right singular spaces of $\matM \matD \matH \transp$ will be less coherent but possibly
still far from incoherent. Thus we compare the performance of the SRHT approximations on $\matB$, which has maximally coherent right singular spaces, to their performance on $\matC,$ which has almost maximally incoherent right singular spaces.

\subsection{Empirical comparison of the SRHT and Gaussian algorithms}
\begin{figure}[htp]
\centering
 \subfigure{\includegraphics[width=2.5in, keepaspectratio=true]{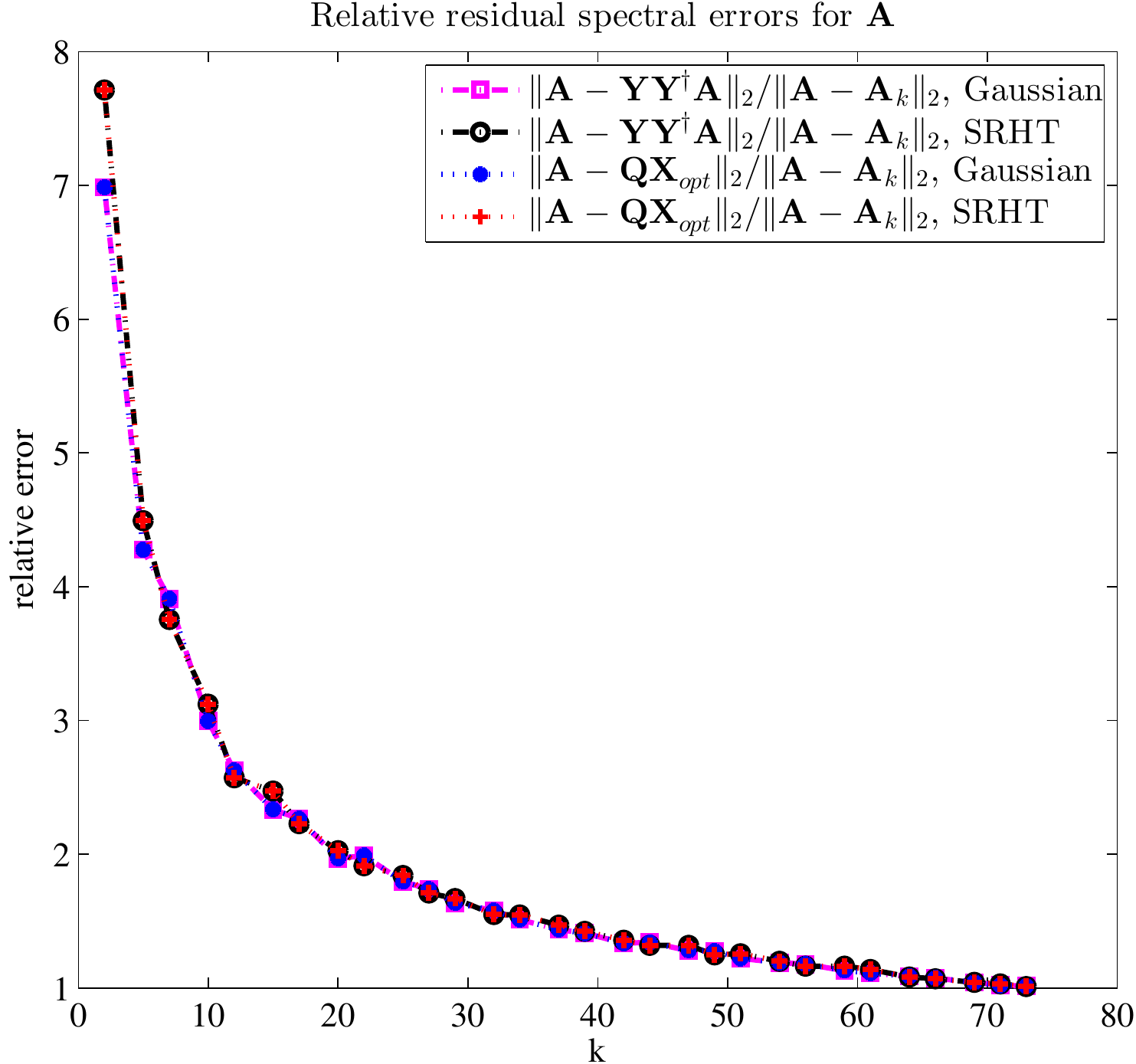}}%
 \subfigure{\includegraphics[width=2.5in, keepaspectratio=true]{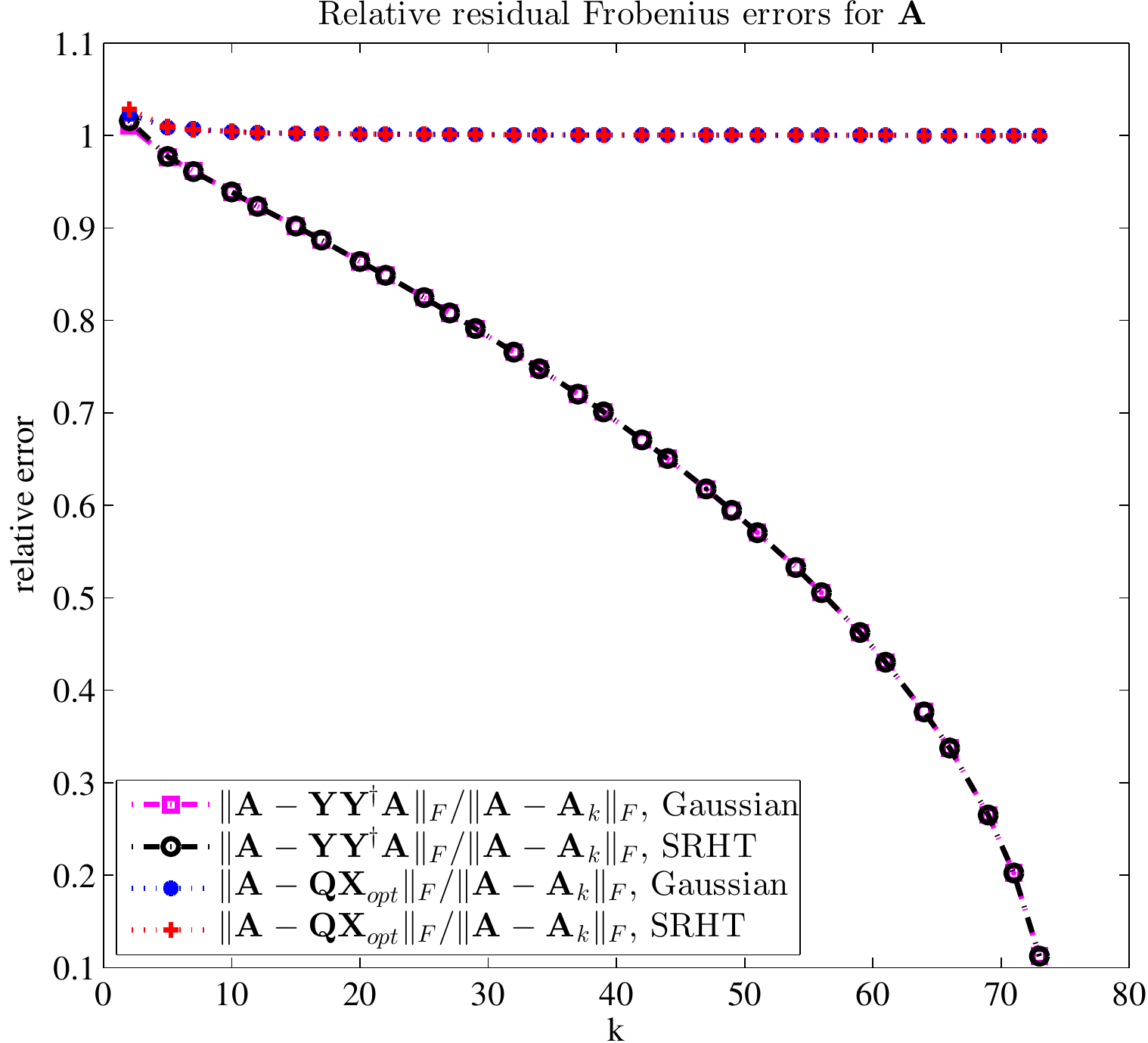}}\\%
 \subfigure{\includegraphics[width=2.5in, keepaspectratio=true]{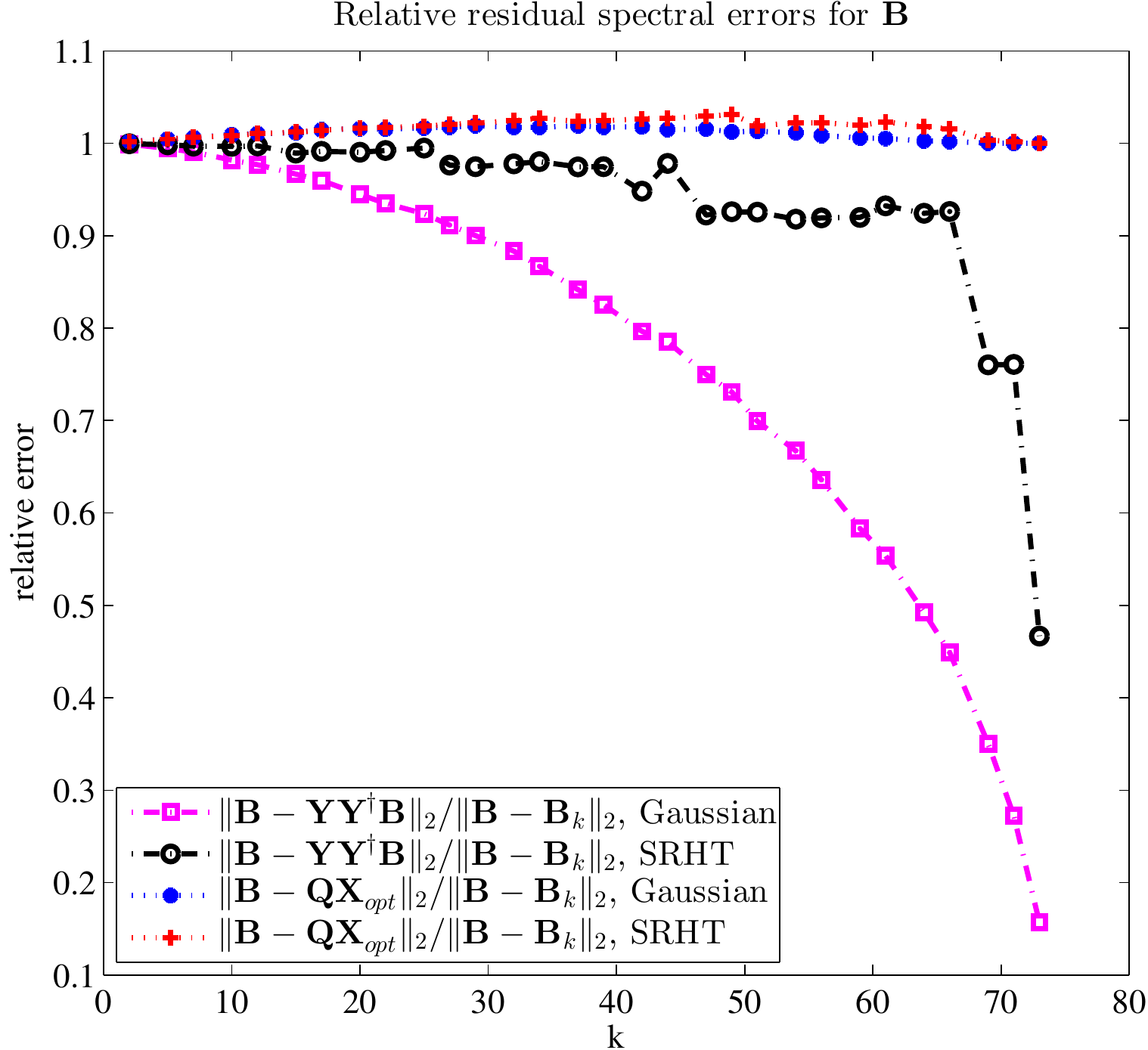}}%
 \subfigure{\includegraphics[width=2.5in, keepaspectratio=true]{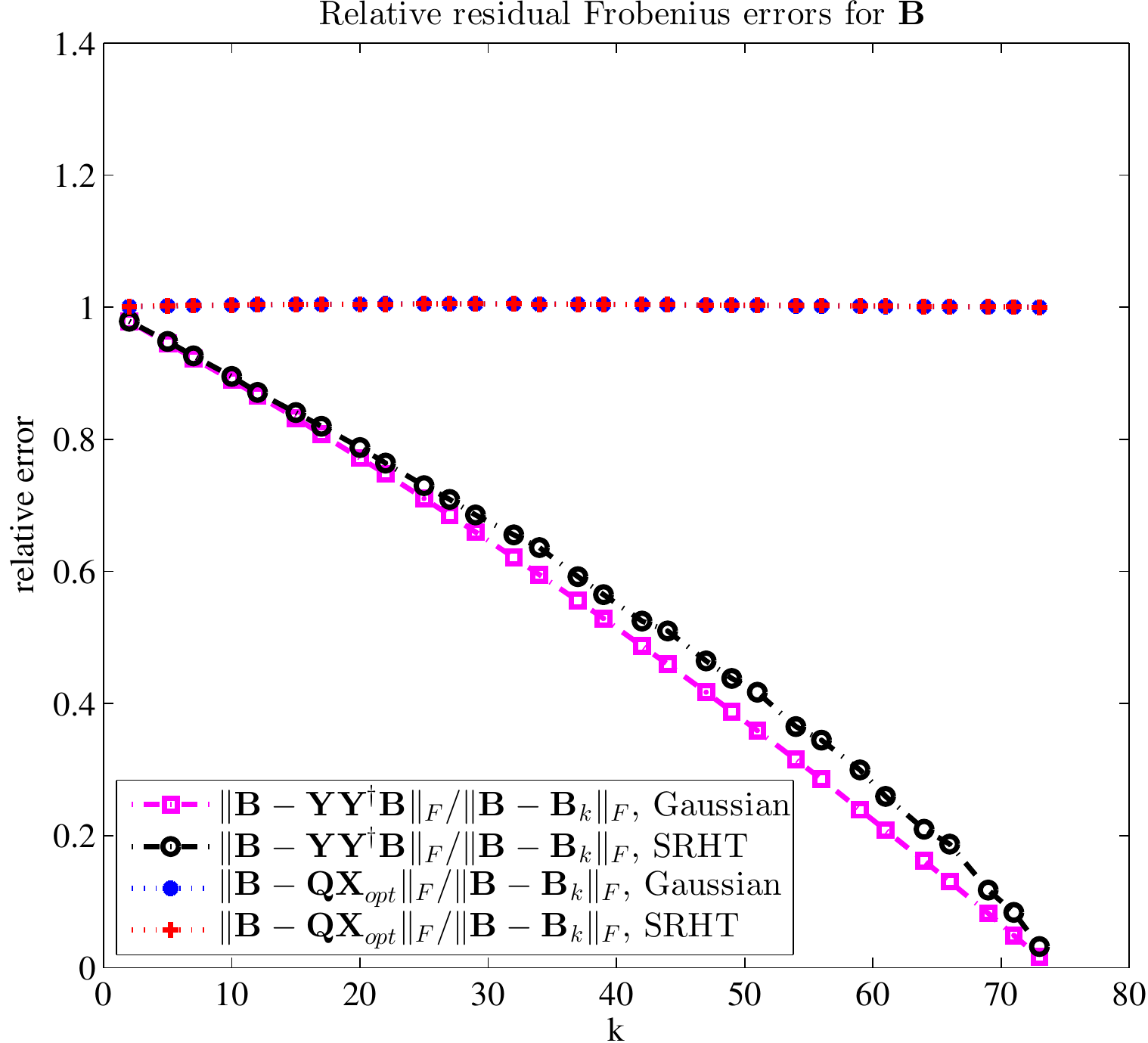}}\\%
 \subfigure{\includegraphics[width=2.5in, keepaspectratio=true]{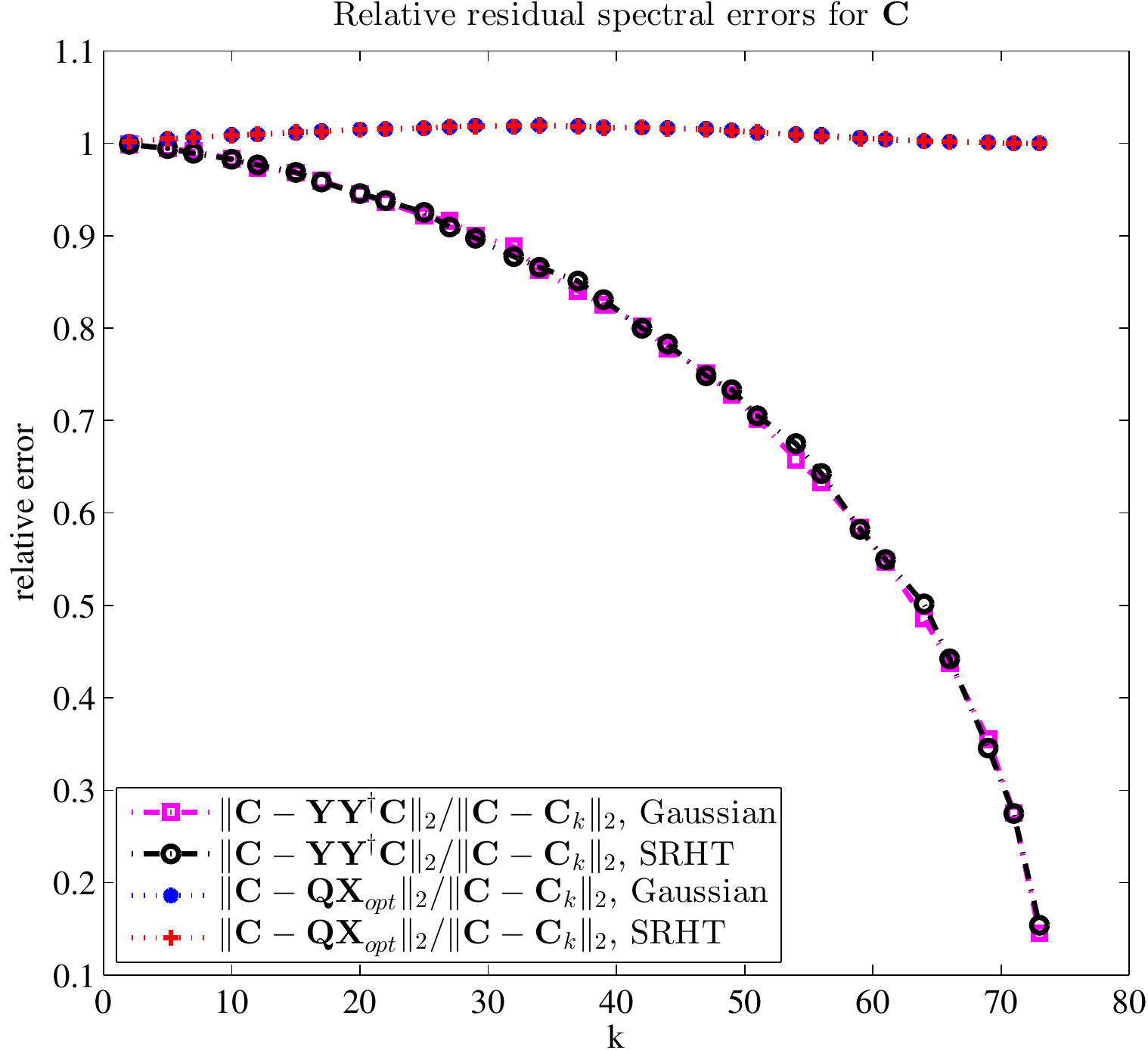}}%
 \subfigure{\includegraphics[width=2.5in, keepaspectratio=true]{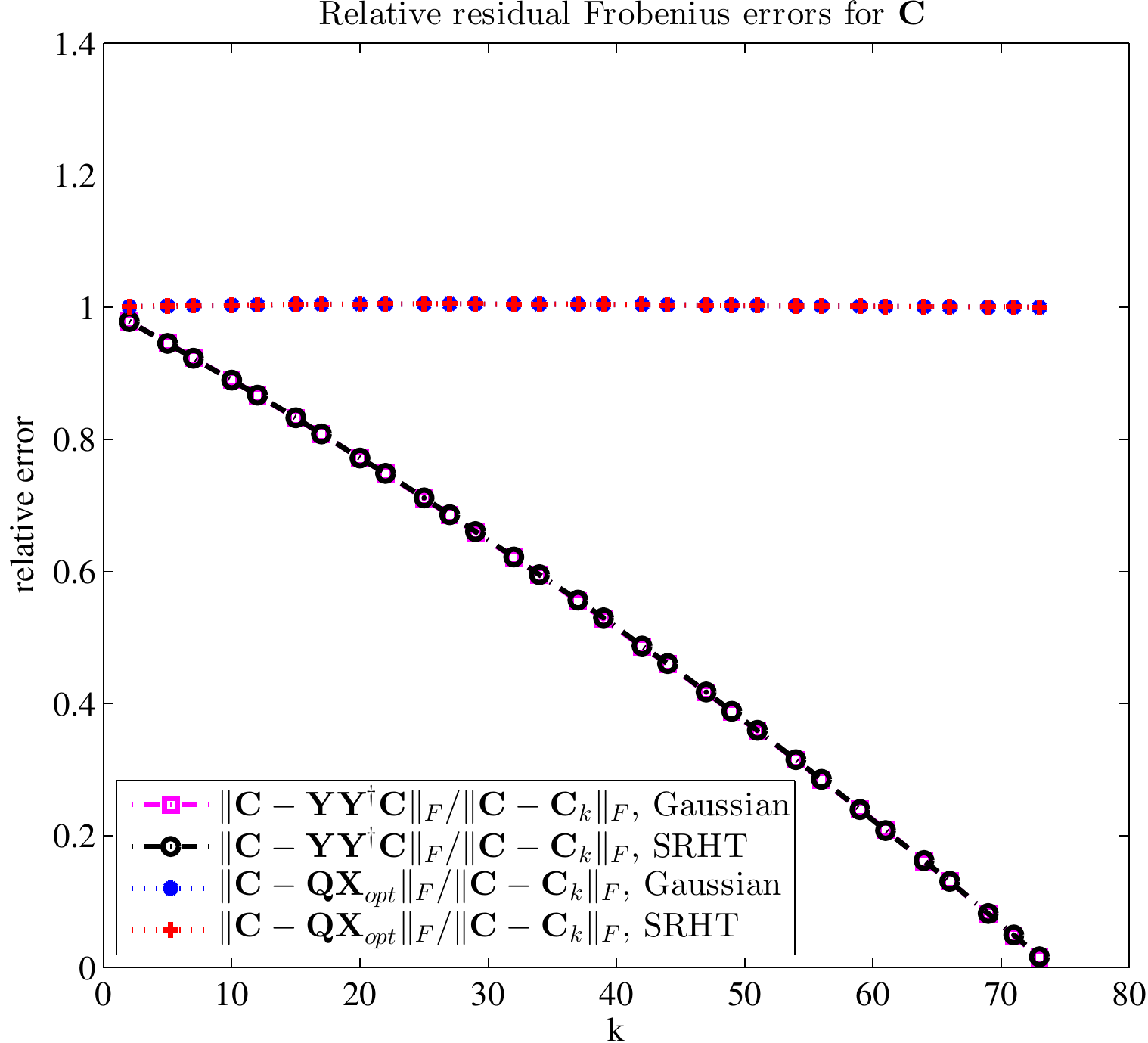}}%
 \caption{Relative spectral and Frobenius norm residual errors of the SRHT and Gaussian low-rank approximation algorithms
 ($\XNorm{\matM - \matY\pinv{\matY} \matM}/\XNorm{\matM - \matM_k}$ and $\XNorm{\matM - \matQ \matX_{opt}}/\XNorm{\matM - \matM_k}$ for $\xi = 2, \mathrm{F}$)
 as a function of $k$ for the three datasets $\matM = \matA, \matB, \matC.$ Each point is the worst of the errors observed over 10 trials. $r = \lceil 2 k \ln n \rceil$ column samples were used in each trial.}
 \label{fig:residualerrors}
\end{figure}
Figure~\ref{fig:residualerrors} depicts the relative residual errors of the Gaussian and SRHT algorithms for both approximations addressed in Theorem~\ref{thm:quality-of-approximation-guarantee}: $\matY \pinv{\matY} \matM$ and $\matQ \matX_{opt},$ which we shall hereafter refer to respectively as the non-rank-restricted and rank-restricted approximations.
The relative residual errors ($\XNorm{\matM - \matY\pinv{\matY} \matM}/\XNorm{\matM - \matM_k}$ and $\XNorm{\matM - \matQ \matX_{opt}}/\XNorm{\matM - \matM_k}$ for $\xi = 2, \mathrm{F}$) shown in this figure for each value of $k$ were obtained by taking the largest of the relative residual errors observed over 10 trials of low-rank approximations each formed using $r = \lceil 2 k \ln n \rceil$ samples.

With the exception of the residual spectral errors on dataset $\mat{A},$ which range between 2 and 9 times greater than the optimal rank-$k$ spectral residual error for $k < 20,$ we see that the residual errors for all three datasets are less than 1.1 times the residual error of $\matM_k$ if not significantly smaller. Specifically, the relative residual errors of the restricted-rank approximations remain less than 1.1 over the entire range of $k$ while the relative residual errors of the non-rank-restricted approximations actually decrease as $k$ increases.

By comparing the residual errors for datasets $\mat{B}$ and $\mat{C},$ which has the same singular values as $\mat{B}$ but is less coherent, we see evidence that the spectral norm accuracy of the SRHT approximations is increased on less coherent datasets; the same is true for the Frobenius norm accuracy to a lesser extent. The Gaussian approximations seem insensitive to the level of coherence. Only on the highly coherent dataset $\mat{B}$ do we see a notable decrease in the residual errors when Gaussian sampling is used rather than an SRHT; however, even in this case the residual errors of the SRHT approximations are comparable with that of $\matB_k.$ In all, Figure~\ref{fig:residualerrors} suggests that the gain in computational efficiency provided by the SRHT does not come at the cost of a significant loss in accuracy and that taking $r = \lceil 2 k \ln n \rceil$ samples suffices to obtain approximations with small residual errors relative to those of the optimal rank-$k$ approximations. Up to the specific
value of the constant, this latter observation coincides with the conclusion of Theorem~\ref{thm:quality-of-approximation-guarantee}.

\begin{figure}[htp]
 \subfigure{\includegraphics[width=2.6in, keepaspectratio=true]{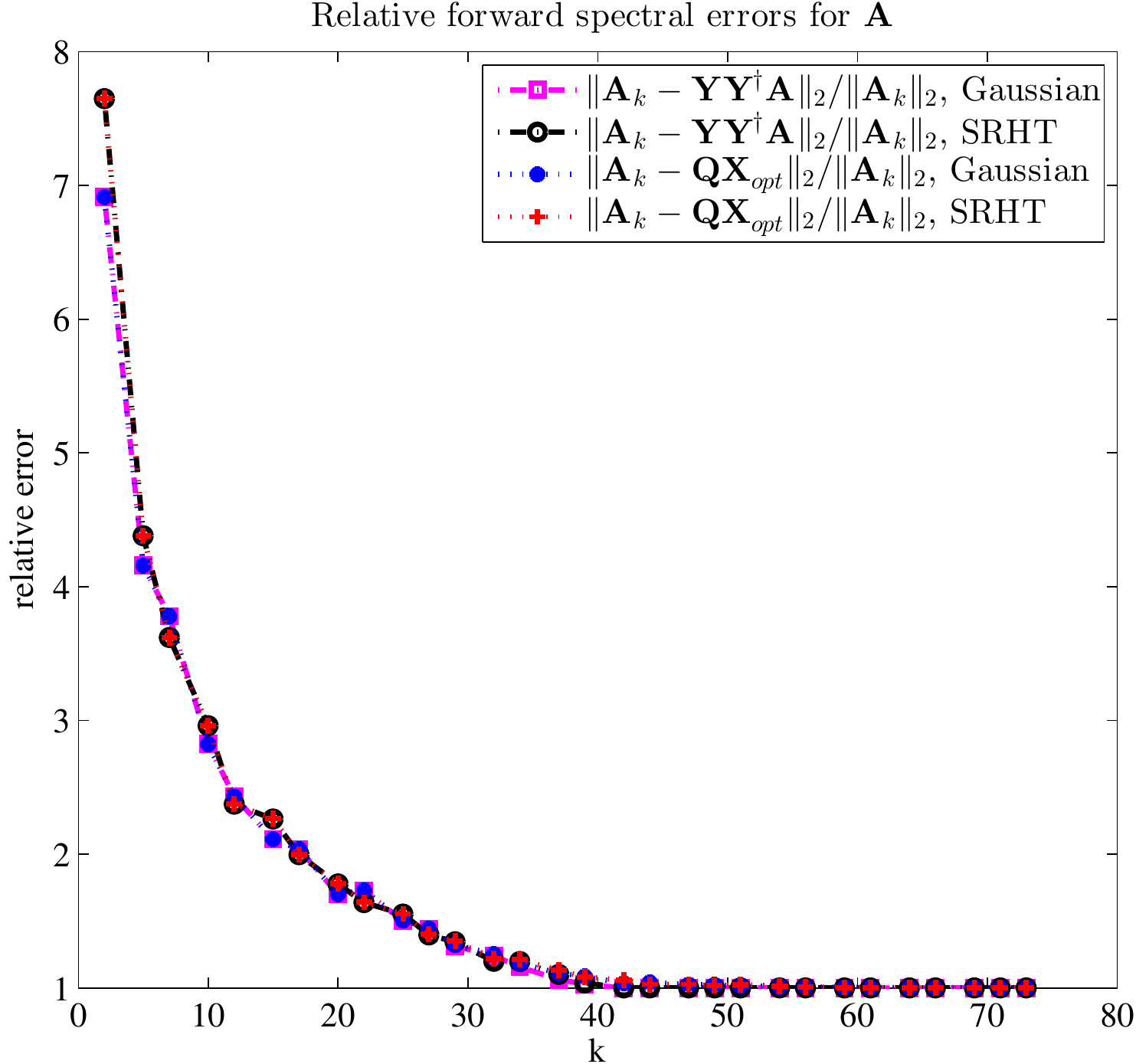}}%
 \subfigure{\includegraphics[width=2.6in, keepaspectratio=true]{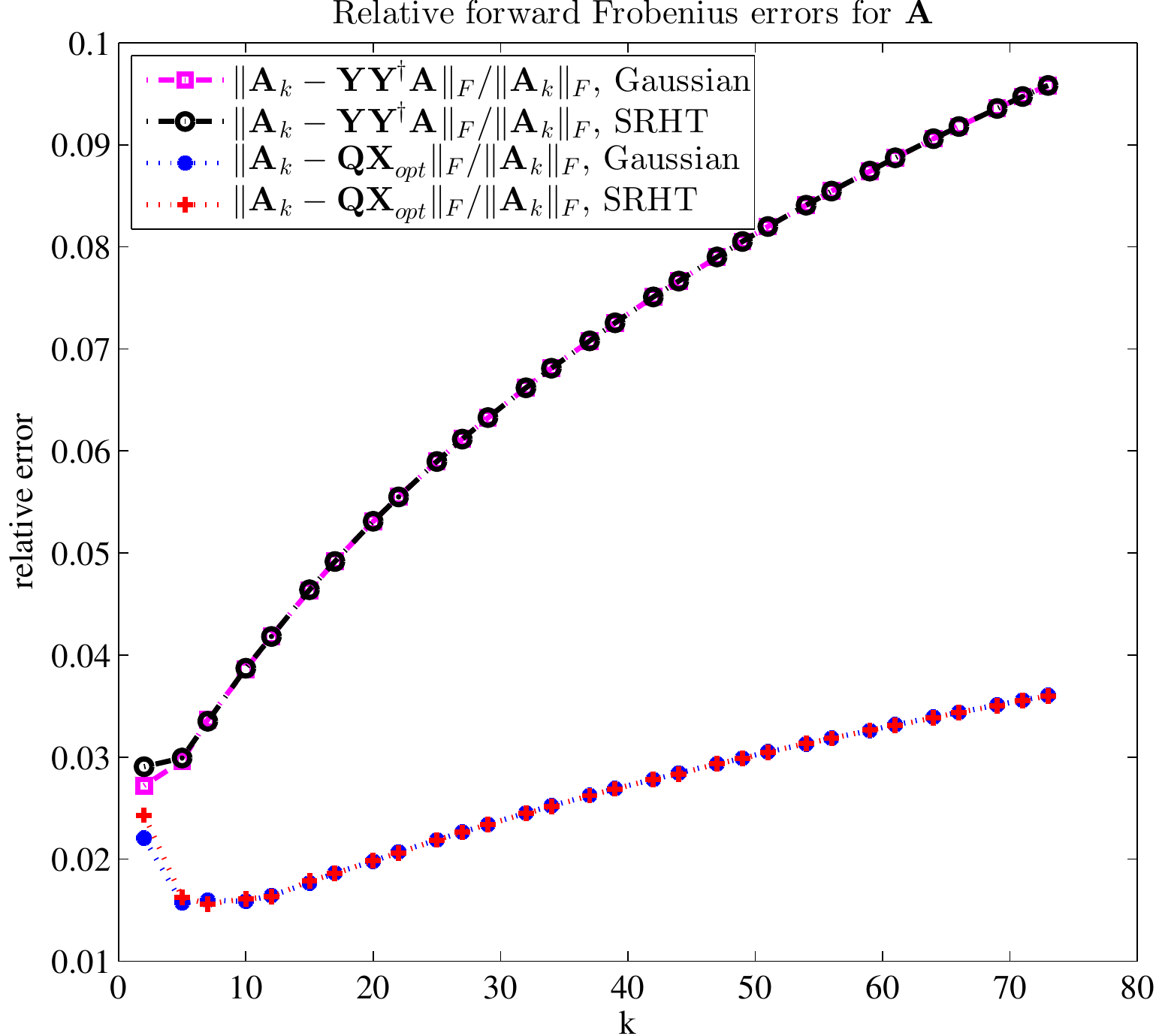}}\\%
 \subfigure{\includegraphics[width=2.6in, keepaspectratio=true]{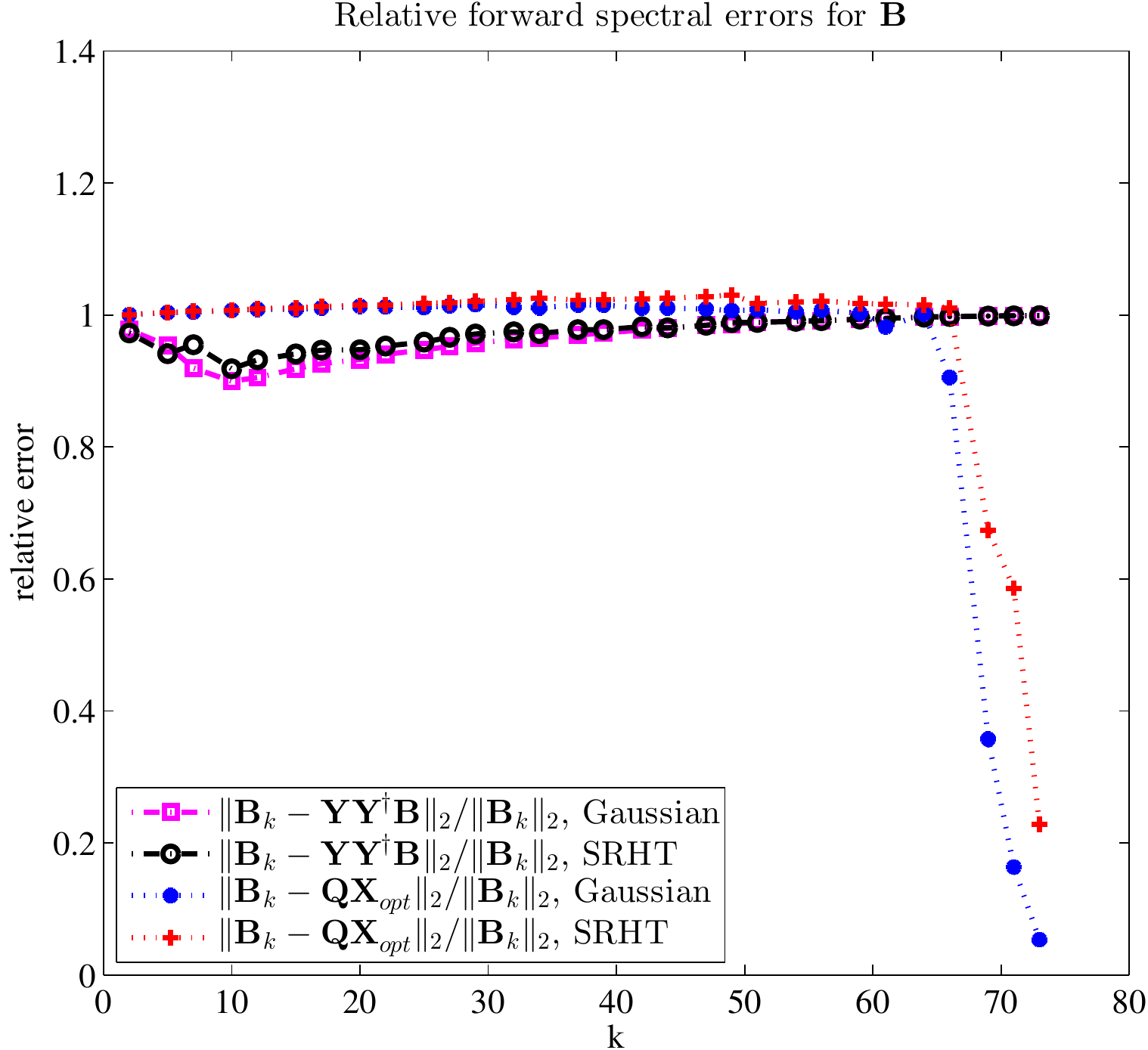}}%
 \subfigure{\includegraphics[width=2.6in, keepaspectratio=true]{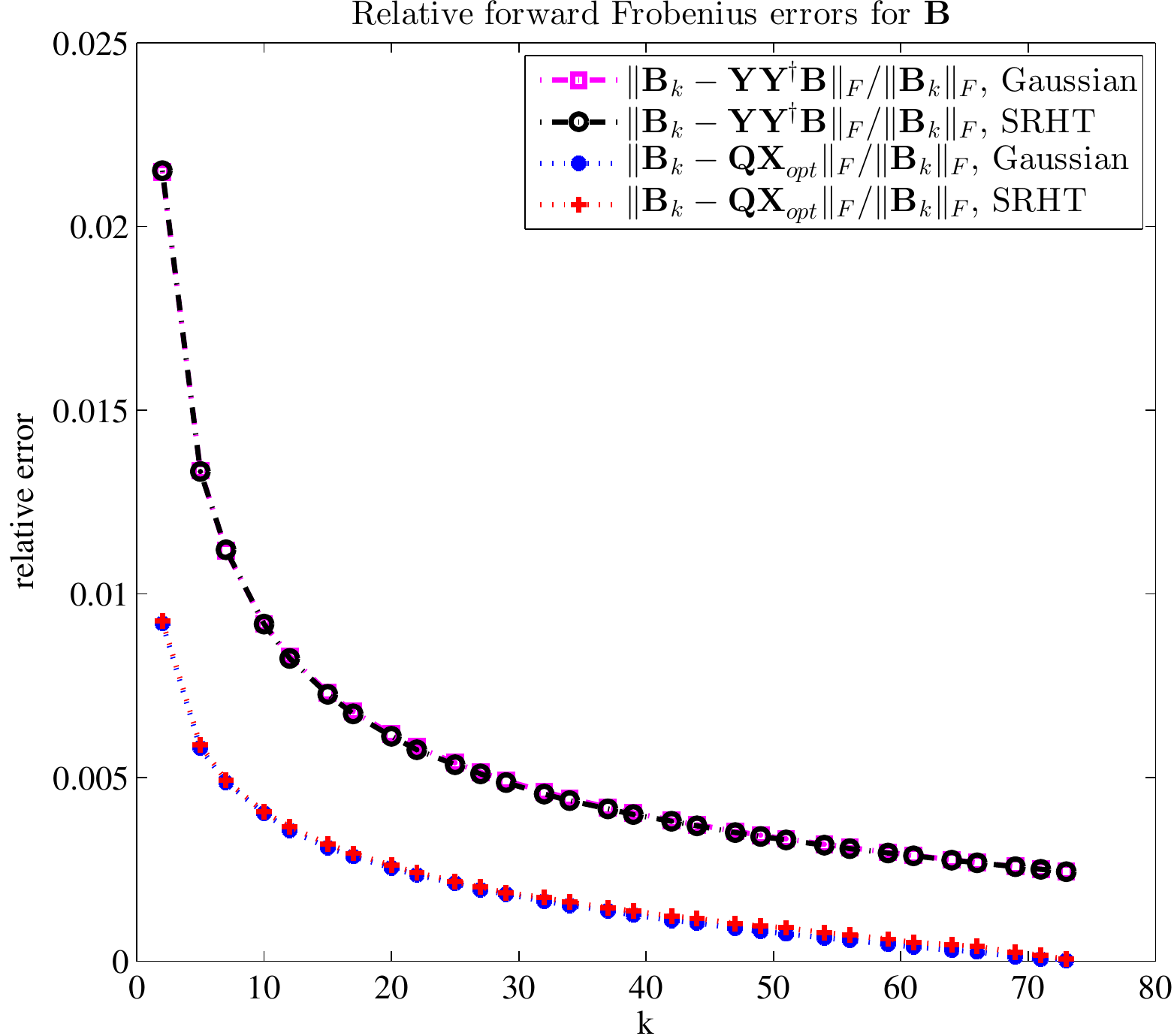}}\\%
 \subfigure{\includegraphics[width=2.6in, keepaspectratio=true]{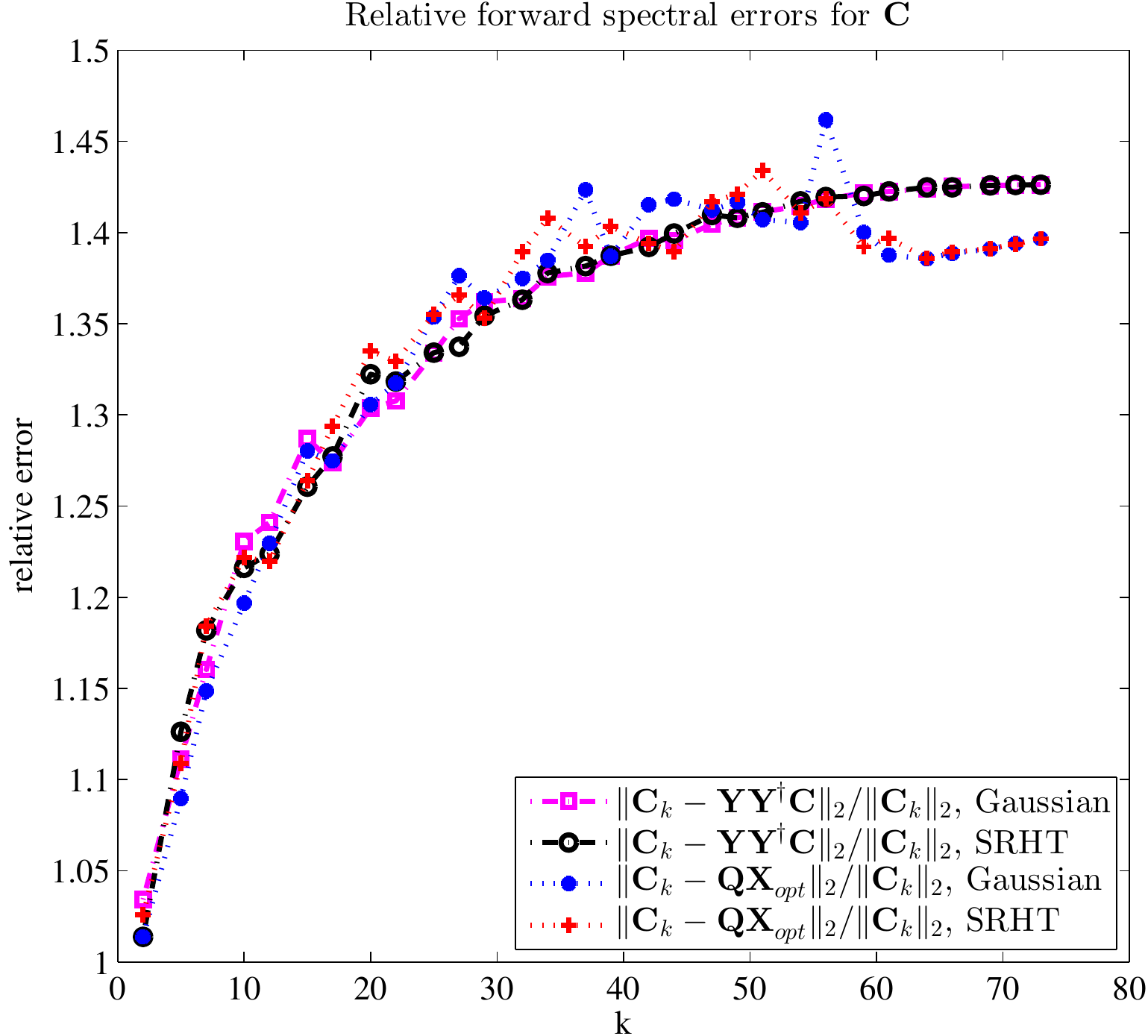}}%
 \subfigure{\includegraphics[width=2.6in, keepaspectratio=true]{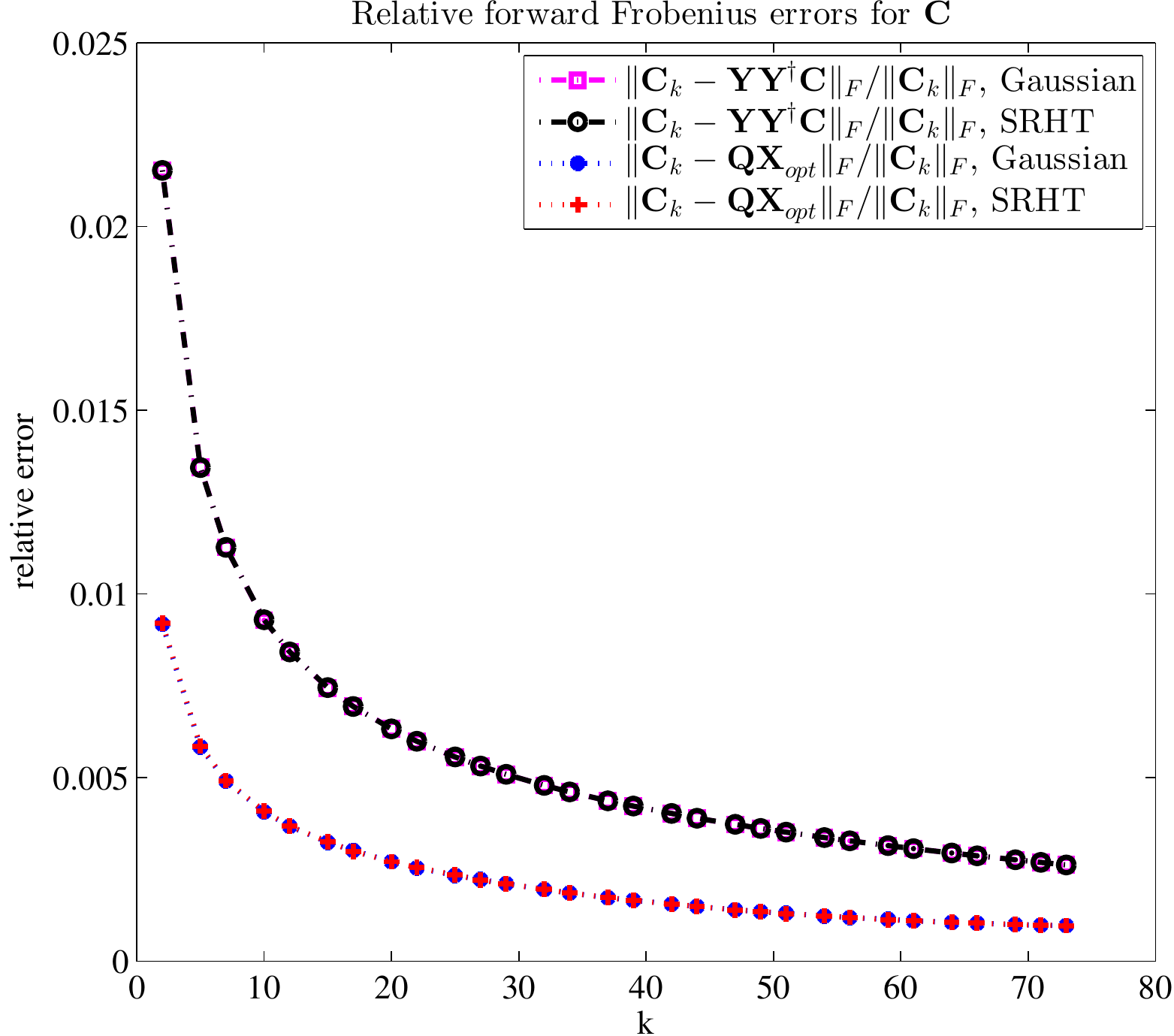}}%
 \caption{The relative spectral and Frobenius norm forward errors of the SRHT and Gaussian low-rank approximation algorithms
 ($\XNorm{\matM_k - \matY\pinv{\matY} \matM}/\XNorm{\matM_k}$ and $\XNorm{\matM_k - \matQ \matX_{opt}}/\XNorm{\matM_k}$ for $\xi = 2, \mathrm{F}$)
 as a function of $k$ for the three datasets $\matM = \matA, \matB, \matC.$ Each point is the worst of the errors observed over 10 trials. $r = \lceil 2 k \ln n \rceil$ column samples were used in each trial.}
 \label{fig:forwarderrors}
\end{figure}

Figure~\ref{fig:forwarderrors} depicts the relative forward errors of the Gaussian and SRHT algorithms ($\XNorm{\matM_k - \matY\pinv{\matY} \matM}/\XNorm{\matM_k}$ and $\XNorm{\matM_k - \matQ \matX_{opt}}/\XNorm{\matM_k}$ for $\xi = 2, \mathrm{F}$) for the non-rank-restricted and rank-restricted approximations. The error shown for each $k$ is the largest relative forward error observed among 10 trials of low-rank approximations each formed using $r = \lceil 2 k \ln n\rceil$ samples. We observe that the forward errors of both algorithms for both choices of sampling matrices are on the scale of the norm of $\mat{M}_k.$ By looking at the relative spectral norm forward errors we see that in this norm, perhaps contrary to intuition, the rank-restricted approximation does not provide a more accurate approximation to $\matM_k$ than does the non-rank-restricted approximation. However the rank-restricted approximation clearly provides a more accurate approximation to $\matM_k$ than the non-rank-restricted
approximation in the
Frobenius norm. A rather
unexpected observation is that the rank-restricted approximations are more accurate in the spectral norm for highly coherent matrices ($\matB$) than they are for matrices which are almost minimally coherent ($\matC$). Overall, Figure~\ref{fig:forwarderrors} suggests that the SRHT low-rank approximation algorithms provide accurate approximations to $\matM_k$ when $r$ is in the regime suggested by Theorem~\ref{thm:quality-of-approximation-guarantee}.

\begin{figure}
 \centering
 \subfigure{\includegraphics[width=4.4in, keepaspectratio=true]{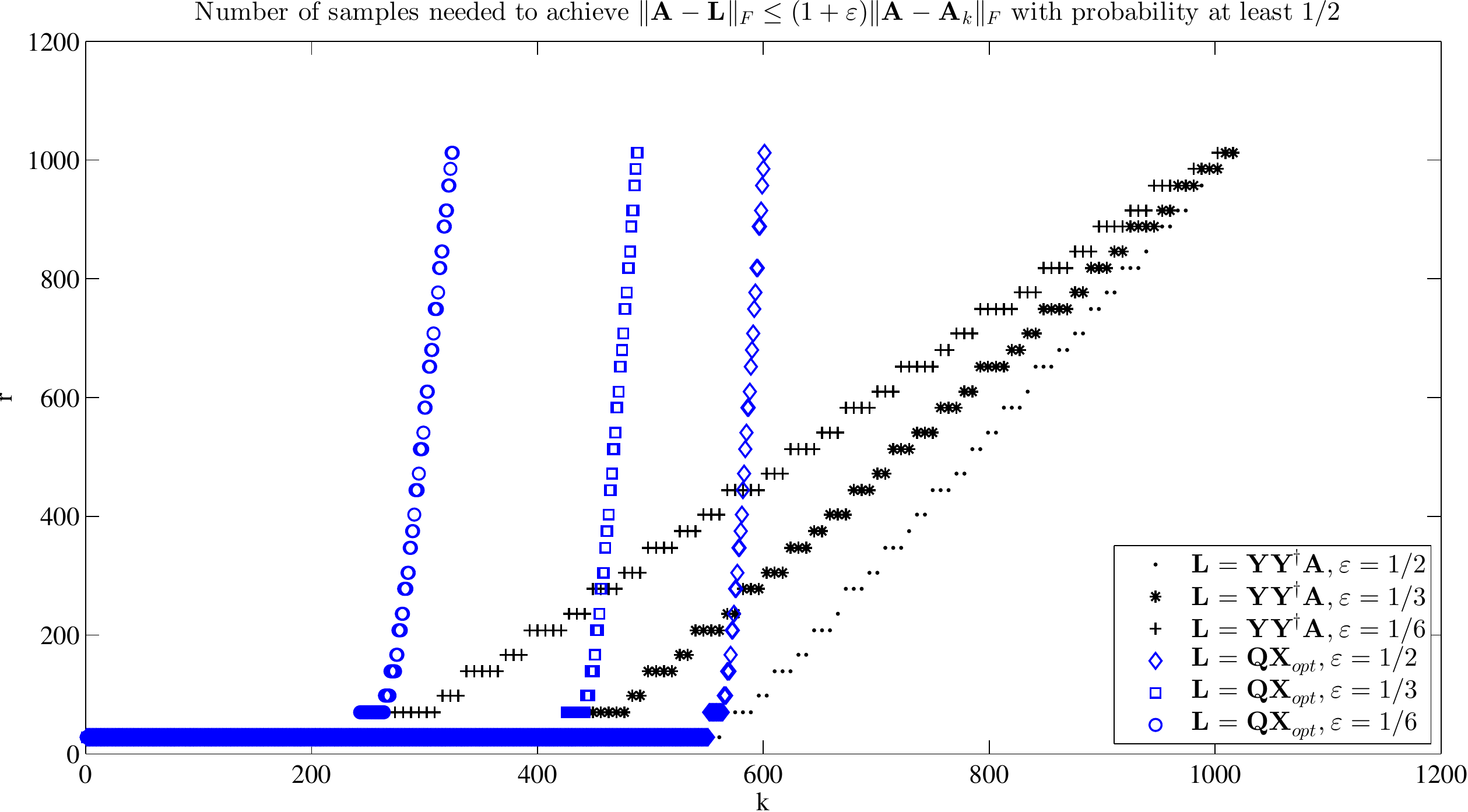}}\\%
 \subfigure{\includegraphics[width=4.4in, keepaspectratio=true]{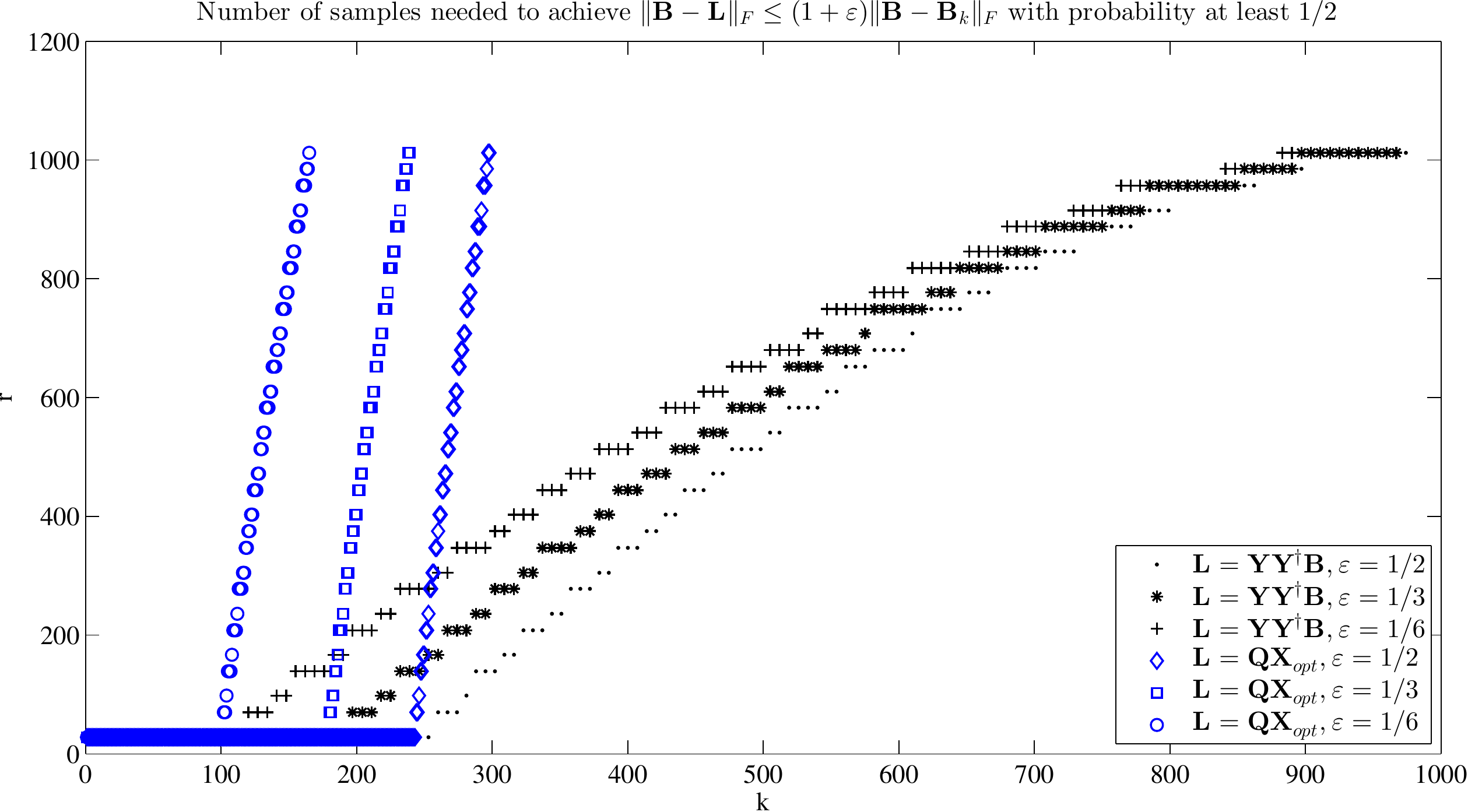}}\\%
 \subfigure{\includegraphics[width=4.4in, keepaspectratio=true]{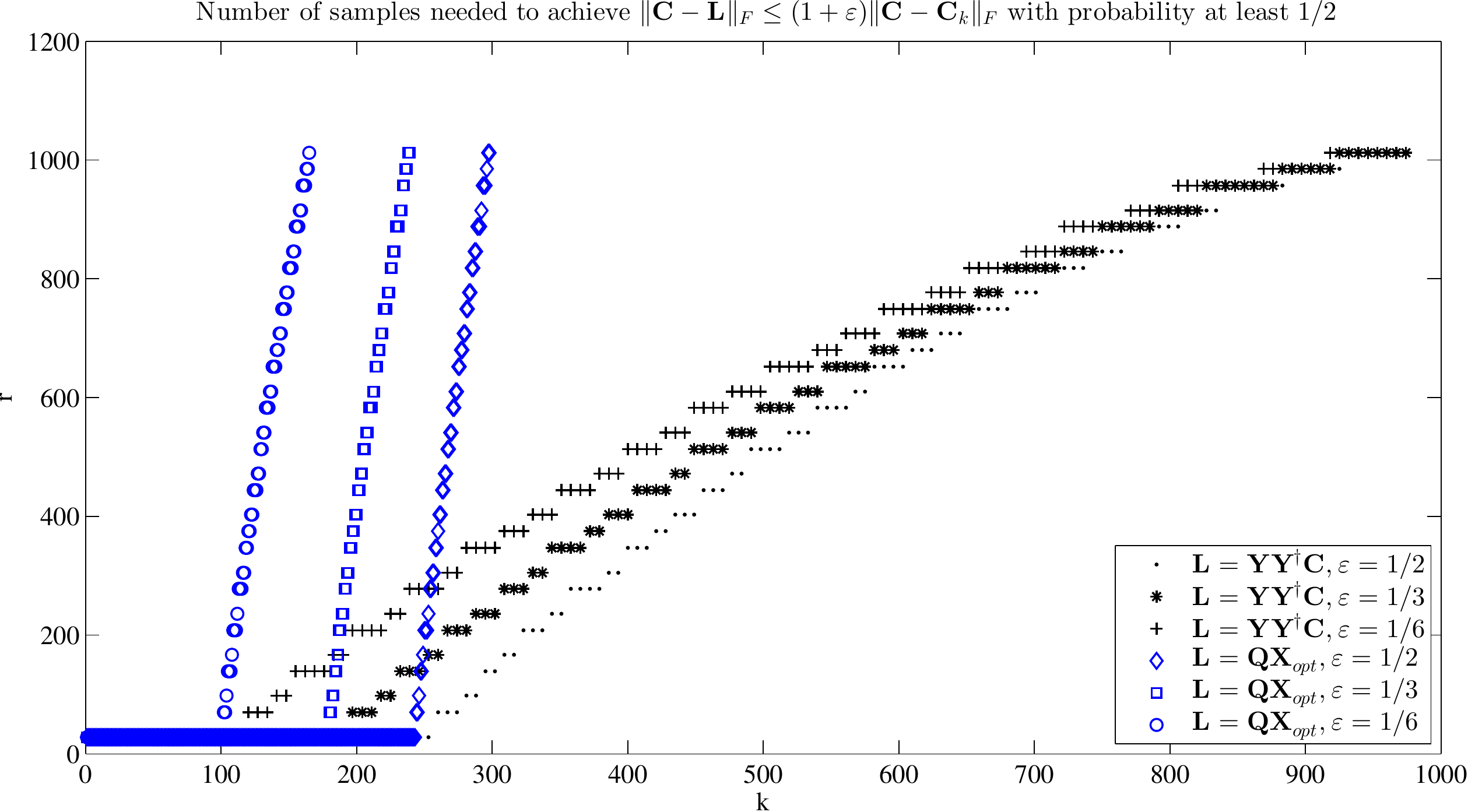}}
 \caption{The value of $r$ empirically necessary to ensure that with probability at least $1/2,$ approximations generated by the SRHT algorithms satisfy $\FNorm{\matM - \matY\matY^\dagger \matM} \leq (1 + \varepsilon) \FNorm{\matM - \matM_k}$ and $\FNorm{\matM - \matQ \matX_{opt}} \leq (1 + \varepsilon) \FNorm{\matM - \matM_k}$ (for $\matM = \matA, \matB, \matC$).}
 \label{fig:empiricalrnecessary}
\end{figure}

\subsection{Empirical evaluation of our error bounds}
Figures~\ref{fig:residualerrors} and~\ref{fig:forwarderrors} show that when $r = \lceil 2 k \ln n \rceil$ samples are taken, the SRHT low-rank approximation algorithms both provide approximations to $\matM$ that are within a factor of $(1 + \varepsilon)$ as accurate in the Frobenius norm as $\matM_k,$ as Theorem~\ref{thm:quality-of-approximation-guarantee} suggests should be the case. More precisely, Theorem~\ref{thm:quality-of-approximation-guarantee} assures us that $528 \varepsilon^{-1} [\sqrt{k} + \sqrt{8 \ln(8 n/\delta)}]^2 \ln(8k/\delta)$ column samples are sufficient to ensure that, with at least probability $1 - \delta$, $\matY \pinv{\matY} \matM$ and $\matQ \matX_{opt}$ have Frobenius norm residual and forward error within $(1 + \varepsilon)$ of that of $\matM_k.$ The factor $528$ can certainly be reduced by optimizing the numerical constants given in Theorem~\ref{thm:quality-of-approximation-guarantee} (as noted after the statement of the Theorem). But what is the smallest $r$ that ensures the
Frobenius norm residual error bounds $\FNorm{\matM - \matY\matY^\dagger \matM} \leq (1 + \varepsilon) \FNorm{\matM - \matM_k}$ and $\FNorm{\matM - \matQ\matX_{opt}} \leq (1 + \varepsilon) \FNorm{\matM - \matM_k}$ are satisfied with some fixed probability? To investigate, in Figure~\ref{fig:empiricalrnecessary} we plot the values of $r$ determined empirically to be sufficient to obtain $(1+\varepsilon)$ Frobenius norm residual errors relative to the optimal rank-$k$ approximation; we fix the failure probability $\delta=1/2$ and vary $\varepsilon.$ Specifically, the $r$ plotted for each $k$ is the smallest number of samples for which $\FNorm{\matM - \matY \matY^\dagger \matM} \leq (1 + \varepsilon) \FNorm{\matM - \matM_k}$ (or $\FNorm{\matM - \matQ \matX_{opt}} \leq (1 + \varepsilon) \FNorm{\matM - \matM_k}$) in at least 5 out of 10 trials.


It is clear that, for fixed $k$ and $\varepsilon,$ the number of samples $r$ required to form a non-rank-restricted approximation to $\matM$ with $(1+\varepsilon)$ relative residual error is smaller than the $r$ required to form a rank-restricted approximation with $(1+\varepsilon)$ relative residual error. Note that for small values of $k$, the $r$ necessary for relative residual error to be achieved is actually smaller than $k$ for all three datasets. This is a reflection of the fact that when $k_1 < k_2$ are small, the ratio $\FNorm{\matM - \matM_{k_2}}/\FNorm{\matM - \matM_{k_1}}$ is very close to one. Outside of the initial flat regions, the empirically determined value of $r$ seems to grow linearly with $k$;
this matches with the observation of Woolfe et al. that taking $r=k+8$ suffices to consistently form accurate low-rank approximations using the SRFT scheme, which is very similar to the SRHT scheme~\cite{WLRT07}. We also note that this matches with Theorem $\ref{thm:quality-of-approximation-guarantee}$ which predicts that the necessary $r$ grows at most linearly with $k$ with a slope like $\ln n.$

Finally, Theorem~\ref{thm:quality-of-approximation-guarantee} does \emph{not} guarantee that $(1+\varepsilon)$ spectral norm relative residual errors can be achieved. Instead, it provides bounds on the spectral norm residual errors achieved in terms of $\TNorm{\matM - \matM_k}$ and $\FNorm{\matM - \matM_k}$ that are guaranteed to hold when $r$ is sufficiently large. In Figure~\ref{fig:predictedspecerrvsactual} we compare the spectral norm residual error guarantees of Theorem~\ref{thm:quality-of-approximation-guarantee} to what is achieved in practice. To do so, we take the optimistic viewpoint that the constants in Theorem~\ref{thm:quality-of-approximation-guarantee} can be optimized to unity. Under this view, if more columns than
$$r_2 = \varepsilon^{-1} [\sqrt{k} + \sqrt{\ln(n/\delta)}]^2 \ln(k/\delta)$$ are used to construct the SRHT approximations, then the spectral norm residual error is no larger than
\[
 b_2 = \left(1 + \sqrt{\frac{\ln(n/\delta) \ln(\rho/\delta)}{r}}\right) \cdot \TNorm{\matM - \matM_k} + \sqrt{\frac{\ln(\rho/\delta)}{r}} \cdot \FNorm{\matM - \matM_k},
\]
where $\rho$ is the rank of $\matM,$ with probability greater than $1-\delta.$ Our comparison consists of using $r_2$ samples to construct the SRHT approximations and then comparing the predicted upper bound on the spectral norm residual error, $b_2$, to the empirically observed spectral norm residual errors. Figure~\ref{fig:predictedspecerrvsactual} shows, for several values of $k$, the upper bound $b_2$ and the observed relative spectral norm residual errors, with precision parameter $\varepsilon = 1/2$ and failure parameter $\delta = 1/2.$ For each value of $k,$ the empirical spectral norm residual error plotted is the largest of the errors from among 10 trials of low-rank approximations. Note from Figure~\ref{fig:predictedspecerrvsactual} that with this choice of $r,$ the spectral norm residual errors of the rank-restricted and non-rank-restricted SRHT approximations are essentially the same.

\begin{figure}
 \centering
 \subfigure{\includegraphics[width=1.85in, keepaspectratio=true]{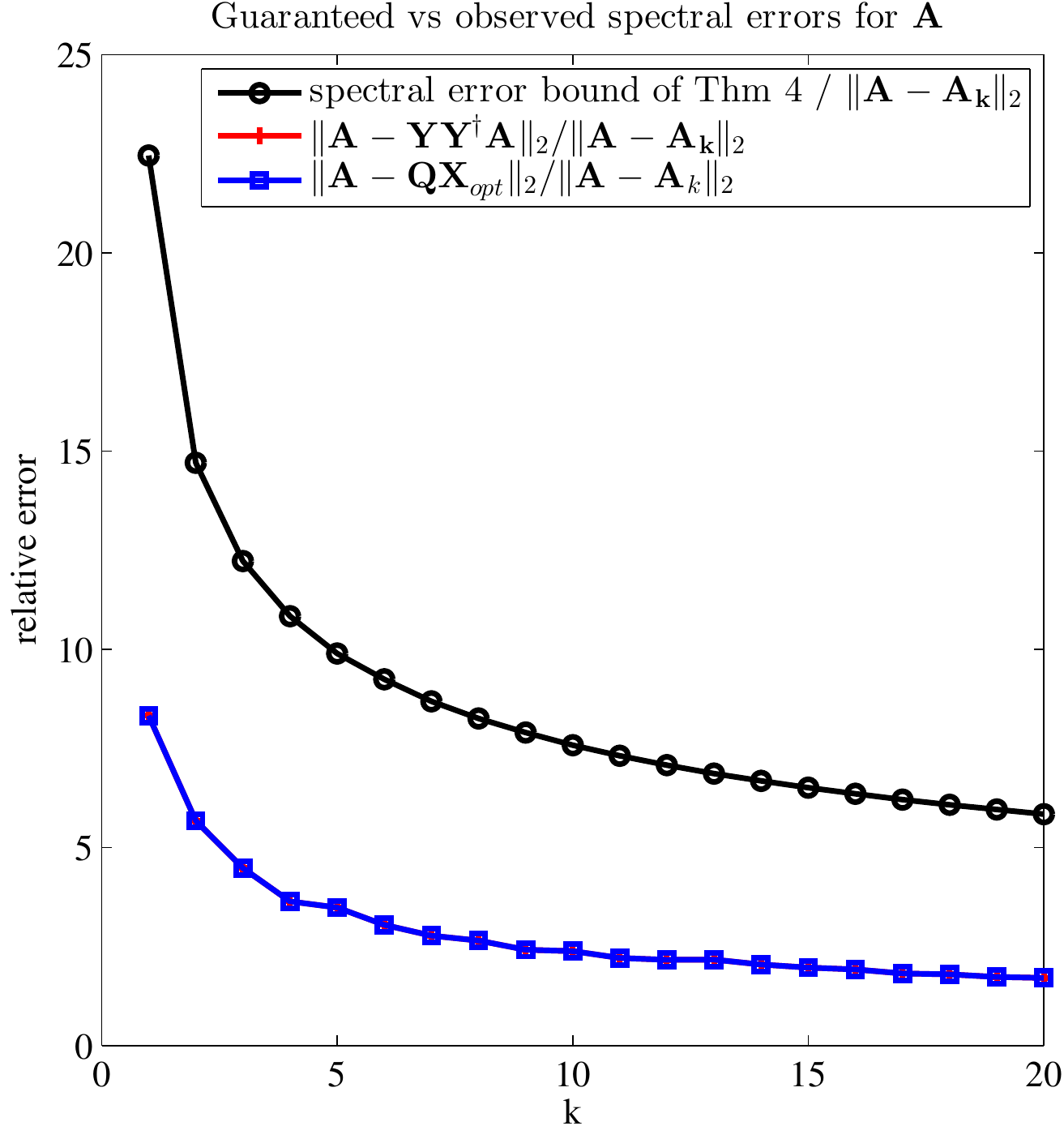}}%
 \subfigure{\includegraphics[width=1.85in, keepaspectratio=true]{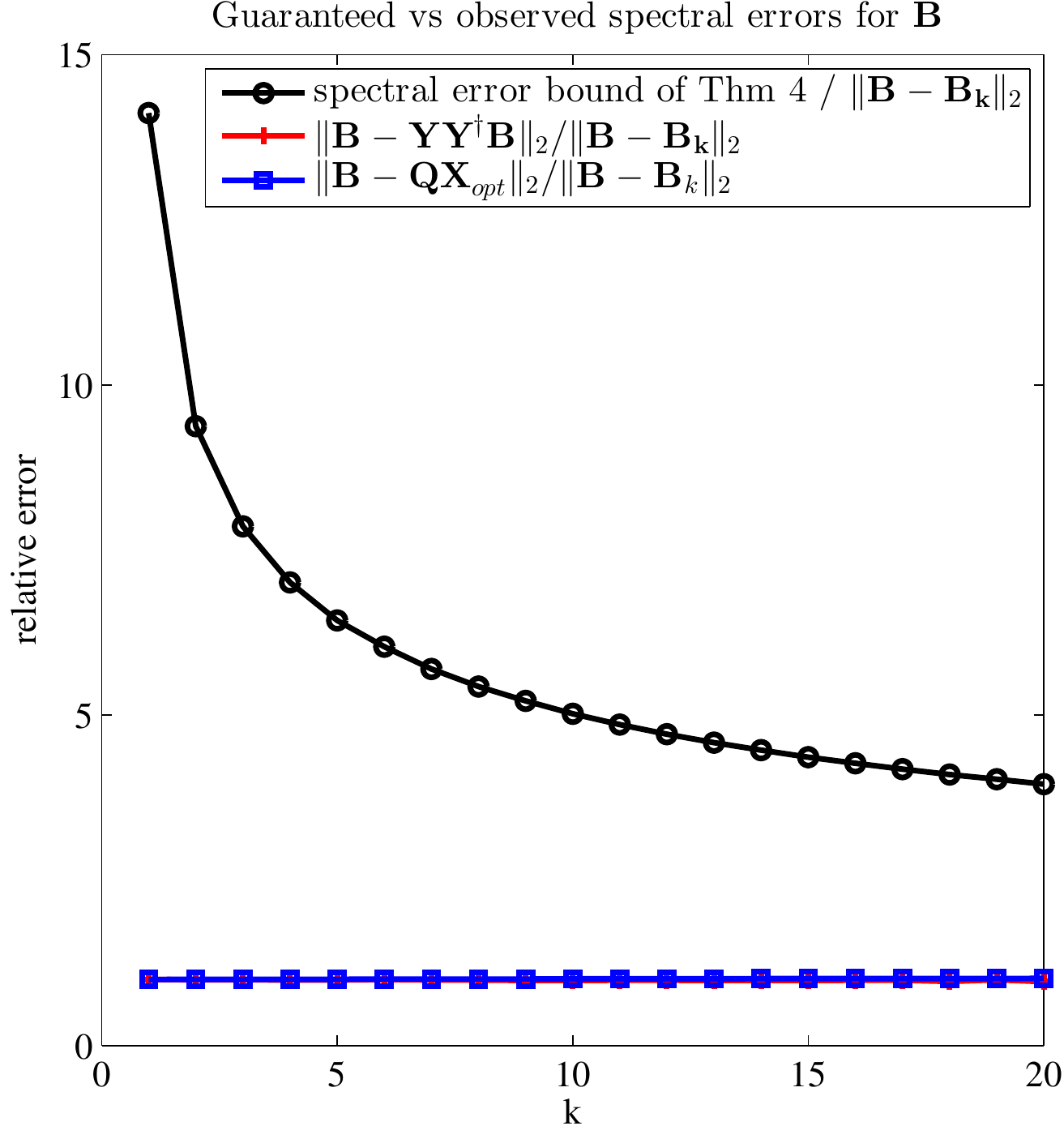}}%
 \subfigure{\includegraphics[width=1.85in, keepaspectratio=true]{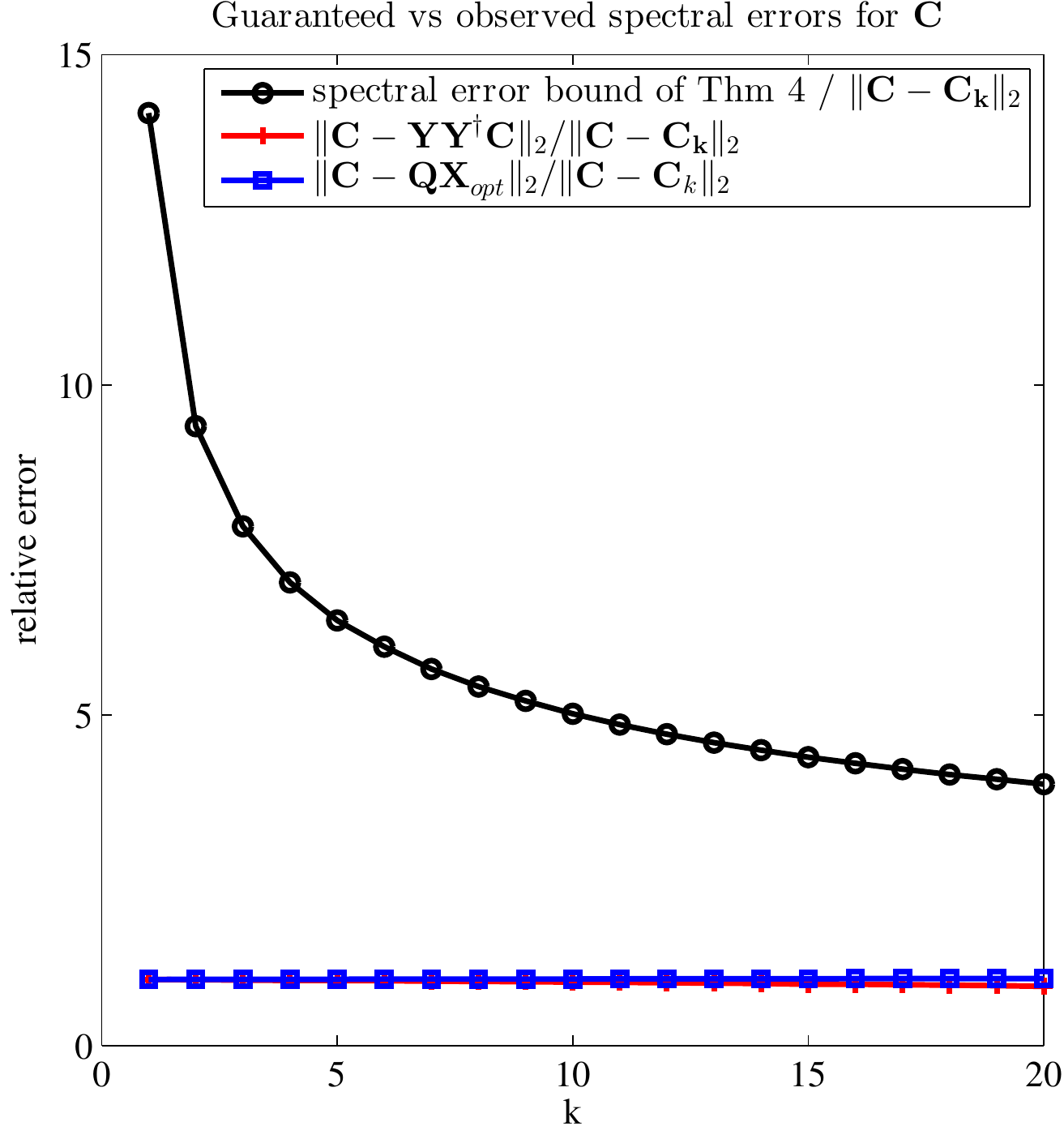}}
 \caption{The empirical spectral norm residual errors relative to those of the optimal rank-$k$ approximants ($\TNorm{\matM - \matY\matY^\dagger \matM}/\TNorm{\matM - \matM_k}$ and $\TNorm{\matM - \matQ\matX_{opt}}/\TNorm{\matM - \matM_k}$) plotted alongside the same ratio for the bound given in Theorem~\ref{thm:quality-of-approximation-guarantee}, when $r = \lceil 2[\sqrt{k} + \sqrt{\ln(2n)}]^2 \ln(2k) \rceil$ (for $\matM = \matA, \matB, \matC$).}
 \label{fig:predictedspecerrvsactual}
\end{figure}

 Judging from Figures~\ref{fig:empiricalrnecessary} and~\ref{fig:predictedspecerrvsactual}, even when we assume the constants present can be optimized away, the bounds given in Theorem~\ref{thm:quality-of-approximation-guarantee}
 are pessimistic: it seems that in fact approximations with Frobenius norm residual error within $(1+\varepsilon)$ of the error of the
 optimal rank-$k$ approximation can be obtained with $r$ linear in $k,$ and the spectral norm residual errors are smaller than the supplied upper bounds.
 Thus there is still room for improvement in our understanding of the SRHT low-rank approximation algorithm, but as explained in Section~\ref{sec:priorwork},
 Theorem~\ref{thm:quality-of-approximation-guarantee}, especially the spectral norm bounds, represents a significant improvement over prior efforts.

To bring perspective to this discussion, consider that even if one limits consideration to deterministic algorithms, the
known error bounds for the Gu-Eisenstat rank-revealing QR---a popular and widely used algorithm for low-rank approximation---are quite pessimistic and do not reflect the
excellent accuracy that is seen in practice~\cite{GE96}. Regardless, we do not advocate using these approximation schemes for applications in which highly accurate low-rank approximations are needed. Rather, Theorem~\ref{thm:quality-of-approximation-guarantee} and our numerical experiments suggest that they are appropriate in situations where one is willing to trade some accuracy for a gain in computational efficiency.

\subsection*{Acknowledgements}

We would like to thank Joel Tropp and Mark Tygert for the initial suggestion that we attempt
to sharpen the analysis of the SHRT low-rank approximation algorithm and for fruitful conversations
on our approach. We are also grateful to an anonymous reviewer for pointing out the value in interpreting Lemma~\ref{lem:mm} as a relative error bound
and to Malik Magdon-Ismail for providing the proof of Lemma 5.3.

Christos Boutsidis acknowledges the support from XDATA program of the Defense Advanced Research Projects Agency (DARPA), administered through Air Force Research Laboratory contract FA8750-12-C-0323. Alex Gittens was supported by ONR awards N00014-08-1-0883 and N00014-11-1002, AFOSR award FA9550-09-1-0643, DARPA award N66001-08-1-2065, and a Sloan Research Fellowship rewarded to Joel Tropp.